\documentclass{article}
\usepackage[T1]{fontenc}
\usepackage{amsthm}
\usepackage{amsmath}
\usepackage{amssymb}
\usepackage{tikz}
\usetikzlibrary{arrows.meta,positioning,matrix,fit,backgrounds,calc}
\usepackage{verbatim}
\usepackage{array}
\usepackage{enumerate}

\usepackage{microtype}
\usepackage{needspace}
\usepackage{etoolbox}
\usepackage[ruled,vlined]{algorithm2e}
\usepackage{authblk}

\newcommand*{\set}[1]{\{#1\}}
\newcommand*{\tsr}[2]{\mathcal{K}(#1,#2)}
\newcommand*{\dem}[2]{\delta(#1,#2)}
\newcommand*{\drem}[2]{\delta^{\mathrm{rem}}(#1,#2)}

\theoremstyle{plain}
\newtheorem{theorem}{Theorem}[section]
\newtheorem{proposition}[theorem]{Proposition}
\newtheorem{corollary}[theorem]{Corollary}
\newtheorem{lemma}[theorem]{Lemma}

\theoremstyle{definition}
\newtheorem{definition}[theorem]{Definition}
\newtheorem{example}[theorem]{Example}

\BeforeBeginEnvironment{theorem}{\Needspace{8\baselineskip}}
\BeforeBeginEnvironment{proposition}{\Needspace{6\baselineskip}}
\BeforeBeginEnvironment{corollary}{\Needspace{6\baselineskip}}
\BeforeBeginEnvironment{lemma}{\Needspace{6\baselineskip}}

\usepackage[letterpaper,top=2cm,bottom=2cm,left=3cm,right=3cm,marginparwidth=2cm]{geometry}

\usepackage{graphicx}
\usepackage[colorlinks=true, allcolors=blue]{hyperref}
\usepackage{cleveref}

\tikzset{%
    add/.style args={#1 and #2}{
        to path={%
 ($(\tikztostart)!-#1!(\tikztotarget)$)--($(\tikztotarget)!-#2!(\tikztostart)$)%
  \tikztonodes},add/.default={.2 and .2}}
} 

\title{\bf Proportional Representation in Temporal Voting with \\ Ranked Preferences}

\author[1]{Noam Hazon}
\author[1]{Leora Schmerler}
\author[2]{Nicholas Teh}
\affil[1]{Department of Data Science and AI, Ariel University}
\affil[2]{Department of Computer Science, University of Oxford}

\begin{document}
\date{}
\maketitle

\begin{abstract}
We study proportional representation in temporal voting, where one candidate is selected in each round. While prior work has focused on approval ballots, we consider ranked preferences, which may change over time. A natural approach treats each voter's top candidates as approved, but the right cutoff may differ across voters and rounds. We therefore require proportionality to hold for every admissible choice of cutoffs, whether fixed and common, common but varying across rounds, or set individually for each voter in each round. Combining these interpretations with temporal versions of justified representation (JR), proportional JR (PJR), extended JR (EJR), and proportionality for solid coalitions (PSC) gives us a hierarchy of axioms. We ask which of these axioms can be guaranteed, and with how much knowledge of the future. Unlike with approval ballots, no version of EJR can be guaranteed, and for the other axioms, flexibility in the cutoffs comes at a price. With a fixed common cutoff, JR, PJR, and PSC can be guaranteed, but only by rules that see all preferences in advance. Once the cutoff may vary across rounds, even such rules cannot guarantee JR or PSC for groups that agree in only some rounds. For groups that agree in every round, however, knowing only the number of rounds suffices for PJR in polynomial time, and PSC needs no knowledge of the future at all. Under individual cutoffs, no version of JR or PJR can be guaranteed, yet a rule as simple as serial dictatorship achieves PJR up to an additive loss that no rule can improve on, however much it knows. Natural preference restrictions restore exact guarantees. Finally, we show that checking our axioms is often coNP-complete; but perhaps surprisingly, a stronger axiom can be easier to check.
\end{abstract}

\section{Introduction}
Consider the alignment of a language model with human feedback. For each prompt, the model generates several responses, and annotators rank these responses from best to worst~\cite{ouyang2022training}. Suppose that one response is selected for each prompt, for example as a target for further training. Annotators often disagree, and how their rankings should be combined is a question of social choice~\cite{conitzer2024social}. If half of the annotators consistently rank the most concise response first, a third prefer the most detailed one, and the remaining sixth prefer the most cautious one, then selecting the response ranked first by the most annotators for every prompt gives the same annotators their first choice every time. A model trained on the selected responses would then reflect only the majority, whereas pluralistic alignment asks for models that serve people with diverse values and perspectives~\cite{sorensen2024roadmap}. A proportional sequence of decisions should instead give each group influence in proportion to its size. The challenge is that preferences may change from one decision to the next: the annotators who agree on one prompt need not agree on the next.

For approval ballots, proportional justified representation (PJR) makes this idea precise~\cite{bulteau2021justified,chandak2024proportional}. A group agrees in a round if its members approve a common candidate, and its entitlement grows with both its size and the number of rounds in which it agrees. For example, a group containing a quarter of the electorate that agrees in twelve rounds must have a winner approved by at least one of its members in at least three rounds. These winners may occur in any round, and different members may approve different winners.

With ranked preferences, it is not clear which candidates a voter finds acceptable. In our example, an annotator may find her second-ranked response nearly as good as her first for one prompt, whereas for another prompt, or for another annotator, only the first-ranked response may be acceptable. Counting only first choices ignores possible compromises. Counting the first several candidates requires a cutoff, and a single cutoff may not suit every voter or every round. We therefore ask whether one sequence of decisions can be proportional simultaneously for every admissible choice of cutoffs. We consider three interpretations: a cutoff common to all voters and rounds (\emph{fixed-rank}), a common cutoff that may change between rounds (\emph{rank}), and a separate cutoff for each voter in each round (\emph{all-rank}). We also adapt \emph{proportionality for solid coalitions} (PSC), which protects groups whose members rank the same set of candidates above all others.

The choice of cutoffs is only one part of the definition. We distinguish groups that agree in every round from groups that agree in only some rounds, and call an axiom \emph{weak} if it protects only the former. We also consider \emph{justified representation} (JR), which requires just one represented round for a sufficiently large group, and \emph{extended justified representation} (EJR), which requires one member of a group to receive its entire proportional entitlement. We denote the weak versions by wJR, wPJR, wEJR, and wPSC. These axioms carry over the forms of group protection familiar from multi-winner voting~\cite{aziz2017justified,sanchez2017proportional,dummett1984voting}, while allowing both the preferences and the rounds of agreement to change over time.

A further distinction concerns information. An \emph{offline} rule sees all preferences before making any decision. A \emph{semi-online} rule knows the number of rounds, which we call the \emph{horizon}, but learns the preferences one round at a time. An \emph{online} rule does not even know when the process will end. Our main question is which combinations of these proportionality requirements can be guaranteed, and how much information a rule needs to do so.

\subsection{Our Results}

\paragraph{Axioms and their relationships (Section~\ref{sec:ranked-axioms}).}
We define fixed-rank, rank, and all-rank versions of the approval-based axioms,
together with fixed-rank and rank versions of PSC and their weak counterparts.
Their implications reflect three choices: how the cutoffs may vary, in
which rounds a group must agree, and how representation is measured
(Figure~\ref{fig:implication-full}). With three candidates, even
fixed-rank-wEJR can be unattainable at every horizon of at least two rounds.
With two candidates and complete rankings, however, fixed-rank-wEJR and rank-wEJR
coincide with fixed-rank-wPJR and rank-wPJR, respectively; individual
cutoffs still make all-rank-wEJR unattainable in general.
We also characterize the all-rank JR and PJR axioms directly through the
rankings. For the weak axioms, the relevant count is the number of rounds
in which no candidate is unanimously preferred by a group to the winner.
For the non-weak axioms, first-choice representation provides the additional
count needed. The weak characterization connects all-rank-wJR to local
stability.

\paragraph{Existence and computability (Section~\ref{sec:attainability}).}
For groups that agree in some rounds, fixed-rank-PJR is always attainable
offline. However, no semi-online rule guarantees even fixed-rank-JR or
fixed-rank-PSC, and rank-JR and rank-PSC can be unattainable even offline.
For groups that agree in every round, a polynomial-time semi-online rule
satisfies rank-wPJR, and online rules satisfy the weak PSC axioms. These
rules also admit guarantees with the stronger Droop size condition.
Knowledge of the horizon is necessary for the approval-based guarantees,
since no online rule always satisfies fixed-rank-wJR.
Table~\ref{tab:attainability} gives the least sufficient information and the
running times of our rules. Individual cutoffs can make even all-rank-wJR
unattainable. Nevertheless, serial dictatorship satisfies all-rank-PJR up
to an additive loss in the number of represented rounds, and this loss is
optimal for every electorate size and horizon. The matching lower bound
holds even for offline rules and the weak axiom.

\paragraph{Stronger guarantees under preference restrictions
(Section~\ref{sec:preference-structure}).}
When few candidates are ranked first in each round, a polynomial-time
semi-online budget rule achieves the smallest additive loss that depends
only on the number of distinct first choices. In particular, three first
choices suffice for exact all-rank-PJR, whereas four first choices can make
it unattainable at every horizon of at least two rounds. We also characterize
when all-rank-JR is attainable, and give a polynomial-time online rule with
a guarantee on every contiguous time interval. Its loss is zero for two
first choices and an optimal one round for three. With two candidates and
complete rankings, the same rule gives fixed-rank-wEJR and rank-wEJR on every interval,
completing the candidate--horizon existence classification for the weak EJR
axioms.

Single-peaked and single-crossing preferences permit any number of first
choices, but protect only groups that agree in every round. On these
domains, a polynomial-time semi-online rule satisfies all-rank-wPJR.
A common single-crossing order also permits polynomial-time online
rank-wPSC. Even all-rank-wJR, however, cannot always be satisfied online on
a static Euclidean profile with fifteen voters and fifteen candidates.
For all-rank-wPJR on a common single-crossing order, the optimal online
additive loss has logarithmic order in the number of voters.

\paragraph{Verification and truncated ballots
(Sections~\ref{sec:verification} and~\ref{sec:truncated-ballots}).}
We determine the candidate thresholds for verifying every defined weak
JR, PJR, and PSC axiom (Table~\ref{tab:verification-complete}).
Fixed-rank-wPSC is always polynomial-time verifiable; fixed-rank-wJR becomes
coNP-complete with four candidates, and every other entry becomes
coNP-complete with three. On three-candidate Euclidean profiles, verification
of weak JR and PJR is polynomial for the fixed-rank versions, coNP-complete
for the rank versions, and polynomial for the all-rank versions.
Thus, strengthening an axiom can make verification easier. A common
single-crossing voter order permits polynomial-time verification of both
weak and non-weak all-rank JR and PJR, whereas changing orders makes all four
problems coNP-complete even with three candidates and two first choices.

We then allow voters to rank only acceptable candidates. The Greedy Cohesive
Rule (GCR), Solid Coalition Refinement (SCR), first-choice and ordered budget
rules, and a modified Method of Equal Shares (MES) retain their guarantees.
In particular, restricting the groups considered by SCR according to list
length extends the polynomial-time common-order rule to truncated ballots;
rank-wPSC is also polynomial-time verifiable there. Verification of the
weak all-rank axioms can nevertheless become coNP-complete. We classify
cases by ballot lengths and omitted candidates, and on single-peaked
profiles give an algorithm whose running time is exponential only in the
number of distinct voter types that omit at least two candidates in some
round, or at least one candidate when there are only two candidates.

\subsection{Related Work}
\paragraph{Temporal voting with approval ballots.}
Elkind, Obraztsova, and Teh~\cite{elkind2024temporal} propose a unified framework for temporal multi-winner voting and survey the models studied in this area. Bulteau et al.~\cite{bulteau2021justified} adapt justified representation axioms to repeated collective decisions. Chandak, Goel, and Peters~\cite{chandak2024proportional} distinguish groups that agree in every round from groups that agree in only some rounds, and analyze voting rules under different information assumptions. Elkind et al.~\cite{elkind2025verifying} study the verification of these axioms and give a two-stage Greedy Cohesive Rule that achieves extended justified representation, which requires one group member to approve as many winners as the group is owed. Phillips et al.~\cite{phillips2026strengthening} study further strengthenings and adapt this greedy approach to full justified representation.
Teh~\cite{teh2026FJRtemporal} studies the efficient computation and verification of full justified representation in temporal voting.
Our fixed-rank existence proof builds on the two-stage rule, but must meet the fixed-rank requirements for all $m$ cutoffs simultaneously rather than for a single approval profile.

\paragraph{Proportionality in multi-winner voting.}
Our approval-based axioms adapt proportionality axioms for committee elections; see Lackner and Skowron~\cite{lackner2023multi} for an overview. Aziz et al.~\cite{aziz2017justified} introduced JR and EJR, and S{\'a}nchez-Fern{\'a}ndez et al.~\cite{sanchez2017proportional} introduced PJR, which is implied by EJR and implies JR. Stronger axioms include full justified representation (FJR)~\cite{peters2021proportional}, EJR+~\cite{brill2023robust}, and FJR+~\cite{teh2026strengthening}.\footnote{Each of these axioms implies EJR in multi-winner voting~\cite{peters2021proportional,brill2023robust,teh2026strengthening}, and EJR+ and FJR also imply EJR in temporal voting~\cite{phillips2026strengthening}. If a temporal axiom $Y$ implies EJR, then fixed-rank-$Y$, rank-$Y$, and all-rank-$Y$, defined as in Section~\ref{sec:approval-axioms}, each imply fixed-rank-EJR, so none of them can always be satisfied (Theorem~\ref{thm:ejr}). For multi-winner elections, Brill and Peters~\cite{brill2023robust} similarly show that rank-EJR+ can be unsatisfiable, and propose rank-PJR+ instead.}

\paragraph{Ranked preferences.}
Dummett~\cite{dummett1984voting} introduced proportionality for solid coalitions in multi-winner elections. Brill and Peters~\cite{brill2023robust} transfer approval-based proportionality axioms to rankings by requiring them to hold at every cutoff common to all voters; applied in every round with the same cutoff, this gives our fixed-rank axioms. We additionally allow the cutoff to vary across rounds (rank) and across voters (all-rank). Aziz et al.~\cite{aziz2024committee} study the compatibility of proportionality with committee monotonicity, which requires that increasing the committee size does not remove a previously selected candidate. Their Solid Coalition Refinement rule is the starting point for our online rule for solid coalitions. Unlike a multi-winner election, our model allows the same candidate to win repeatedly and allows rankings to change; consequently, agreement and representation must be tracked across rounds. Our connection between all-rank representation and Pareto efficiency is related to the multi-winner local stability condition of Aziz et al.~\cite{aziz2017condorcet}.

\paragraph{Proportionality over time.}
Lackner~\cite{lackner2020perpetual} introduced perpetual voting rules, which take earlier decisions into account in order to treat voters fairly over time. Lackner and Maly~\cite{lackner2023proportional} analyze perpetual voting rules through proportionality and through bounds on how long voters can remain unrepresented. Kozachinskiy, Shen, and Steifer~\cite{kozachinskiy2025optimal} bound the number of rounds in which a voter does not approve the winner, provided that small groups of voters rarely fail to agree, and show that the exponents in their bounds cannot be improved. Teh~\cite{teh2026price} bounds the worst-case loss in utilitarian welfare caused by requiring JR, PJR, EJR, or EJR+, and shows that maximizing welfare subject to these axioms is computationally hard. Borodin and Lueger~\cite{borodin2026online} show that serial dictatorship gives an additive PJR guarantee for approval ballots as the horizon grows. Our serial dictatorship result determines the optimal finite-horizon loss for all-rank-PJR. Our first-choice results instead bound the loss independently of the number of voters, and our interval guarantee applies to every group on every contiguous interval, not only to representation accumulated from the first round. The lower bound for three first choices on time intervals uses a classical scheduling example~\cite{holte1992pinwheel}; our positive guarantee allows arbitrary changes in the voters' rankings.

\paragraph{Welfare and other objectives over time.}
Elkind, Neoh, and Teh~\cite{elkind2024temporalelections} study utilitarian and egalitarian welfare in temporal voting with approval ballots, including the complexity of maximizing each objective and its compatibility with strategyproofness and proportionality. Elkind, Neoh, and Teh~\cite{elkind2025nimby} consider public chores, which benefit society but impose costs on the affected voters, and show that proportionality and equitability are either not well defined or costly in terms of welfare in that setting. Zech et al.~\cite{zech2024multiwinner} study two-stage committee elections and seek a winning committee in the second stage that overlaps as much as possible with the first-stage committee. Temporal fairness has also been studied beyond voting, including the assignment of projects to time slots~\cite{elkind2022fairness} and the allocation of indivisible items that arrive over time~\cite{elkind2025temporal,choi2026temporal}.

\section{Preliminaries}
\label{sec:preliminaries}
For a positive integer $k$, let $[k]=\set{1,\dots,k}$.
There is a set $V=[n]$ of \emph{voters}, a set $C=\set{c_1,\dots,c_m}$ of \emph{candidates}, and a set $T=[|T|]$ of \emph{rounds}; we call $|T|$ the \emph{horizon}. We assume that $n,m,|T|\geq1$. A \emph{group} is a nonempty set $S\subseteq V$.

In each round $j\in T$, each voter $i$ reports a strict ranking $\succ_i^j$ of $C$; we write $c\succ_i^j c'$ if voter $i$ prefers $c$ to $c'$ in round $j$, and omit the round index when it is clear from the context. For $A,B\subseteq C$, we write $A\succ_i^j B$ if $a\succ_i^j b$ for all $a\in A$ and $b\in B$. The \emph{round-$j$ profile} is $P^j=(\succ_1^j,\dots,\succ_n^j)$, and the \emph{profile} is $P=(P^1,\dots,P^{|T|})$. A profile is \emph{static} if $P^1=\cdots=P^{|T|}$. Let $\mathrm{rank}_i^j(c)=\left|\set{c'\in C:c'\succ_i^j c}\right|+1$ be the rank of candidate $c$ in the ranking of voter $i$ in round $j$. The candidate of rank $1$ is voter $i$'s \emph{first choice} in round $j$.

\paragraph{Preference domains.}
A round-$j$ profile is \emph{single-peaked} with respect to a linear order of $C$, called an \emph{axis}, if, for every voter, candidates are ranked lower as one moves away from her first choice (her \emph{peak}) along the axis in either direction. It is \emph{single-crossing} with respect to a linear order of $V$ if, for every pair of candidates $a,b$, the voters who prefer $a$ to $b$ form a prefix or a suffix of that order. It is \emph{one-dimensional Euclidean} if voters and candidates can be placed on the real line so that every voter ranks the candidates by increasing distance from her position, with no ties in distance. Such a profile is single-peaked with respect to the order of the candidates on the line and single-crossing with respect to the order of the voters.

\paragraph{Voting rules.}
A \emph{temporal voting rule} maps $(V,T,C,P)$ to a sequence $W=(w^1,\dots,w^{|T|})$ with $w^j\in C$ for each $j\in T$; we call $w^j$ the \emph{winner} of round $j$. An \emph{offline} rule may use the whole profile. A rule is \emph{semi-online} if its choice of $w^j$ depends only on $P^1,\dots,P^j$ and $|T|$, and \emph{online} if, in addition, it does not depend on $|T|$.

\subsection{Axioms for Approval Ballots}
We recall the axioms of Bulteau et al.~\cite{bulteau2021justified} and Chandak et al.~\cite{chandak2024proportional} for approval ballots. In this setting, each voter $i$ approves a set $A_i^j\subseteq C$ in each round $j$. The round-$j$ approval profile is $A^j=(A_1^j,\dots,A_n^j)$, and the approval profile is $A=(A^1,\dots,A^{|T|})$. A group $S$ \emph{agrees} in round $j$ if its members approve a common candidate, that is, $\bigcap_{i\in S}A_i^j\neq\varnothing$, and $S$ is \emph{represented} in round $j$ if some member approves the winner, that is, $w^j\in\bigcup_{i\in S}A_i^j$. The axioms require that a sufficiently large group that agrees in many rounds be represented; they differ in how representation is counted.
\begin{definition}
A sequence $W$ satisfies \emph{justified representation (JR)} (respectively, \emph{proportional justified representation (PJR)}, \emph{extended justified representation (EJR)}) if, for every positive integer $\ell$ and every group $S$ that agrees in $k>0$ rounds and satisfies $|S|\geq\ell\cdot\frac{n}{k}$,
\begin{itemize}
 \item (JR) $S$ is represented in at least one round;
 \item (PJR) $S$ is represented in at least $\ell$ rounds;
 \item (EJR) some voter $i\in S$ approves the winner in at least $\ell$ rounds, that is, $|\set{j\in T:w^j\in A_i^j}|\geq\ell$.
\end{itemize}
\end{definition}
The largest $\ell$ satisfying $|S|\geq\ell n/k$ is $\lfloor k|S|/n\rfloor$. Thus, PJR requires every group that agrees in $k>0$ rounds to be represented in at least $\lfloor k|S|/n\rfloor$ rounds. For JR, it suffices to consider $\ell=1$, since the size condition $|S|\geq\ell n/k$ is weakest then; thus, JR is the special case of PJR with $\ell=1$. The \emph{weak} versions wJR, wPJR, and wEJR apply the same requirements only to groups that agree in every round, that is, with $k=|T|$. A rule satisfies an axiom if its output satisfies the axiom for every input.

We also consider a stronger size condition for the weak proportionality axioms.
The \emph{Hare quota} used above gives a group $S$ an entitlement of
$\lfloor |T||S|/n\rfloor$ rounds. Under the \emph{Droop quota}, a group
with $|S|>\ell n/(|T|+1)$ is owed $\ell$ rounds; its largest integer
entitlement is therefore $\lceil (|T|+1)|S|/n\rceil-1$.
This is at least the Hare entitlement. Unless stated otherwise, our axioms
use the Hare quota; we specify the Droop size condition when a rule satisfies
this stronger guarantee.

\section{Proportionality Axioms and Their Relationships}
\label{sec:ranked-axioms}
With ranked preferences, both agreement and representation depend on which positions in a ranking a group considers acceptable. We first require the approval-based axioms to hold simultaneously for different choices of cutoffs (Section~\ref{sec:approval-axioms}). We then adapt proportionality for solid coalitions, which protects groups that rank a common set above all other candidates (Section~\ref{sec:psc-axioms}). These approaches inherit different forms of group protection from multi-winner voting. Here, however, a candidate may win repeatedly and preferences may change, so proportionality concerns the rounds in which a group agrees and is represented.

\subsection{Approval-based Axioms}
\label{sec:approval-axioms}
Following Brill and Peters~\cite{brill2023robust}, we treat a voter's first $r$ candidates as approved. For $r\in[m]$, let $A_i^j(r)=\set{c\in C:\mathrm{rank}_i^j(c)\leq r}$ and $N_c^j(r)=\set{i\in V:c\in A_i^j(r)}$.
Thus, $N_c^j(r)$ is the set of voters who rank candidate $c$ among their first $r$ candidates in round $j$. The approval sets are nested: $A_i^j(r)\subseteq A_i^j(r')$ whenever $r\leq r'$. In particular, every approval set $A_i^j(r)$ contains voter $i$'s first choice.

We consider three ways of choosing the cutoffs.
\begin{itemize}
 \item \emph{Fixed-rank:} a single cutoff $r\in[m]$ for all voters and rounds, giving the approval profile $A(r)=(A^1(r),\dots,A^{|T|}(r))$, where $A^j(r)=(A_1^j(r),\dots,A_n^j(r))$.
 \item \emph{Rank:} a \emph{rank vector} $\mathbf r=(r^1,\dots,r^{|T|})\in[m]^{|T|}$ with one cutoff per round, giving $A(\mathbf r)=(A^1(r^1),\dots,A^{|T|}(r^{|T|}))$.
 \item \emph{All-rank:} a \emph{rank matrix} $R\in[m]^{|T|\times n}$ with a cutoff $r_i^j$ for each voter $i$ in each round $j$, giving $A(R)=(A^1(R^1),\dots,A^{|T|}(R^{|T|}))$, where $A^j(R^j)=(A_1^j(r_1^j),\dots,A_n^j(r_n^j))$.
\end{itemize}

\begin{definition}
Let $X\in\set{\text{wJR, JR, wPJR, PJR, wEJR, EJR}}$. A sequence $W$ satisfies \emph{fixed-rank-$X$} if it satisfies $X$ with respect to $A(r)$ for every $r\in[m]$; it satisfies \emph{rank-$X$} if it satisfies $X$ with respect to $A(\mathbf r)$ for every $\mathbf r\in[m]^{|T|}$; and it satisfies \emph{all-rank-$X$} if it satisfies $X$ with respect to $A(R)$ for every $R\in[m]^{|T|\times n}$.
\end{definition}
A single sequence must satisfy the requirement for every choice of cutoffs; it cannot be chosen separately for each choice. When the cutoffs are clear from the context, we say that a group agrees or is represented in a round if it does so in the corresponding approval profile. For example, rank-wPJR can be stated directly as follows.

\begin{definition}[rank-wPJR]
A sequence $W$ satisfies \emph{rank-wPJR} if, for every rank vector $\mathbf r\in[m]^{|T|}$, every positive integer $\ell$, and every group $S$ with $|S|\geq\ell\cdot\frac{n}{|T|}$ such that $\bigcap_{i\in S}A_i^j(r^j)\neq\varnothing$ for all $j\in T$, there are at least $\ell$ rounds $j\in T$ with $w^j\in\bigcup_{i\in S}A_i^j(r^j)$.
\end{definition}

\subsection{Axioms for Solid Coalitions}
\label{sec:psc-axioms}
The most common formalization of proportionality for rankings in multi-winner elections is \emph{proportionality for solid coalitions (PSC)}~\cite{dummett1984voting}. A group $S$ forms a \emph{solid coalition} over a nonempty set $C'\subseteq C$ in round $j$ if $C'\succ_i^j C\setminus C'$ for all $i\in S$; that is, all members of $S$ rank the candidates in $C'$ above all other candidates, although possibly in different orders. PSC protects such a group by requiring winners from $C'$.

A natural temporal version of PSC would require that every group $S$ with $|S|\geq\ell n/|T|$ that forms a solid coalition over some set $C^j$ in every round $j$ receives a winner from $C^j$ in at least $\ell$ rounds; the set $C^j$ may change from round to round. Writing $r^j=|C^j|$ relates this condition to the approval-based axioms. A group $S$ forms a solid coalition over a set of size $r^j$ in round $j$ exactly when all its members have the same approval set at cutoff $r^j$, that is, $|\bigcap_{i\in S}A_i^j(r^j)|=r^j$.
In this case, we say that $S$ is \emph{solid} in round $j$ under $\mathbf r$. The weak axiom below formalizes the preceding requirement; the other version also protects groups that are solid in only some rounds.

\begin{definition}[rank-wPSC and rank-PSC]
A sequence $W$ satisfies \emph{rank-PSC} if, for every rank vector $\mathbf r\in[m]^{|T|}$, every positive integer $\ell$, and every group $S$ that is solid under $\mathbf r$ in $k>0$ rounds and satisfies $|S|\geq\ell\cdot\frac{n}{k}$, there are at least $\ell$ rounds $j\in T$ with $w^j\in\bigcup_{i\in S}A_i^j(r^j)$. It satisfies \emph{rank-wPSC} if this holds whenever $k=|T|$.
\end{definition}
The \emph{fixed-rank} versions, fixed-rank-PSC and fixed-rank-wPSC, restrict the requirement to constant rank vectors. In rank-PSC and fixed-rank-PSC, representation is counted over all rounds, not only those in which the group is solid. In a round in which $S$ is solid, it is represented exactly when the winner belongs to the common set $\bigcap_{i\in S}A_i^j(r^j)$. We do not consider all-rank versions of PSC: the members of a solid group have the same approval set in a round, which is possible only if they use the same cutoff.

\subsection{Relations among the Axioms}
A solid group agrees, and both families measure representation in the same way. Hence, rank-wPJR implies rank-wPSC and rank-PJR implies rank-PSC, and the same holds for the fixed-rank versions. Every rank vector corresponds to a rank matrix, and every cutoff to a constant rank vector, so every all-rank axiom implies the corresponding rank axiom, which in turn implies the corresponding fixed-rank axiom. Within each approval-based family, EJR implies PJR, which implies JR; each axiom also implies its weak version. Figure~\ref{fig:implication-full} shows all the axioms together with these implications; its colors summarize the existence results of Section~\ref{sec:attainability} for unrestricted rankings. Allowing more cutoff choices strengthens an axiom, whereas restricting protection to groups that agree throughout the horizon weakens it. In Section~\ref{sec:attainability}, we examine which combinations can be satisfied and what information they require.

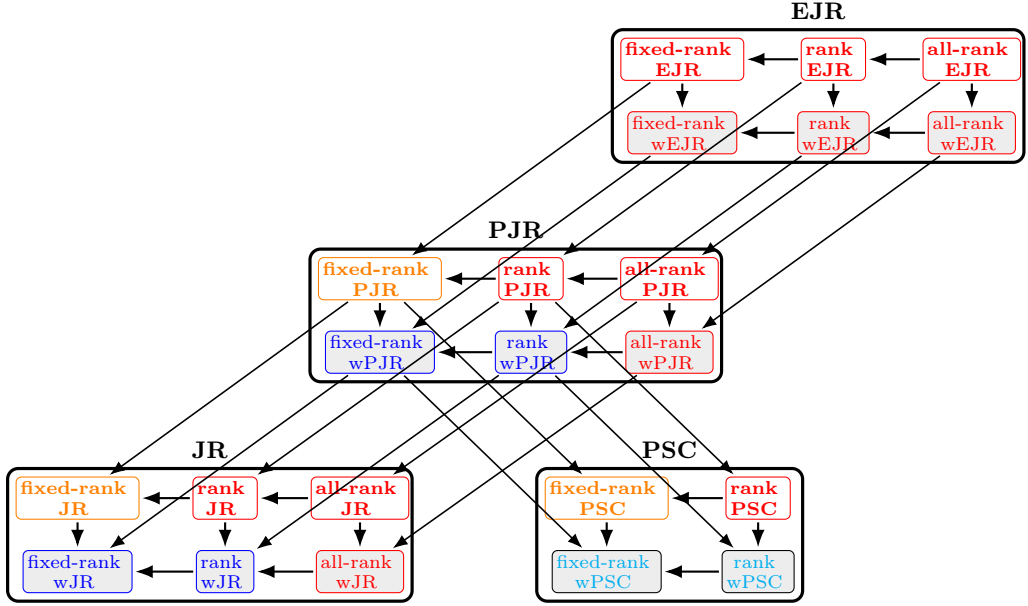
\begin{figure}[!htbp]
\begin{center}
\begin{tikzpicture}[
  >=Latex,
  axiom/.style={draw, rounded corners=2pt, inner sep=1.6pt, align=center, font=\scriptsize},
  strong/.style={axiom, font=\scriptsize\bfseries},
  weak/.style={axiom, fill=black!7},
  implLocal/.style={-Latex, thick, shorten <=1pt, shorten >=1pt},
  implBetween/.style={-Latex, semithick, shorten <=1pt, shorten >=1pt},
  familybox/.style={draw, very thick, rounded corners=4pt, inner sep=3pt}
]

\matrix (PSC) [
  matrix of nodes,
  anchor=south west,
  row sep=4mm,
  column sep=7mm
] at (7,0)
{
  |[strong, text=orange, draw=orange]| \shortstack{fixed-rank\\PSC}  & |[strong, text=red, draw=red]| \shortstack{rank\\PSC}  \\
  |[weak, text=cyan]|  \shortstack{fixed-rank\\wPSC} & |[weak, text=cyan]|  \shortstack{rank\\wPSC} \\
};
\node[familybox, fit=(PSC-1-1)(PSC-2-2), label={[font=\small\bfseries]above:PSC}] {};

\foreach \r in {1,2}{
  \draw[implLocal] (PSC-\r-2) -- (PSC-\r-1); 
}
\foreach \c in {1,2}{
  \draw[implLocal] (PSC-1-\c) -- (PSC-2-\c); 
}

\matrix (JR) [
  matrix of nodes,
  anchor=south west,
  row sep=4mm,
  column sep=7mm
] at (0,0)
{
  |[strong, text=orange, draw=orange]| \shortstack{fixed-rank\\JR}  & |[strong, text=red, draw=red]| \shortstack{rank\\JR}  & |[strong, text=red, draw=red]| \shortstack{all-rank\\JR} \\
  |[weak, text=blue, draw=blue]|  \shortstack{fixed-rank\\wJR} & |[weak, text=blue, draw=blue]|  \shortstack{rank\\wJR} & |[weak, text=red, draw=red]|  \shortstack{all-rank\\wJR} \\
};
\node[familybox, fit=(JR-1-1)(JR-2-3), label={[font=\small\bfseries]above:JR}] {};

\foreach \r in {1,2}{
  \draw[implLocal] (JR-\r-3) -- (JR-\r-2); 
  \draw[implLocal] (JR-\r-2) -- (JR-\r-1); 
}
\foreach \c in {1,2,3}{
  \draw[implLocal] (JR-1-\c) -- (JR-2-\c); 
}

\matrix (PJR) [
  matrix of nodes,
  anchor=south west,
  row sep=4mm,
  column sep=7mm
] at (4,2.9)
{
  |[strong, text=orange, draw=orange]| \shortstack{fixed-rank\\PJR}  & |[strong, text=red, draw=red]| \shortstack{rank\\PJR}  & |[strong, text=red, draw=red]| \shortstack{all-rank\\PJR} \\
  |[weak, text=blue, draw=blue]|  \shortstack{fixed-rank\\wPJR} & |[weak, text=blue, draw=blue]|  \shortstack{rank\\wPJR} & |[weak, text=red, draw=red]|  \shortstack{all-rank\\wPJR} \\
};
\node[familybox, fit=(PJR-1-1)(PJR-2-3), label={[font=\small\bfseries]above:PJR}] {};

\foreach \r in {1,2}{
  \draw[implLocal] (PJR-\r-3) -- (PJR-\r-2); 
  \draw[implLocal] (PJR-\r-2) -- (PJR-\r-1); 
}
\foreach \c in {1,2,3}{
  \draw[implLocal] (PJR-1-\c) -- (PJR-2-\c); 
}

\matrix (EJR) [
  matrix of nodes,
  anchor=south west,
  row sep=4mm,
  column sep=7mm,
  nodes={text=red, draw=red}
] at (8,5.8)
{
  |[strong]| \shortstack{fixed-rank\\EJR}  & |[strong]| \shortstack{rank\\EJR}  & |[strong]| \shortstack{all-rank\\EJR} \\
  |[weak]|  \shortstack{fixed-rank\\wEJR} & |[weak]|  \shortstack{rank\\wEJR} & |[weak]|  \shortstack{all-rank\\wEJR} \\
};
\node[familybox, fit=(EJR-1-1)(EJR-2-3), label={[font=\small\bfseries]above:EJR}] {};

\foreach \r in {1,2}{
  \draw[implLocal] (EJR-\r-3) -- (EJR-\r-2);
  \draw[implLocal] (EJR-\r-2) -- (EJR-\r-1);
}
\foreach \c in {1,2,3}{
  \draw[implLocal] (EJR-1-\c) -- (EJR-2-\c);
}

\foreach \r in {1,2}{
  \foreach \c in {1,2,3}{
    \draw[implBetween] (EJR-\r-\c) -- (PJR-\r-\c);
    \draw[implBetween] (PJR-\r-\c) -- (JR-\r-\c);
  }
  \foreach \c in {1,2}{
    \draw[implBetween] (PJR-\r-\c) -- (PSC-\r-\c);
  }
}

\end{tikzpicture}
\end{center}
\caption{Implications among the axioms (an arrow points from a stronger
axiom to a weaker one; weak axioms are shaded) and our results for
unrestricted rankings from Section~\ref{sec:attainability}
(Table~\ref{tab:attainability}). Missing arrows do not assert
non-implications.
The text color of an axiom indicates the least information with which it
can always be satisfied: \textcolor{cyan}{online}, \textcolor{blue}{semi-online},
or \textcolor{orange}{offline}; \textcolor{red}{red} text means that it
cannot always be satisfied. The border color indicates an impossibility
result: \textcolor{blue}{not always satisfiable online},
\textcolor{orange}{not always satisfiable semi-online}, or
\textcolor{red}{not always satisfiable}. The colors concern information,
not running time. The two-candidate exceptions for weak EJR are given in
Corollary~\ref{cor:ejr-candidate-horizon}.}
\label{fig:implication-full}
\end{figure}

\subsection{Extended Justified Representation}
Unlike for approval ballots, where EJR can always be satisfied~\cite{elkind2025verifying}, no version of EJR can be guaranteed on unrestricted ranked profiles. Part~(i) of the following theorem adapts Example~7 of Brill and Peters~\cite{brill2023robust}. The number of candidates matters: with two candidates and complete rankings, fixed-rank-wEJR and rank-wEJR coincide with fixed-rank-wPJR and rank-wPJR, respectively (part~(ii)), whereas all-rank-wEJR can still be unattainable (part~(iii)).

\begin{theorem}
\label{thm:ejr}
\begin{enumerate}[(i)]
 \item For every horizon $|T|\geq2$, there is a static profile with two
 voters and three candidates, single-peaked with respect to a common axis,
 for which no sequence satisfies fixed-rank-wEJR. Hence, no version of EJR
 or wEJR can always be satisfied on unrestricted profiles.
 \item On complete profiles with at most two candidates, rank-wEJR is
 equivalent to rank-wPJR, and fixed-rank-wEJR is equivalent to
 fixed-rank-wPJR. Both equivalences also hold when the axioms are applied
 to any fixed interval of rounds.
 \item For every horizon $|T|\geq2$, there is a static strict
 one-dimensional Euclidean profile with two voters and two candidates for
 which no sequence satisfies all-rank-wEJR.
\end{enumerate}
\end{theorem}
\begin{proof}
\emph{(i)} Let $C=\set{a,b,c}$. In every round, voter~$1$ ranks $a\succ b\succ c$
and voter~$2$ ranks $c\succ b\succ a$; both rankings are single-peaked on
the axis $a,b,c$. At cutoff $1$, each singleton group agrees in every
round and is owed $\lfloor |T|/2\rfloor\geq1$ rounds. Thus both $a$ and
$c$ must win. At cutoff $2$, the whole electorate agrees on $b$ in every
round and is owed $|T|$ rounds. Fixed-rank-wEJR therefore requires one
voter to approve every winner. Voter~$1$ does not approve $c$, and voter~$2$
does not approve $a$, a contradiction.

\emph{(ii)} The case of one candidate is immediate. With two candidates,
fix a rank vector $\mathbf r$ and a group that agrees in every round under
$\mathbf r$. In a
round with cutoff one, all its members have the same first choice; in a
round with cutoff two, all its members approve both candidates. In either
case, its members have identical approval sets. Consequently, every member
of the group approves the winner in every round in which the group is
represented. The group's number of represented rounds is therefore each
member's number of approved winners, proving the rank equivalence.
Restricting the rank vector to be constant proves the fixed-rank
equivalence. The same argument applies to any interval of rounds.

\emph{(iii)} The voters have rankings $a\succ_1 b$ and $b\succ_2 a$ in every round.
These rankings are realized by placing the candidates and their respective
voters at $0$ and $1$. At cutoff one, each singleton group is owed
$\lfloor |T|/2\rfloor\geq1$ rounds, so both candidates must win.

Now choose individual cutoffs according to the winner. In a round won by
$a$, give voters~$1$ and~$2$ cutoffs $2$ and $1$, respectively. Both approve
$b$, but only voter~$1$ approves the winner. In a round won by $b$, give them
cutoffs $1$ and $2$; both approve $a$, but only voter~$2$ approves the winner.
The whole electorate agrees in every round and is owed $|T|$ rounds under
wEJR. Each voter approves exactly the rounds won by her first choice, so
neither approves all $|T|$ winners. This is a contradiction.
\end{proof}

By part~(ii), the online rule of Theorem~\ref{thm:first-choice-intervals}
will give fixed-rank-wEJR and rank-wEJR with two candidates and
complete rankings, on every time interval. Part~(iii) shows that this conclusion
does not extend to individual cutoffs, even with two voters and static
preferences.

We focus mainly on JR, PJR, and PSC in what follows. We return to the positive
two-candidate EJR case in Corollary~\ref{cor:ejr-candidate-horizon}, and give
a polynomial-time test for the individual-cutoff EJR axioms on this domain
in Proposition~\ref{prop:binary-all-rank-ejr-verification}.

\subsection{An Ordinal Interpretation of the All-rank Axioms}
\label{sec:pareto}
The weak all-rank axioms quantify over all individual cutoffs under which a group agrees in every round. Their meaning can also be expressed directly in terms of the rankings: a group is necessarily represented in a round exactly when it has no unanimously preferred alternative to the winner. For a group $S$, a candidate $c$ \emph{Pareto dominates} a candidate $c'$ in round $j$ if $c\succ_i^j c'$ for all $i\in S$; a candidate is \emph{Pareto efficient} for $S$ in round $j$ if no candidate Pareto dominates it for $S$ in that round. We prove the following lemma, and Theorem~\ref{thm:all-rank-characterization} below, in Appendix~\ref{app:ranked-axioms-proofs}.

\begin{lemma}
\label{lem:all-rank-minimum-representation}
Fix a sequence $W$ and a group $S$. The minimum number of rounds in which $S$ is represented, over all rank matrices $R$ under which $S$ agrees in every round, equals the number of rounds in which $w^j$ is Pareto efficient for $S$.
\end{lemma}

Consequently, $W$ satisfies all-rank-wPJR if and only if, for every group $S$, the winner is Pareto efficient for $S$ in at least $\lfloor |S||T|/n\rfloor$ rounds. Similarly, $W$ satisfies all-rank-wJR if and only if every group with $|S|\geq n/|T|$ has at least one such round. Lemma~\ref{lem:all-rank-minimum-representation} also connects the all-rank axioms to two familiar conditions. A sequence $W$ is \emph{Pareto efficient} if $w^j$ is Pareto efficient for $V$ in every round $j$. Aziz et al.~\cite{aziz2017condorcet} call a committee of $k$ members \emph{locally stable} if there is no candidate $c$ and no group of more than $n/(k+1)$ voters who all prefer $c$ to every committee member. The corresponding condition in our setting, using the Hare quota as in our other definitions~\cite{brill2023robust}, is as follows: $W$ is \emph{locally stable} if there is no group $S$ with $|S|\geq n/|T|$ such that, in every round $j$, some candidate $c^j$ Pareto dominates $w^j$ for $S$.

By Lemma~\ref{lem:all-rank-minimum-representation}, this temporal local-stability condition is equivalent to all-rank-wJR: both require every group of size at least $n/|T|$ to have a round in which the winner is Pareto efficient for it. All-rank-wPJR also implies Pareto efficiency of the sequence. Indeed, applying the lemma to $S=V$ requires a Pareto-efficient winner in all $|T|$ rounds.

The connection is with efficiency for each group, not only for the electorate as a whole. A sequence can be Pareto efficient for the electorate while repeatedly neglecting the same group; all-rank-wPJR limits how often this can happen according to the group's size.

\emph{Unanimity}, which is weaker than Pareto efficiency, requires that a candidate ranked first by every voter in a round be selected in that round. This already follows from rank-wPSC. If all voters rank $c$ first in round $j^*$, choose $r^{j^*}=1$ and $r^j=m$ in every other round. The whole electorate is then solid in every round and is owed $|T|$ represented rounds, so the winner in round $j^*$ must be $c$.

The preceding characterization concerns groups that agree in every round.
For the non-weak axioms, we must also account for rounds in which representation
is unavoidable even without agreement. A winner ranked first by a member of
$S$ is approved under every positive cutoff. This leads to two useful counts:
\[
 \begin{aligned}
 F_W(S)&=\left|\set{j\in T:w^j\text{ is the first choice of some }i\in S}\right|,\\
 D_W(S)&=\left|\set{j\in T:w^j\text{ is Pareto dominated for }S}\right|.
 \end{aligned}
\]
The rounds counted by $F_W(S)$ and $D_W(S)$ are disjoint. The remaining rounds
have a winner that is Pareto efficient for $S$ but is nobody's first choice
within $S$.

\begin{theorem}
\label{thm:all-rank-characterization}
For complete rankings, fix a sequence $W$ and a group $S$.
For a rank matrix $R$, let
$k_R(S)$ and $\mathrm{rep}_R(S)$ be the numbers of rounds in which $S$
agrees and is represented, respectively. Then
\[
 \max_R\left\{\left\lfloor\frac{k_R(S)|S|}{n}\right\rfloor
              -\mathrm{rep}_R(S)\right\}
 =\left\lfloor\frac{|S|(D_W(S)+F_W(S))}{n}\right\rfloor-F_W(S).
\]
Consequently, $W$ satisfies all-rank-PJR if and only if
$|S|D_W(S)-(n-|S|)F_W(S)<n$ for every group $S$. It satisfies all-rank-JR
if and only if there is no group $S$ with $F_W(S)=0$ and
$|S|D_W(S)\geq n$.
\end{theorem}

The left-hand side is the largest shortfall of $S$ over all rank matrices.
It is attained by cutoffs under which $S$ agrees without being represented
in every round whose winner is Pareto dominated for $S$, agrees in every
round whose winner is a member's first choice (where representation is
unavoidable), and does not agree in the remaining rounds.
The theorem removes the need to choose cutoffs when testing a fixed group.
For the weak axioms, the relevant count is the number of Pareto-efficient
winning rounds. For the non-weak axioms, first-choice representation also
matters. This distinction will be useful for verification on single-crossing
profiles in Section~\ref{sec:single-crossing-verification}.

\section{Existence and Computability Results}
\label{sec:attainability}
We now ask which axioms can be satisfied on unrestricted rankings. The answer depends on both the allowed cutoffs and the rounds in which a group must agree. We first give exact guarantees and show what information they require, then determine the optimal additive loss for all-rank-PJR. Table~\ref{tab:attainability} summarizes the exact guarantees; the running times refer to the rules presented here. The colors in Figure~\ref{fig:implication-full} place the same results on the diagram of implications among the axioms, showing for each axiom the least information with which it can always be satisfied.

\begin{table}[!htbp]
\centering
\small
\renewcommand{\arraystretch}{1.18}
\begin{tabular}{@{}p{0.33\linewidth}p{0.19\linewidth}p{0.42\linewidth}@{}}
\hline
\textbf{Axioms} & \textbf{Least information} & \textbf{Computation or nonexistence} \\
\hline
Fixed-rank JR, PJR, PSC & Offline & GCR, exponential time; Theorem~\ref{thm:fixed-rank-gcr}. \\
Fixed-rank and rank wJR, wPJR & Semi-online & MES, polynomial time; Theorem~\ref{thm:mes-rank-wpjr}. \\
Fixed-rank-wPSC & Online & SCR, polynomial time; Theorem~\ref{thm:online-scr}. \\
Rank-wPSC & Online & SCR, exponential time; Theorem~\ref{thm:online-scr}. \\
Rank JR, PJR, PSC & --- & Not always attainable; Theorem~\ref{thm:rank-impossibility}. \\
All-rank JR, PJR, wJR, wPJR & --- & Not always attainable; Example~\ref{thm:all-rank-small-example}. \\
Every version of EJR and wEJR & --- & Not always attainable; Theorem~\ref{thm:ejr}(i). \\
\hline
\end{tabular}
\caption{Exact guarantees for unrestricted rankings. Offline information is necessary in the first row by Theorem~\ref{thm:semi-online-impossibility}, and knowledge of the horizon is necessary in the second by Theorem~\ref{thm:online-impossibility}. Nonexistence means that some profile admits no sequence satisfying the axiom.
The two-candidate exceptions for weak EJR are given in
Corollary~\ref{cor:ejr-candidate-horizon}.}
\label{tab:attainability}
\end{table}

\subsection{Groups That Agree in Some Rounds}
\label{sec:some-rounds}
We first protect groups in proportion to the number of rounds in which they agree, without requiring agreement throughout the horizon. When the cutoff is common to all voters and rounds (fixed-rank), these requirements can be met for all such cutoffs simultaneously. Allowing the cutoff to change across rounds (rank), or requiring decisions without future preferences, leads to impossibility.

Fixed-rank-PJR can always be satisfied offline. Our rule adapts the two-stage Greedy Cohesive Rule (GCR) of Elkind et al.~\cite{elkind2025verifying}, which satisfies EJR for approval ballots and was later adapted to full justified representation~\cite{phillips2026strengthening}. The new difficulty is that a single sequence must meet the requirements for all approval profiles $A(1),\dots,A(m)$ simultaneously.

For a cutoff $r\in[m]$ and a group $S$, let $\tsr{S}{r}=\set{j\in T:\bigcap_{i\in S}A_i^j(r)\neq\varnothing}$ be the set of rounds in which $S$ agrees under $A(r)$, and let $\dem{S}{r}=\lfloor|\tsr{S}{r}|\cdot|S|/n\rfloor$ be the number of rounds that $S$ is owed under $A(r)$, which we call its \emph{demand}. Fixed-rank-PJR requires that, for every $r\in[m]$, every group $S$ is represented under $A(r)$ in at least $\dem{S}{r}$ rounds.

The rule has two stages (Algorithm~\ref{alg:fixed-rank-gcr}). The first stage decides how many rounds to reserve for which groups, and the second stage assigns the rounds and winners. Stage~1 processes the cutoffs $r=1,\dots,m$ in increasing order. For each $r$, it computes the \emph{remaining demand} $\drem{S}{r}$ of every group $S$ that agrees in some round: the part of $\dem{S}{r}$ that is not already covered by rounds reserved, at smaller cutoffs, for groups intersecting $S$. It then repeatedly selects a group with the largest remaining demand, reserves $g(S,r)=\drem{S}{r}$ rounds for it, and discards this group and all groups intersecting it for the rest of this cutoff. It moves on to the next cutoff when all remaining groups have zero remaining demand. Stage~2 processes the pairs $(S,r)$ with $g(S,r)>0$ in increasing order of $|\tsr{S}{r}|$, so that groups agreeing in fewer rounds are handled first. For each pair, it chooses $g(S,r)$ unassigned rounds in which $S$ agrees under $A(r)$, and in each of them selects a candidate approved by all members of $S$.

\begin{algorithm}[hbpt]
\caption{2-Stage Greedy Cohesive Rule for fixed-rank-PJR}
\label{alg:fixed-rank-gcr}
initialize $g(S,r) \gets 0$ for every $r \in [m]$ and $S \subseteq V$\;
\tcp{\textbf{Stage 1: determine how many rounds are reserved}}
\For{$r\gets 1$ \KwTo $m$}
{
    $\mathcal V \gets \set{S \subseteq V : \tsr{S}{r} \neq \varnothing}$\;
    \ForEach{$S \in \mathcal V$}
    {
        $\drem{S}{r} \gets \max\{0,\, \dem{S}{r} - \sum_{t=1}^{r-1}\sum_{S'\subseteq V, S'\cap S\neq\varnothing} g(S',t)\}$
    }
    \While{there exist $S\in\mathcal V$ such that $\drem{S}{r} > 0$}
    {
        choose $S \in \arg\max_{S\in\mathcal V} \drem{S}{r}$, breaking ties arbitrarily\;
    
    $g(S,r) \gets \drem{S}{r}$\;
    $\mathcal V\gets \{S'\in\mathcal V:S'\cap S=\varnothing\}$\;
    }
}
\tcp{\textbf{Stage 2: assign winners to reserved rounds}}
order all pairs $(S,r)$ with $g(S,r)>0$ as $(S_1,r_1),\dots,(S_q,r_q)$ so that $|\tsr{S_1}{r_1}|\leq\cdots\leq |\tsr{S_q}{r_q}|$, breaking ties arbitrarily\;

$T_{\mathrm{free}}\gets T$\;

\For{$h\gets 1$ \KwTo $q$}
{
    choose $T_h\subseteq T_{\mathrm{free}}\cap\tsr{S_h}{r_h}$ such that $|T_h|=g(S_h,r_h)$;

    \ForEach{$j\in T_h$}
    {
        choose $w^j\in \bigcap_{i\in S_h}A_i^j(r_h)$;
    }
    $T_{\mathrm{free}}\gets T_{\mathrm{free}}\setminus T_h$\;
}
choose an arbitrary $w^j$ for all $j \in T_{\mathrm{free}}$\;
\KwRet{$W$}\;
\end{algorithm}

\begin{theorem}
\label{thm:fixed-rank-gcr}
Algorithm~\ref{alg:fixed-rank-gcr} satisfies fixed-rank-PJR. 
\end{theorem}

The proof has two parts. First, Stage~2 can assign all reserved rounds: the total number of rounds reserved for pairs $(S,r)$ with $|\tsr{S}{r}|\leq k$ never exceeds $k$, and since Stage~2 processes the pairs in increasing order of $|\tsr{S}{r}|$, each pair finds enough unassigned rounds in which its group agrees. Second, the greedy choice in Stage~1 ensures that, for every cutoff $r$ and group $S$, at least $\dem{S}{r}$ rounds are reserved at cutoffs at most $r$ for groups intersecting $S$. A winner approved at a smaller cutoff is still approved at cutoff $r$, so these rounds represent $S$ under $A(r)$. Appendix~\ref{app:gcr-proof} gives the full proof.

Stage~1 considers all groups of voters, so the rule runs in time $O(2^n\cdot\mathrm{poly}(n,m,|T|))$. Whether a polynomial-time rule can guarantee fixed-rank-PJR is an open question. We next examine whether the cutoff can vary across rounds, and whether future preferences are necessary.

The fixed-rank guarantee does not extend to the rank interpretation: even rank-JR and rank-PSC can be unattainable on a static one-dimensional Euclidean profile. The idea is that a rule cannot spread its winners over all blocks of similar candidates: a group whose block is neglected can use a large cutoff in the rounds in which the winner lies outside its block, and a cutoff of one in the remaining rounds.

\begin{theorem}
\label{thm:rank-impossibility}
There is a static one-dimensional Euclidean profile for which no sequence
satisfies rank-JR and no sequence satisfies rank-PSC. Hence, rank-PJR cannot
always be satisfied either.
\end{theorem}

\begin{proof}
Let $n=m=12$ and $|T|=5$. Partition the candidates into three blocks
\[
B_1=\set{c_1,c_2,c_3,c_4},\quad
B_2=\set{c_5,c_6,c_7,c_8},\quad
B_3=\set{c_9,c_{10},c_{11},c_{12}}.
\]
Place the candidates, respectively, at
\[
1,2,3,4;\quad
101,102,103,104;\quad
201,202,203,204.
\]
For every $i\in V$, place voter $i$ at the position of $c_i$ minus
$1/10$. In every round, each voter ranks the candidates by increasing
distance, and the same rankings are used in all five rounds. No voter is
equidistant from two candidates, so the rankings are strict. Voter $i$
ranks $c_i$ first, and, for each $s\in\set{1,2,3}$, every voter $i$ with
$c_i\in B_s$ has $B_s$ as her set of four most-preferred candidates.

Fix any sequence $W=(w^1,\dots,w^5)$. Since there are five rounds and
three blocks, there is some $B\in\set{B_1, B_2, B_3}$ such that
$w^j\in B$ for at most one $j\in T$. Hence, at least three candidates in
$B$ are never selected. Let $S$ consist of the three voters whose
first-ranked candidates are three such candidates. Define $\mathbf r=(r^1,\dots,r^5)$ by
\[
r^j=
\begin{cases}
4, & \text{if } w^j\notin B,\\
1, & \text{if } w^j\in B.
\end{cases}
\]

There are at least four rounds with $w^j\notin B$. In each such round,
$A^j_i(r^j)=B$ for every $i\in S$, whereas $w^j\notin B$. In every round
with $w^j\in B$, each voter $i\in S$ approves only $c_i$, and none of
these candidates is ever selected. Thus, $w^j\notin\bigcup_{i\in S}A^j_i(r^j)$ for every $j \in T$.
The group $S$ agrees in at least four rounds and $|S|=3=n/4$. Therefore,
for $\ell=1$, rank-JR requires $S$ to be represented in at least one
round, a contradiction. Since rank-PJR implies rank-JR, no sequence
satisfies rank-PJR either.

Moreover, in each round with $w^j\notin B$, all members of $S$ have the
approval set $B$, so $S$ is solid under $\mathbf r$ in at least four rounds. Since $|S|=n/4$ and $S$ is
not represented in any round, rank-PSC is also violated for $\ell=1$.
\end{proof}

\paragraph{The need for future preferences.}
\label{sec:information}
Even the attainable fixed-rank guarantees require future preferences. A semi-online rule must commit to early winners without knowing which groups will agree later. The following construction exploits this: in the first four rounds, the rule must select winners preferred by four of eight pairs of voters, and the preferences in the last four rounds are then chosen so that the members of these four pairs whose first choices were not selected form a new group.

\begin{theorem}
\label{thm:semi-online-impossibility}
No semi-online rule always satisfies fixed-rank-JR, and no semi-online rule always satisfies fixed-rank-PSC, even on profiles that are single-peaked with respect to a common axis. Hence, no semi-online rule always satisfies fixed-rank-PJR.
\end{theorem}
\begin{proof}
Suppose, for contradiction, that a semi-online rule always satisfies fixed-rank-PSC. Let $n=16$, $|T|=8$, and $C=\set{a_1,b_1,\dots,a_8,b_8}$, with the axis $a_1,b_1,a_2,b_2,\dots,a_8,b_8$. Partition the voters into eight pairs $S_1=\set{1,2},\dots,S_8=\set{15,16}$, and the candidates into blocks $B_h=\set{a_h,b_h}$ for $h\in[8]$. In the first four rounds, for every $h$, one voter in $S_h$ ranks $a_h\succ b_h$ at the top and the other ranks $b_h\succ a_h$ at the top; we complete these rankings to single-peaked rankings on the axis, and each voter uses the same ranking in all four rounds.

Let $H$ be the set of indices $h$ such that some winner among $w^1,\dots,w^4$ belongs to $B_h$; clearly, $|H|\leq4$. Since the rule is semi-online, $w^1,\dots,w^4$ do not depend on the preferences in rounds $5,\dots,8$, which we now define. Every candidate is some voter's first choice in the first four rounds; let $\succ^*$ be the ranking of a voter whose first choice is $w^1$. In rounds $5,\dots,8$, every voter in a pair $S_h$ with $h\in H$ has the ranking $\succ^*$, and every other voter keeps her ranking from the first four rounds. All rankings are single-peaked on the axis.

For every $h\notin H$, the pair $S_h$ is solid over $B_h$ under the cutoff $r=2$ in all eight rounds, and $|S_h|=n/|T|$. Thus, fixed-rank-PSC requires a winner from $B_h$, which can only occur in rounds $5,\dots,8$. Since the blocks are disjoint, rounds $5,\dots,8$ contain winners from at most four such blocks, so $8-|H|\leq4$, and hence $|H|=4$. Consequently, $w^1,\dots,w^4$ belong to four distinct blocks, one winner per block; $w^5,\dots,w^8$ belong to the four blocks $B_h$ with $h\notin H$, again one winner per block; and for every $h\in H$, exactly one voter in $S_h$ has her first choice selected in the first four rounds. Let $S$ consist of the other voter of each of these four pairs, so that $|S|=4$.

Under the cutoff $r=1$, the group $S$ is solid over $\set{w^1}$ in rounds $5,\dots,8$, and $|S|=n/4$. Hence, fixed-rank-PSC requires a round in which the winner is the first choice of a member of $S$. In the first four rounds, the first choices of members of $S$ are not selected, by the choice of $S$. In the last four rounds, every member's first choice is $w^1$, which belongs to a block $B_h$ with $h\in H$, whereas $w^5,\dots,w^8$ belong to blocks $B_h$ with $h\notin H$. This contradiction proves the claim for fixed-rank-PSC. Every solid group agrees, and every requirement above uses $\ell=1$, so the same argument proves the claim for fixed-rank-JR.
\end{proof}

\subsection{Groups That Agree in Every Round}
\label{sec:weak-rules}
We now restrict protection to groups that agree throughout the horizon. This permits a common cutoff to vary across rounds: rank-wPJR is achievable without knowing future preferences, provided that the horizon is known. The proof shows that, until a group receives its entitlement, its members pay only in rounds that represent it.

The rule adapts the \emph{Method of Equal Shares (MES)}, following the simplified version of Brill and Peters~\cite{brill2023robust}. Each voter starts with a budget of $|T|/n$, and every winner costs one unit. Given budgets $(b_i)_{i\in V}$ and a cutoff $r$, candidate $c$ is \emph{$\rho$-affordable} in round $j$ if $\rho>0$ and $\sum_{i\in N_c^j(r)}\min\set{\rho,b_i}=1$; that is, the voters who approve $c$ at cutoff $r$ can buy it if each pays at most $\rho$. In each round, the rule finds the smallest cutoff $r$ at which some candidate is affordable (that is, the voters approving it have a total budget of at least one), selects the candidate that is $\rho$-affordable for the smallest $\rho$, and deducts $\min\set{\rho,b_i}$ from each voter $i\in N_{w^j}^j(r)$ (Algorithm~\ref{alg:MES1}). Unlike MES for multi-winner elections and for temporal voting with approval ballots, this rule never stops early: at cutoff $m$, every voter approves every candidate, so some candidate is always affordable.

\begin{algorithm}[hbpt]
\caption{MES for rank-wPJR}
\label{alg:MES1}
$b_i \gets \frac{|T|}{n}$ for all $i \in V$\;
\For{$j \gets 1$ \KwTo $|T|$}
{
    \For{$r \gets 1$ \KwTo $m$}
    {
        $X \gets \{ c \in C : \sum_{i \in N^j_c(r)} b_i \geq 1 \}$\;
        \If{$X \neq \varnothing$}
        {
            select $c \in X$ that is $\rho$-affordable for the lowest $\rho$\;
            $w^j \gets c$\;
            \For{$i \in N^j_c(r)$}
            {
                $b_i \gets b_i - \min(\rho,b_i)$\;
            }
            \textbf{continue} to the next $j$\;
        }  
    }
}
\KwRet{$W$}
\end{algorithm}

\begin{theorem}
\label{thm:mes-rank-wpjr}
    Algorithm~\ref{alg:MES1} is a polynomial-time semi-online rule that satisfies rank-wPJR.
\end{theorem}
\begin{proof}
Before round $j$, the total remaining budget is $|T|-j+1\geq1$.
At cutoff $m$, every voter approves every candidate, so the rule can always
buy a winner. Each round spends exactly one unit, and all budgets remain
nonnegative. Consequently, every voter's budget is zero at the end.

Suppose that rank-wPJR is violated by a rank vector $\mathbf r$, a group $S$
that agrees in every round, and an integer $\ell$ with
$|S|\geq\ell n/|T|$. The initial total budget of $S$ is at least $\ell$, but
$S$ is represented in at most $\ell-1$ rounds. Since all budgets are zero at
the end, the members of $S$ pay at least $\ell$ units in total. Each round
costs one unit, so they pay at most $\ell-1$ units in represented rounds, and
must therefore pay a positive amount in some unrepresented round.
Choose the first such round $j$.

Before this round, $S$ has paid only in represented rounds, of which there
are at most $\ell-1$. Its remaining budget is therefore at least one.
Any candidate in $\bigcap_{i\in S}A_i^j(r^j)$ is affordable at cutoff $r^j$.
The algorithm consequently buys the winner at some cutoff $r\leq r^j$.
If a member $i\in S$ pays in that round, then
$w^j\in A_i^j(r)\subseteq A_i^j(r^j)$, contradicting the choice of $j$ as
unrepresented. This proves rank-wPJR. 

For each candidate and cutoff, the smallest affordable price $\rho$ can be
found by sorting the relevant budgets and solving the piecewise-linear
equation $\sum_{i\in N_c^j(r)}\min\{\rho,b_i\}=1$. There are at most
$|T|m^2$ such calculations. Thus, the running time is polynomial. The rule uses the
horizon but no future preferences, so it is semi-online.
\end{proof}

A small change to the initial budgets gives the Droop size condition. The
additional budget is not fully spent; the proof instead compares a group's
remaining budget with that of the entire electorate.

\begin{corollary}
\label{cor:mes-droop}
If Algorithm~\ref{alg:MES1} starts each voter with $(|T|+1)/n$ units,
it satisfies rank-wPJR with the Droop size condition. In particular, every
group $S$ that agrees in every round under some rank vector is
represented in at least $\lceil (|T|+1)|S|/n\rceil-1$ rounds.
\end{corollary}
\begin{proof}
Suppose that a group $S$ is represented in $q$ rounds, fewer than its stated
entitlement. Its initial budget $x=(|T|+1)|S|/n$ then satisfies $x>q+1$.
Before a first payment in an unrepresented round, $S$ would retain at least
$x-q>1$ units. A common approved candidate would be affordable at the group's
cutoff, so the rule would select a winner at no larger cutoff. Any member
charged for that winner would approve it, a contradiction. Hence $S$ pays
only in represented rounds and retains more than one unit at the end.
Every round spends one unit, so the electorate starts with $|T|+1$ units
and ends with exactly one. This is impossible.
\end{proof}

\paragraph{The need to know the horizon.}
The horizon is needed for the preceding guarantee: without it, a rule cannot guarantee even fixed-rank-wJR. The reason is that the horizon changes which groups are large enough to be owed a round, and hence whether a compromise candidate should be selected. The construction follows the idea of Theorem~7.2 of Aziz et al.~\cite{aziz2024committee}.

\begin{theorem}
\label{thm:online-impossibility}
No online rule always satisfies fixed-rank-wJR. Hence, no online rule always satisfies any of the approval-based axioms.
\end{theorem}
\begin{proof}
Suppose, for contradiction, that an online rule always satisfies fixed-rank-wJR. Let $n=4$ and $C=\set{c_1,c_2,c_3,c_4,x}$, and consider the static profile in which each voter $i$ ranks $c_i\succ x$ above all other candidates.

First, let $|T|=2$ and $r=2$. At least two voters, say $i$ and $i'$, have first choices that are not selected. The group $\set{i,i'}$ agrees on $x$ in both rounds and $|\set{i,i'}|=n/|T|$, so it must be represented. Its members approve only $c_i$, $c_{i'}$, and $x$, so $x\in\set{w^1,w^2}$.

Next, let $|T|=4$ and $r=1$. Every singleton group $\set{i}$ agrees in every round and $|\set{i}|=n/|T|$, so $c_i$ must win in some round. Hence, the four winners are $c_1,\dots,c_4$, and $x$ never wins.

The first two rounds have the same preferences in both instances. Since the rule is online, it selects the same $w^1$ and $w^2$ in both instances, which contradicts the two conclusions above.
\end{proof}

\paragraph{Solid coalitions without a known horizon.}
The online impossibility does not extend to solid coalitions: fixed-rank-wPSC and rank-wPSC can be satisfied online by adapting the Solid Coalition Refinement (SCR) rule of Aziz et al.~\cite{aziz2024committee}. Consider round $j$ and a rank vector $\mathbf r=(r^1,\dots,r^j)\in[m]^j$ for the rounds so far. If a group $S$ is solid under $\mathbf r$ in every round $t\leq j$, we call the pair $(S,\mathbf r)$ a \emph{claim}, and write $C^t(S,\mathbf r)$ for the common approval set $A_i^t(r^t)$ of the members $i\in S$ in round $t$. The \emph{representation} of a claim before round $j$ and its \emph{priority} are
\[
 \mathrm{rep}^{j-1}(S,\mathbf r)=\left|\set{t<j:w^t\in C^t(S,\mathbf r)}\right|,
 \quad
 \pi^{j-1}(S,\mathbf r)=\frac{|S|}{\mathrm{rep}^{j-1}(S,\mathbf r)+1}.
\]
Thus, larger groups have higher priority, and a group's priority decreases each time it is represented. We define $\mathrm{rep}^j(S,\mathbf r)$ and $\pi^j(S,\mathbf r)$ in the same way, counting round $j$ as well, once $w^j$ is selected.

In each round $j$, the rule considers the constant rank vectors $(r,\dots,r)$ for fixed-rank-wPSC, or all rank vectors in $[m]^j$ for rank-wPSC. For each such vector, it keeps the claims of the maximal groups that are solid under it in every round so far; these groups partition the voters. To select the winner, the rule starts with $X=C$ and repeatedly refines $X$: it selects the claim of highest priority whose common approval set in round $j$ is a proper subset of $X$, and replaces $X$ by this set. When a single candidate remains, it is the winner (Algorithm~\ref{alg:online-scr}).

\begin{algorithm}[hbpt]
\caption{Online SCR}
\label{alg:online-scr}
\For{$j=1,2,\dots$}
{
    $\mathsf{Claims}^j\gets\varnothing$\;
    $\mathcal{R}^j \gets 
    \begin{array}{ll}
    \set{(r,\dots,r) : r\in[m]} & \text{for fixed-rank-wPSC}\\
    \set{(r^1,\dots,r^j) : r^1,\dots,r^j \in[m]} & \text{for rank-wPSC}
    \end{array}$ 

    \ForEach{$\mathbf r=(r^1,\dots,r^j)\in\mathcal {R}^j$}
    {
        $\mathcal V^j(\mathbf r) \gets$ partition $V$ into maximal groups $S$ satisfying \\ $A^t_i(r^t) = A^t_{i'}(r^t)$ for all $i,i' \in S$ and all $t \leq j$\;
        \ForEach{$S\in\mathcal V^j(\mathbf r)$}
        {
            add the claim $(S,\mathbf r)$ to $\mathsf{Claims}^j$\;
        }
    }
    $X\gets C$\;
    \While{$|X|>1$}{
        choose $(S,\mathbf r) \in \arg\max_{(S,\mathbf r)\in\mathsf{Claims}^j, 
        C^j(S,\mathbf r)\subsetneq X} \pi^{j-1}(S,\mathbf{r})$, using a fixed tie-breaking order\;
        $X\gets C^j(S,\mathbf r)$\;
    }
    $w^j \gets$ the unique candidate in $X$\;
}
\end{algorithm}

\begin{theorem}\label{thm:online-scr}
Algorithm~\ref{alg:online-scr} is an online rule that satisfies fixed-rank-wPSC if $\mathcal R^j = \set{(r,\dots,r):r\in[m]}$, and rank-wPSC if $\mathcal R^j = \set{(r^1,\dots,r^j):r^1,\dots,r^j\in[m]}$.
\end{theorem}
\begin{proof}
We use the same idea as in Theorem 5.5 of \cite{aziz2024committee} and first show that the algorithm is well-defined. That is, we need to show that in each iteration of the refinement loop, an \emph{eligible} claim exists, namely a claim whose common approval set in round $j$ is a proper subset of $X$.
Consider an iteration in which $|X| >1$. If $X=C$, choose an arbitrary voter $i$; otherwise, $X$ is the common approval set of the claim $(S,\mathbf r)$ selected in the previous iteration, and we choose a voter $i \in S$. In both
cases, voter $i$'s first choice $c^*$ belongs to $X$. Note that the vector $\mathbf{r} = (1,\dots,1)$ belongs to $\mathcal{R}^j$ in both versions of the rule. The claim of the maximal group containing $i$ for this vector belongs to
$\mathsf{Claims}^j$ and has common approval set $\set{c^*}\subsetneq X$. Therefore, an eligible claim always exists. In addition, since every refinement replaces $X$ by a strict subset, the loop terminates after at most $m-1$ iterations, and thus the algorithm always terminates.

We now prove that the algorithm satisfies the relevant version of wPSC. Fix an arbitrary horizon $|T|$ and let $p=\frac{n}{|T|}$.
The value $p$ is used only in the proof; it is not known to
the SCR algorithm. Give every voter one unit of virtual budget. We show retrospectively that the winner in each round can be bought at price $p$ without charging any voter more than her remaining budget; all $|T|$ winners then cost $|T| p =n$, exactly the total initial budget. We maintain the following invariant after every round $j$:
for every rank vector $\mathbf{r} \in \mathcal{R}^j$ and every group $S$ that is solid under $\mathbf{r}$ in every round $t\leq j$, if
\begin{equation}
\label{eq:scr-budget-condition}
\pi^j(S,\mathbf{r}) \geq p,
\end{equation}
then the voters in $S$ have been charged only in rounds $t\leq j$ with
$w^t\in A_i^t(r^t)$ for $i\in S$. The invariant concerns all such groups,
not only the maximal groups kept by the algorithm. When using it in
round $j-1$, we restrict the rank vector to its first $j-1$ entries.

Clearly, the invariant holds before the first round, since no payments have been made. Consider round $j$, and assume that the invariant holds after round $j-1$. There are two possible cases. In the first case, at least one of the claims selected during the refinement loop (at round $j$) has a priority of at least $p$.
Let $(S',\mathbf{r}')$ be the last selected claim whose
priority is at least $p$, and write $u' = \mathrm{rep}^{j-1}(S',\mathbf{r}')$. Then $\pi^{j-1}(S',\mathbf{r}') = \frac{|S'|}{u'+1}\geq p$.
Let $\bar{\mathbf r}$ consist of the first $j-1$
entries of $\mathbf{r}'$. When $j>1$, the group $S'$ is solid under $\bar{\mathbf r}$ in every round $t\leq j-1$, and $\bar{\mathbf r} \in \mathcal{R}^{j-1}$. Therefore, by the induction hypothesis, the voters of $S'$ were previously charged only in the $u'$ rounds with winners that represented them.
Since the cost of every winner is $p$, they have previously paid at most $u'p$. Their remaining total budget is therefore at least $|S'|-u'p$.
For $j=1$, the same bound holds directly because $u'=0$.
Recall that $\frac{|S'|}{u'+1}\geq p$, and thus $|S'|-u'p \geq p$. 
Thus, the voters in $S'$ can pay the price $p$ for the winner in round $j$. Charge the entire price to these voters, without exceeding their remaining budgets. 

We now need to verify that the invariant remains true.
First note that after the claim $(S',\mathbf{r}')$ is selected, all subsequent refinements are subsets of $C^j(S',\mathbf{r}')$. Therefore, $w^j\in C^j(S',\mathbf{r}')$.
Consider a claim $(S, \mathbf{r})$ satisfying condition~\eqref{eq:scr-budget-condition}. If $w^j$ represents $S$, the invariant continues to hold for $S$:
since $\pi^{j-1}(S,\mathbf{r}) \ge \pi^j(S,\mathbf{r})\geq p$, the induction hypothesis applies to $S$ up to round $j-1$, and any payment made in round $j$ is for a winner that represents $S$. Otherwise, $w^j\notin C^j(S,\mathbf{r})$, so
$\mathrm{rep}^j(S,\mathbf r)=\mathrm{rep}^{j-1}(S,\mathbf r)$ and $\pi^{j-1}(S, \mathbf{r}) = \pi^j(S, \mathbf{r}) \geq p$.
We claim that $S\cap S'=\varnothing$. Suppose otherwise, and choose $i\in S\cap S'$. The two sets $C^j(S,\mathbf{r})$ and $C^j(S',\mathbf{r}')$ are both approval sets of voter $i$, so one contains the other. Recall that $w^j\notin C^j(S,\mathbf{r})$ but $w^j\in C^j(S',\mathbf{r}')$, and thus $C^j(S,\mathbf{r}) \subsetneq C^j(S',\mathbf{r}')$.
Let $S^{*}$ be the maximal group in $\mathcal V^j(\mathbf{r})$ containing $S$. Note that $S^{*}$ has the same common approval sets and representation as $S$.
In particular, $C^j(S^{*},\mathbf{r}) \subsetneq C^j(S',\mathbf{r}')$. Since $C^j(S^{*},\mathbf{r})$ is nonempty, $|C^j(S',\mathbf{r}')|\geq2$, so the refinement loop continues after $(S',\mathbf{r}')$ is selected, and the claim $(S^{*},\mathbf{r})$ is eligible at the next refinement. In addition, $\pi^{j-1}(S^{*},\mathbf{r}) \geq \pi^{j-1}(S,\mathbf{r}) \geq p$, so the next refinement would select a claim of priority at least $p$.
This contradicts the choice of
$(S',\mathbf{r}')$ as the last selected claim whose priority is at least $p$.
Thus, $S\cap S'=\varnothing$. Since only voters in $S'$ pay in round $j$, no voter in the unrepresented group $S$
is charged, and thus the invariant remains true.

For the second case, assume that every claim selected by the refinement
loop has priority below $p$. Before round $j$, the total remaining budget is
$n-(j-1)p=(|T|-j+1)p\geq p$, so the price can be paid from the remaining
budgets.

Consider any group $S$ and rank vector $\mathbf r$ satisfying the invariant's
condition after round $j$. If this round does not represent $S$, its priority
was also at least $p$ before the round, and its common approval set in round $j$
is a proper subset of $C$, since it does not contain $w^j$. The maximal group containing $S$ for
$\mathbf r$ has the same common approval sets and representation, and at
least as many voters. Its claim was therefore eligible at the first
refinement and had priority at least $p$, so the first refinement would have selected a claim of priority at least $p$, a contradiction. Thus, every group
to which the invariant applies after this round is represented in this
round. Such a group also met the priority condition before the round, since
representation counts cannot decrease. Any feasible distribution of the
price $p$ consequently preserves the invariant. This completes the induction.

After $|T|$ rounds, the total payment is $|T|p=n$. Since each voter started
with one unit and budgets remain nonnegative, every voter's budget is zero.
Suppose that the required wPSC axiom is violated by $S$, $\mathbf r\in
\mathcal R^{|T|}$, and $\ell\in\mathbb N$. Then $|S|\geq\ell p$ and
$\mathrm{rep}^{|T|}(S,\mathbf r)\leq\ell-1$, so
\[
\pi^{|T|}(S,\mathbf r)
=\frac{|S|}{\mathrm{rep}^{|T|}(S,\mathbf r)+1}\geq p.
\]
The invariant implies that $S$ paid only in its represented rounds. Its
members therefore paid at most $(\ell-1)p<|S|$, contradicting their zero
remaining budget. Constant rank vectors give fixed-rank-wPSC, and arbitrary
rank vectors give rank-wPSC.
\end{proof}

For each rank vector $\mathbf r\in\mathcal R^j$, the groups in $\mathcal V^j(\mathbf r)$ partition the voters, so $|\mathsf{Claims}^j|\leq n|\mathcal R^j|$; these groups are found by grouping the voters according to $(A_i^1(r^1),\dots,A_i^j(r^j))$. The refinement loop has at most $m-1$ iterations, and each iteration scans $\mathsf{Claims}^j$. Hence, round $j$ takes time $|\mathcal R^j|\cdot\mathrm{poly}(n,m,j)$. For fixed-rank-wPSC, this is polynomial. For rank-wPSC, it is $m^j\cdot\mathrm{poly}(n,m,j)$, so this version is not a polynomial-time rule. If the horizon is known, Algorithm~\ref{alg:MES1} satisfies rank-wPJR, and hence rank-wPSC, in polynomial time.

Unlike MES, SCR needs no change to the rule to obtain the Droop
size condition: its budgets are used only in the proof.

\begin{corollary}
\label{cor:scr-droop}
Each version of Algorithm~\ref{alg:online-scr} satisfies its corresponding
weak PSC axiom with the Droop size condition.
\end{corollary}
\begin{proof}
Use the payment argument of Theorem~\ref{thm:online-scr} with
$p=n/(|T|+1)$ instead of $n/|T|$. The proof of the invariant is unchanged,
and the total remaining budget before every round is at least $p$.
After $|T|$ rounds the electorate retains $n-|T|p=p$ units.
If a solid group $S$ with $|S|>\ell p$ is represented in only
$q\leq\ell-1$ rounds, its final priority exceeds $p$. By the invariant,
it has paid only in those $q$ rounds and retains
$|S|-qp>(\ell-q)p\geq p$ units. This exceeds the electorate's remaining
budget, a contradiction.
\end{proof}

The constant one in this quota cannot be increased uniformly over all
profiles. At horizon one, take two equally sized groups with opposite first
choices. For any $\beta>1$, the entitlement
$\lceil(1+\beta)|S|/n\rceil-1$ would require both groups to be represented
at cutoff one, which a single winner cannot achieve.

Section~\ref{sec:structured} shows that a common single-crossing order
permits a polynomial-time online rule for rank-wPSC. Theorem~\ref{thm:truncated-common-sc-scr}
extends that result to truncated ballots and also gives polynomial-time
verification on this domain.

\paragraph{Individual cutoffs.}
Requiring agreement in every round makes rank-wPJR attainable, but allowing a different cutoff for each voter can make even all-rank-wJR unattainable. To see why, consider rankings under which the only candidates that are Pareto efficient for a group are the first choices of its members. By Lemma~\ref{lem:all-rank-minimum-representation}, such a group is then guaranteed representation only in rounds in which the first choice of one of its members wins.

Let $c_1,\dots,c_d$ be candidates. A ranking is \emph{cyclic} (on $c_1,\dots,c_d$) if, for some $x\in[d]$, it ranks
\[
 c_x\succ c_{x+1}\succ\cdots\succ c_{x+d-1}
\]
above all other candidates, with indices read modulo $d$. For a group whose rankings are all cyclic, Pareto efficiency is easy to describe.

\begin{lemma}
\label{lem:cyclic}
Suppose that, in round $j$, every member of a group $S$ has a cyclic ranking on $c_1,\dots,c_d$. Then a candidate is Pareto efficient for $S$ in round $j$ if and only if it is the first choice of some member of $S$.
\end{lemma}
\begin{proof}
A member's first choice is not Pareto dominated, since she ranks it above every other candidate. Every candidate outside $\set{c_1,\dots,c_d}$ is Pareto dominated by any member's first choice. Finally, let $c_a$ be a candidate that is not the first choice of any member, and let $c_b$ be the last candidate before $c_a$, in the cyclic order $c_1,\dots,c_d$, that is the first choice of a member. No member's ranking starts strictly after $c_b$ and at or before $c_a$ in this order, so every member ranks $c_b$ above $c_a$. Thus $c_b$ Pareto dominates $c_a$.
\end{proof}

If every round has cyclic rankings on a common set of candidates for that
round, Lemma~\ref{lem:cyclic} gives
$D_W(S)+F_W(S)=|T|$ for every group $S$.
Theorem~\ref{thm:all-rank-characterization} and
Lemma~\ref{lem:all-rank-minimum-representation} then imply that
all-rank-PJR and all-rank-wPJR coincide on these profiles. The same is true
of all-rank-JR and all-rank-wJR. Thus, for the cyclic examples below, allowing
agreement in only some rounds does not strengthen the requirement.

\begin{example}
\label{thm:all-rank-small-example}
Let $n=m=4$ and $|T|=2$, and let the profile be static, with voter $i$ ranking $c_i\succ c_{i+1}\succ c_{i+2}\succ c_{i+3}$ (indices modulo four). Every sequence leaves at least two candidates unselected. Let $S$ consist of the two voters whose first choices are unselected. By Lemma~\ref{lem:cyclic}, neither winner is Pareto efficient for $S$. Since $|S|=n/|T|$, Lemma~\ref{lem:all-rank-minimum-representation} yields a violation of all-rank-wJR. Hence, no version of an all-rank axiom can always be satisfied.
\end{example}

\subsection{Optimal Additive Guarantees}
\label{sec:all-rank}
\label{sec:electorate-size}
The preceding example rules out an exact all-rank guarantee on every profile. We now determine how many represented rounds may be lost for a given electorate size and horizon. The upper bound uses only first choices and is achieved online; the lower bound holds even when the entire profile is known.

We say that $W$ satisfies all-rank-PJR \emph{with additive loss $L$}, where $L\geq0$ is an integer, if, for every rank matrix $R$, every group $S$ that agrees in $k>0$ rounds under $A(R)$ is represented in at least $\lfloor k|S|/n\rfloor-L$ rounds. For all-rank-wPJR, this requirement is restricted to groups that agree in every round. A loss of zero gives the corresponding exact axiom.

Selecting a voter's first choice represents her under every cutoff. Thus, the following sufficient condition provides a guarantee for all-rank-PJR without having to consider the cutoffs or the rounds of agreement separately.

\begin{lemma}
\label{lem:first-choice}
Let $L\geq0$ be an integer. Suppose that, for every group $S$, in at least $\lfloor |S||T|/n\rfloor-L$ rounds the winner is the first choice of some member of $S$, that is,
\[
 \left|\set{j\in T:w^j\in\bigcup_{i\in S}A_i^j(1)}\right|\geq\left\lfloor\frac{|S||T|}{n}\right\rfloor-L.
\]
Then $W$ satisfies all-rank-PJR with additive loss $L$.
\end{lemma}
\begin{proof}
Every round counted on the left-hand side represents $S$ under every rank matrix. The required bound follows from $k\leq|T|$.
\end{proof}

Note that Lemma~\ref{lem:first-choice} does not require the group to agree in any round.

Consider the \emph{serial dictatorship} rule that cycles through the voters in a fixed order and, in each round, selects the current voter's first choice. When $n$ divides $|T|$, every voter chooses equally often, and Lemma~\ref{lem:first-choice} gives all-rank-PJR. Otherwise, only the last incomplete cycle can leave a group short of its proportional share. The following theorem bounds the resulting loss and shows that it is optimal.

\begin{theorem}
\label{thm:all-rank-additive}
Write $|T|=qn+s$, where $0\leq s<n$. Serial dictatorship is an online rule that satisfies all-rank-PJR with additive loss $\lfloor s(n-s)/n \rfloor$.
No smaller additive loss is possible for all profiles with $n$ voters and horizon $|T|$, even for an offline rule and all-rank-wPJR. The lower bound holds on a static profile with $m=n$.
\end{theorem}

\begin{proof}
After $|T|=qn+s$ rounds, $s$ voters have chosen $q+1$ times and the other $n-s$ voters have chosen $q$ times. Thus, for every group $S$, the number of rounds in which a member of $S$ chooses the winner is at least $q|S|+\max\set{0,\,|S|-(n-s)}$.
Since $\lfloor |S||T|/n\rfloor=q|S|+\lfloor s|S|/n\rfloor$, the shortfall with respect to Lemma~\ref{lem:first-choice} is at most $\lfloor s|S|/n \rfloor-\max\set{0,\,|S|-(n-s)}$.
For $|S|\leq n-s$, this is at most $\lfloor s(n-s)/n\rfloor$. For $|S|\geq n-s$, it equals $\lfloor (n-s)(n-|S|)/n\rfloor$, which has the same upper bound because $n-|S|\leq s$. Lemma~\ref{lem:first-choice} proves the upper bound. The rule does not use the horizon, so it is online.

For the lower bound, let $C=\set{c_1,\dots,c_n}$ and consider the static profile in which each voter $i$ has the cyclic ranking
\begin{equation}
 c_i\succ_i c_{i+1}\succ_i\cdots\succ_i c_{i+n-1},
 \label{eq:cyclic-profile}
\end{equation}
with indices read modulo $n$. If $s=0$, the claimed loss is zero, so suppose that $s>0$. Fix any sequence $W$, let $a_i$ be the number of rounds in which $c_i$ wins, and let $S$ consist of the $n-s$ voters with the smallest values of $a_i$. Then $\sum_{i\in S}a_i\leq q(n-s)$.
Indeed, if the sum were at least $q(n-s)+1$, some $a_i$ with $i\in S$ would be at least $q+1$, and so would each of the other $s$ values. The total number of rounds would then be at least $q(n-s)+1+s(q+1)=|T|+1$, a contradiction.

By Lemma~\ref{lem:cyclic}, the winner is Pareto efficient for $S$ in exactly $\sum_{i\in S}a_i$ rounds. By Lemma~\ref{lem:all-rank-minimum-representation}, some rank matrix makes $S$ agree in every round and represents it in only this many rounds, whereas
\[
 \left\lfloor\frac{|S||T|}{n}\right\rfloor
 =q(n-s)+\left\lfloor\frac{s(n-s)}{n}\right\rfloor.
\]
Hence, even all-rank-wPJR can require the stated loss.
\end{proof}

\begin{corollary}
\label{cor:all-rank-existence}
Fix $n$ and $|T|$, and allow the set of candidates to vary. Every profile admits a sequence satisfying all-rank-PJR if and only if $|T|\bmod n\in\set{0,1,n-1}$.
Every profile admits a sequence satisfying all-rank-JR if and only if $|T|=1$ or $|T|\geq n-1$. Both characterizations also hold for the corresponding weak axioms, and in each positive case serial dictatorship suffices.
\end{corollary}

\begin{proof}
For PJR and wPJR, Theorem~\ref{thm:all-rank-additive} gives universal existence exactly when $s(n-s)<n$. This holds for $s\in\set{0,1,n-1}$. If $2\leq s\leq n-2$, then $n\geq4$ and $s(n-s)\geq2(n-2)\geq n$.

For JR and wJR, the cases $|T|=1$ and $|T|=n-1$ follow from the PJR result. If $|T|\geq n$, serial dictatorship selects every voter's first choice at least once, so every group is represented under every rank matrix. Conversely, suppose that $2\leq|T|\leq n-2$, and consider the profile~\eqref{eq:cyclic-profile}. Every sequence leaves at least $n-|T|$ candidates unselected. Let $S$ consist of the voters whose first choices are these candidates. By Lemma~\ref{lem:cyclic}, no winner is Pareto efficient for $S$. Moreover, $n\geq4$ and $|S||T|\geq(n-|T|)|T|\geq2(n-2)\geq n$.
Lemma~\ref{lem:all-rank-minimum-representation} therefore yields a violation of all-rank-wJR.
\end{proof}

The two characterizations differ markedly: all-rank-JR can always be satisfied once $|T|\geq n-1$, whereas the existence of an all-rank-PJR sequence depends on the remainder of $|T|$ modulo $n$. The additive loss in Theorem~\ref{thm:all-rank-additive} does not depend on $q$, but can grow linearly with $n$. We next obtain bounds that depend instead on the preferences.

\section{Stronger Guarantees under Preference Restrictions}
\label{sec:preference-structure}
The general results leave two questions: can restrictions on preferences give exact all-rank guarantees, and can they make the available rules more efficient? We consider two restrictions with different effects. Bounding the number of distinct first choices gives an optimal additive guarantee for all-rank-PJR, which is exact with at most three first choices. Single-peaked and single-crossing preferences allow any number of first choices and give all-rank-wPJR for groups that agree throughout the horizon. Both exact guarantees are achieved by polynomial-time semi-online rules. A common single-crossing order also permits an efficient online rule for solid coalitions.

\subsection{Few First Choices}
\label{sec:first-choices}
The unrestricted additive bound depends on the number of voters. We now bound the loss by the number of distinct first choices instead: throughout this subsection, $d$ is a positive integer, and in every round at most $d$ candidates are ranked first by some voter. This restriction concerns only the top of the rankings. There may be arbitrarily many candidates, the rest of each ranking is unrestricted, and neither the rankings nor the set of first choices needs to stay the same over time.

\paragraph{Proportional justified representation.}
By Lemma~\ref{lem:first-choice}, it suffices to select, for every group, enough first choices of its members. A budget interpretation suggests a simple rule. Give each voter $|T|/n$ units of budget, and charge her only when her first choice wins. In each round, select a candidate for which the voters who rank it first have the largest total remaining budget, and collect one unit from these voters, or all they have left if this is less. Since each round collects at most one unit, the amount a group pays is at most the number of rounds in which a member's first choice wins. It therefore suffices to show that the total unspent budget is small.

\begin{algorithm}[htbp]
\caption{First-choice budget rule}
\label{alg:first-choice-budgets}
$b_i\gets |T|/n$ for every $i\in V$\;
\For{$j\gets1$ \KwTo $|T|$}{
    choose $w^j$ maximizing $\sum_{i\in N_c^j(1)}b_i$ over candidates $c$ with $N_c^j(1)\neq\varnothing$, breaking ties by a fixed candidate order\;
    deduct a total of $\min\set{1,\sum_{i\in N_{w^j}^j(1)}b_i}$ from the voters in $N_{w^j}^j(1)$ in the order of their indices, exhausting each budget before moving to the next voter\;
}
\KwRet{$W$}\;
\end{algorithm}

\begin{samepage}
\begin{theorem}
\label{thm:first-choice-loss}
For $d\geq2$, let $\Delta_d=\left\lfloor(d-1)\left(1-\frac1d\right)^{d-1}\right\rfloor$.
If at most $d$ candidates are ranked first in each round, Algorithm~\ref{alg:first-choice-budgets} is a polynomial-time semi-online rule such that, for every group $S$, the winner is the first choice of some member of $S$ in at least $\lfloor|S||T|/n\rfloor-\Delta_d$ rounds. Consequently, it satisfies all-rank-PJR with additive loss $\Delta_d$. No smaller additive loss that depends only on $d$ suffices for all numbers of voters and all horizons, even for an offline rule and all-rank-wPJR.
\end{theorem}
\end{samepage}

\begin{proof}
Suppose that the total remaining budget before a round is $B$. The nonempty sets $N_c^j(1)$ partition the voters into at most $d$ groups, so the group selected by the rule has at least $B/d$ units and pays at least $\min\set{1,B/d}$. The total after the round is therefore at most $B-\min\set{1,B/d}$, which is an increasing function of $B$. Starting from $|T|$ and applying this bound in each of the $|T|$ rounds, the final total is at most
\[
\begin{cases}
 |T|\left(1- 1/d\right)^{|T|}, & |T|\leq d-1,\\[4pt]
 (d-1)\left(1- 1/d\right)^{d-1}, & |T|\geq d-1.
\end{cases}
\]
In the upper-bound recurrence for the second case, the first $|T|-d+1$ steps each subtract one unit, leaving $d-1$, and each of the remaining $d-1$ steps multiplies the total by $1-1/d$. The first case is no larger than the second, since $t(1-1/d)^t$ is nondecreasing for integers $1\leq t\leq d-1$.

Every group $S$ starts with $|S||T|/n$ units and ends with at most the final total. It therefore pays at least $\frac{|S||T|}{n}-(d-1)(1-\frac1d )^{d-1}$.
Only rounds in which a member's first choice wins collect money from $S$, and each collects at most one unit. The number of such rounds is thus at least the ceiling of this expression, which is at least $\lfloor|S||T|/n\rfloor-\Delta_d$. Lemma~\ref{lem:first-choice} then gives all-rank-PJR with additive loss $\Delta_d$. The rule uses only the current first choices and budgets, and each round takes polynomial time.

For the lower bound, fix a horizon $|T|$, and let $C=\set{c_1,\dots,c_d}$. There is one voter for each tuple $(x_1,\dots,x_{|T|})\in[d]^{|T|}$, so $n=d^{|T|}$. In round $j$, this voter has the cyclic ranking starting at $c_{x_j}$. Fix any sequence $W$, and let $S$ consist of the voters whose first choice differs from the winner in every round. Exactly one value of each coordinate is excluded, so $\frac{|S|}{n}=(1-1/d)^{|T|}$.
By Lemma~\ref{lem:cyclic}, the winner is never Pareto efficient for $S$. By Lemma~\ref{lem:all-rank-minimum-representation}, some rank matrix makes $S$ agree in every round without ever representing it. Hence, the additive loss must be at least $\lfloor |T|(1-1/d)^{|T|}\rfloor$, which equals $\Delta_d$ for $|T|=d-1$.
\end{proof}

The loss $\Delta_d$ depends on the number of first choices rather than on the number of voters, and $\Delta_d=d/e+O(1)$. In particular, $\Delta_2=\Delta_3=0$ and $\Delta_4=1$: three first choices per round suffice for exact all-rank-PJR, whereas four do not. Theorem~\ref{thm:first-choice-loss} complements Theorem~\ref{thm:all-rank-additive}, which can give zero loss at particular horizons even when $d$ is large.

\begin{corollary}
\label{cor:three-four-first-choices}
Fix $d$ and $|T|$, and allow the number of voters and the set of candidates to vary. Every profile with at most $d$ first choices per round admits a sequence satisfying all-rank-PJR if and only if $|T|=1$ or $d\leq3$.
The same characterization holds for all-rank-wPJR. In the positive cases, a polynomial-time semi-online rule suffices. For every $|T|\geq2$, the negative case holds on a static profile with four candidates and $n=2|T|$ voters.
\end{corollary}

\begin{proof}
For $d\in\set{2,3}$, the result follows from Theorem~\ref{thm:first-choice-loss}. For $d=1$, select the unanimous first choice in each round. For $|T|=1$, only the whole electorate can be owed a round, and any voter's first choice represents it under every rank matrix.

For the negative case, use the following rankings in every round:
\[
\begin{array}{c|l}
\text{number of voters}&\text{ranking}\\ \hline
2|T|-3 & c_1\succ c_2\succ c_3\succ c_4\\
1 & c_2\succ c_3\succ c_4\succ c_1\\
1 & c_3\succ c_4\succ c_1\succ c_2\\
1 & c_4\succ c_1\succ c_2\succ c_3
\end{array}
\]
All rankings are cyclic, so by Lemmas~\ref{lem:all-rank-minimum-representation} and~\ref{lem:cyclic}, all-rank-wPJR requires every group $S$ to have a member's first choice selected in at least $\lfloor |S||T|/n\rfloor$ rounds. The $2|T|-3$ voters of the first type require $c_1$ in at least $|T|-2$ rounds. Every two of the three remaining voters form a group that requires one round, so at least two of $c_2,c_3,c_4$ must be selected. Since there are only $|T|$ rounds, $c_1$ wins exactly $|T|-2$ times, two of $c_2,c_3,c_4$ win once each, and the remaining candidate never wins. The first type together with the voter whose first choice never wins forms a group of size $2|T|-2$ that requires $|T|-1$ rounds, but it receives only $|T|-2$. This contradicts all-rank-wPJR.
\end{proof}

The positive part of Corollary~\ref{cor:three-four-first-choices} concerns the non-weak axiom: groups need not agree in every round. In Section~\ref{sec:structured}, we obtain exact guarantees without bounding the number of first choices, but only for the weak axiom and for single-peaked or single-crossing profiles. Appendix~\ref{app:three-choice-quota} shows that, with three first choices, a slightly larger initial budget gives a smaller size requirement.

\paragraph{Justified representation.}
Although four first choices can rule out all-rank-PJR, all-rank-JR, which asks for only one represented round, can still be guaranteed at many horizons. To achieve it, it suffices to give priority to voters whose first choice has never been selected. Each round selects the first choice of at least a $1/d$ fraction of these voters, which gives an exact condition for all-rank-JR.

\begin{theorem}
\label{thm:first-choice-jr}
Fix $d$ and $|T|$, and allow the number of voters and the set of candidates to vary. Every profile with at most $d$ first choices per round admits a sequence satisfying all-rank-JR if and only if $|T|(1-1/d)^{|T|}<1$. The same characterization holds for all-rank-wJR. A single polynomial-time online rule, which is given neither $d$ nor $|T|$, achieves all positive cases.
\end{theorem}

\begin{proof}
Maintain the set $U$ of voters whose first choice has not yet been selected; initially, $U=V$. In round $j$, select a candidate $c$ maximizing $|U\cap N_c^j(1)|$ and remove these voters from $U$. If $U$ is empty, select any voter's first choice. Each round removes at least $|U|/d$ voters, so at the end $|U|\leq n(1-1/d)^{|T|}$.
A group that is represented in no round under some rank matrix must be contained in $U$, since a round in which a member's first choice wins represents the group under every rank matrix. If $|T|(1-1/d)^{|T|}<1$, such a group has fewer than $n/|T|$ members, and hence fewer than $n/k$ for every $k\leq|T|$, so it cannot violate all-rank-JR. The rule uses only the current first choices and the set $U$.

If $|T|(1-1/d)^{|T|}\geq1$, then $d\geq2$. The construction in the proof of Theorem~\ref{thm:first-choice-loss} gives, for every sequence, a group of size $n(1-1/d)^{|T|}\geq n/|T|$ that agrees in every round under some rank matrix and is never represented. Thus, even all-rank-wJR cannot be satisfied.
\end{proof}

For $d\leq3$, the condition $|T|(1-1/d)^{|T|}<1$ holds at every horizon, so the online rule always satisfies all-rank-JR. For $d=4$, it holds exactly when $|T|=1$ or $|T|\geq7$. Thus, with four first choices, all-rank-JR can always be satisfied when there are at least seven rounds, whereas all-rank-wPJR cannot always be satisfied for any horizon of at least two rounds.

\subsubsection{Proportionality on Every Time Interval}
\label{sec:time-intervals}
So far, representation has been counted from the first round. A group may also care about a shorter period during which its preferences have repeatedly been passed over. For a contiguous interval $I\subseteq T$, we apply an axiom to the rounds in $I$ alone, using $|I|$ as the horizon and counting both agreement and representation only within $I$. The rule must choose one sequence that meets the requirement for all intervals simultaneously; it is not restarted at the beginning of each interval.

A variant of Algorithm~\ref{alg:first-choice-budgets} achieves this up to an additive loss. In the \emph{replenished} rule, every budget starts at zero, and every voter receives $1/n$ units at the beginning of each round; the rule then selects the winner and charges the voters who rank it first exactly as Algorithm~\ref{alg:first-choice-budgets} does. Over any interval, a group receives new budget in proportion to its size, while its unspent budget remains bounded independently of the length of the interval.

\begin{theorem}
\label{thm:first-choice-intervals}
If at most $d\geq2$ candidates are ranked first in each round, the replenished rule is a polynomial-time online rule such that, for every contiguous interval $I\subseteq T$ and every group $S$, the winner is the first choice of some member of $S$ in at least $\lfloor|I||S|/n\rfloor-(d-2)$ rounds of $I$. Consequently, it satisfies all-rank-PJR with additive loss $d-2$ on every interval. The rule is given neither $d$ nor the horizon.

For two first choices the loss is zero, and for three first choices the loss of one is optimal: zero loss is impossible even for an offline rule and fixed-rank-wJR on a static one-dimensional Euclidean profile. For general $d$, the optimal worst-case loss on intervals is $\Theta(d)$.
\end{theorem}

\begin{proof}
We first show that the total budget after every round is less than $d-1$. This holds initially. If the total after the preceding round is $B<d-1$, replenishment raises it to $B+1<d$. As in the proof of Theorem~\ref{thm:first-choice-loss}, the selected group can pay at least $(B+1)/d$, so the total after payment is at most $(1-1/d )(B+1)<d-1$.

Fix $I$ and $S$. The group starts the interval with a nonnegative budget, receives $|I||S|/n$ units during the interval, and retains less than $d-1$ units at its end. It therefore pays more than $|I||S|/n-(d-1)$ during $I$. These payments occur only in rounds in which a member's first choice wins, at most one unit per round. Since the number of such rounds is an integer, it is at least $\lfloor |I||S|/n\rfloor-(d-2)$, as claimed. Applying Lemma~\ref{lem:first-choice} to the rounds in $I$ gives the all-rank-PJR guarantee. The rule uses only the current first choices and the $n$ budgets. After multiplying all budgets by $n$, they are integers whose total is less than $n(d-1)$, regardless of the number of elapsed rounds, so each round takes polynomial time.

For the lower bound with three first choices, we use a classical scheduling example in which three tasks must occur in every two, three, and six consecutive positions, respectively~\cite{holte1992pinwheel}. Consider a static profile with six voters and three candidates $a,b,c$, where three voters rank $a$ first, two rank $b$ first, and one ranks $c$ first. Placing $a,b,c$ at $0,10,20$ and the three groups of voters at $-1/10$, $99/10$, and $199/10$ gives a one-dimensional Euclidean profile. Let $|T|=8$. With $r=1$, fixed-rank-wJR on every interval requires $a$ in every interval of length two, $b$ in every interval of length three, and $c$ in every interval of length six. Suppose that $c$ wins in some round $j\in\set{2,\dots,7}$. The intervals $\set{j-1,j}$ and $\set{j,j+1}$ force $a$ to win in rounds $j-1$ and $j+1$, and then the interval $\set{j-1,j,j+1}$ contains no $b$, a contradiction. Thus, $c$ can win only in rounds $1$ and $8$, and the interval $\set{2,\dots,7}$ contains no $c$, again a contradiction. Hence, zero loss is impossible for $d=3$.

Finally, the lower bound in Theorem~\ref{thm:first-choice-loss} with $|T|=d-1$ requires a loss of $\Delta_d=d/e+O(1)$ on the full horizon, which is itself an interval. Together with the upper bound of $d-2$, this proves the $\Theta(d)$ bound.
\end{proof}

The lower bound for three first choices uses intervals that do not start in the first round. It therefore leaves open whether, with three first choices, an online rule can satisfy all-rank-PJR exactly at every horizon.

With two candidates and complete rankings, the zero-loss interval guarantee also settles
the positive EJR case identified in Theorem~\ref{thm:ejr}(ii).
Together with the earlier examples, it determines exactly when the weak
EJR axioms can always be satisfied as a function of the number of candidates
and the horizon.

\begin{corollary}
\label{cor:ejr-candidate-horizon}
For complete profiles with at most two candidates, a polynomial-time online
rule satisfies rank-wEJR and fixed-rank-wEJR on every contiguous interval
simultaneously.

More generally, fix $m$ and $|T|$, and allow the number of voters to vary.
Every profile admits a rank-wEJR sequence, and every profile admits a
fixed-rank-wEJR sequence, if and only if $|T|=1$ or $m\leq2$.
Every profile admits an all-rank-wEJR sequence if and only if $|T|=1$ or
$m=1$.
\end{corollary}
\begin{proof}
For two candidates, Theorem~\ref{thm:first-choice-intervals} gives
all-rank-PJR, and hence rank-wPJR and fixed-rank-wPJR, on every interval.
Theorem~\ref{thm:ejr}(ii) gives the corresponding EJR axioms.
With one candidate every voter approves every winner.

At horizon one, select any voter's first choice. Only the whole electorate
can have positive entitlement, and this voter approves the winner under
every cutoff. For $|T|\geq2$, Theorem~\ref{thm:ejr}(i) rules out
universal existence for fixed-rank-wEJR, and hence for rank-wEJR, with three
candidates, and Theorem~\ref{thm:ejr}(iii) does so for all-rank-wEJR with
two. Adding candidates below the displayed rankings preserves both examples
and proves the remaining negative cases.
\end{proof}

The two-candidate assumption concerns complete rankings. Having two distinct
first choices with additional lower-ranked candidates does not suffice:
the profile in Theorem~\ref{thm:ejr}(i) already has only two first choices.
Nor does the equivalence in Theorem~\ref{thm:ejr}(ii) extend
by the same argument to groups that agree in only some rounds.

\subsection{Single-peaked and Single-crossing Profiles}
\label{sec:structured}
The rules above use only the voters' first choices. We now impose structure on the full rankings to obtain exact guarantees without bounding the number of first choices, for groups that agree in every round. When the horizon is known, a budget rule satisfies all-rank-wPJR on single-peaked and single-crossing profiles, even if the axis or voter order changes between rounds. For a common single-crossing order, rank-wPSC can also be satisfied exactly
by a polynomial-time online rule. In contrast, even all-rank-wJR requires
knowledge of the horizon on a static Euclidean profile, and online rules
for all-rank-wPJR must incur a logarithmic additive loss.

The rule for a known horizon gives each voter $|T|+1$ units of budget and spends exactly $n$ units in each round. It collects these payments from the voters in the given single-crossing order, or in the order of their peaks on the given axis, and the last voter to pay chooses the winner. The additional unit per voter gives the Droop size condition introduced in
Section~\ref{sec:preliminaries}: a group with more than $\ell n/(|T|+1)$
members receives at least $\ell$ represented rounds.

\begin{algorithm}[hbpt]
\caption{Ordered budget rule}
\label{alg:ordered-budgets}
$b_i\gets |T|+1$ for every $i\in V$\;
\For{$j\gets1$ \KwTo $|T|$}{
    order the voters by the given single-crossing order of round $j$, or by their peaks on the given axis of round $j$, breaking ties by voter index\;
    deduct a total of $n$ units from the budgets in this order, exhausting each budget before moving to the next voter and stopping as soon as $n$ units have been deducted\;
    $w^j\gets$ the first choice of the last voter whose budget decreased\;
}
\KwRet{$W$}\;
\end{algorithm}

\begin{theorem}
\label{thm:ordered-budgets}
Suppose that every round-$j$ profile $P^j$ is single-peaked with respect to a given axis or single-crossing with respect to a given voter order; the axes and orders may differ across rounds. Algorithm~\ref{alg:ordered-budgets} is a polynomial-time semi-online rule. For every group $S$, the winner is Pareto efficient for $S$ in at least $\lceil (|T|+1)|S|/n \rceil-1$ rounds. In particular, it satisfies all-rank-wPJR, with the stronger Droop size condition $|S|>\ell n/(|T|+1)$.
\end{theorem}

\begin{proof}
Initially, the total budget is $n(|T|+1)$. Spending $n$ units per round is
therefore always possible, and exactly $n$ units remain at the end. The
rule only needs the current order and budgets, so it is semi-online and
runs in polynomial time.

Fix $S$, and consider a round in which the winner is not Pareto efficient
for $S$. Let $c$ be a candidate that all members of $S$ prefer to $w^j$,
and let $i$ be the last voter to pay. Since $w^j$ is her first choice, $i$
prefers $w^j$ to $c$. In a single-crossing round, the voters preferring $c$
to $w^j$ form a prefix or a suffix of the given order that does not contain
$i$. Hence, all members of $S$ are strictly on the same side of $i$.
In a single-peaked round, if $c$ is to the left of $w^j$, every member of $S$
has her peak strictly to the left of $w^j$: a peak at or to its right would
make that voter prefer $w^j$ to $c$. The case in which $c$ is to the right is symmetric.
Thus, the same conclusion holds in the order by peaks.

If $S$ comes after $i$, it pays nothing in this round. If it comes before
$i$, every remaining budget in $S$ is exhausted, and $S$ pays
strictly less than $n$, since $i\notin S$ pays a positive amount. Thus,
at most one round whose winner is not Pareto efficient for $S$ charges $S$ a
positive amount, because all budgets in $S$ are zero after such a round, and this amount is less than $n$.

Let $h$ be the number of rounds whose winners are Pareto efficient for
$S$. These rounds take at most $nh$ units from $S$. If another round
charges $S$, it exhausts all its budgets, and therefore
$(|T|+1)|S|<nh+n$. If no other round charges $S$, its final budget is
at most the total remaining budget $n$, so
$(|T|+1)|S|\leq nh+n$. In either case, $h\geq \lceil (|T|+1)|S|/n \rceil-1$.
Lemma~\ref{lem:all-rank-minimum-representation} converts this bound into
representation under every rank matrix for which $S$ agrees in every round.
Finally,
$\lceil (|T|+1)|S|/n\rceil-1\geq\lfloor |T||S|/n\rfloor$,
which proves all-rank-wPJR.
\end{proof}

Since the axis or voter order may change between rounds, voters need not keep the same positions over time. The restriction to groups that agree in every round cannot be dropped: by Theorem~\ref{thm:rank-impossibility}, even rank-JR, and hence all-rank-JR, can be unattainable on a static one-dimensional Euclidean profile. Theorem~\ref{thm:ordered-budgets} and Corollary~\ref{cor:three-four-first-choices} therefore complement each other: single-peaked and single-crossing preferences allow any number of first choices but give only the weak axiom, whereas the bound on first choices gives the non-weak axiom.

\paragraph{Online proportionality for solid coalitions.}
The general online SCR rule in Section~\ref{sec:weak-rules} may consider exponentially many rank vectors. A common single-crossing order makes it possible to satisfy rank-wPSC in polynomial time online.
If two voters have the same $r$ most-preferred candidates in a round, then so does every voter between them in the order. It therefore suffices to consider intervals of the voter order, keeping for each interval the smallest common approval set in each round.

\begin{theorem}
\label{thm:single-crossing-online-psc}
If every round is single-crossing with respect to the same given voter
order, then rank-wPSC can be satisfied by a polynomial-time online rule,
also with the Droop size condition.
\end{theorem}

\begin{proof}
Use the refinement loop of Algorithm~\ref{alg:online-scr}, with the following
collection in place of $\mathsf{Claims}^j$. For every nonempty interval $S$
of the given voter order, keep one rank vector $\mathbf r$ whose entry
$r^t$ is the smallest cutoff for which all members of $S$ have the same
approval set $A_i^t(r^t)$. Such a cutoff always exists because cutoff $m$ gives $C$.
The representation count and priority are defined as before. In each round,
compute the new entry and, after selecting the winner, update the count.
Singleton intervals ensure that every refinement with $|X|>1$ is possible,
by the same argument as in Theorem~\ref{thm:online-scr}.

To see why this smaller collection suffices, fix any group $S$ that is solid
under a rank vector $\mathbf r$ up to the current round. For a set $C'\subseteq C$, the voters whose $|C'|$ most-preferred candidates in round $t$ form the set $C'$ are $\bigcap_{a\in C',\ b\notin C'}\{i\in V:a\succ_i^t b\}$.
By single-crossing, each set in this intersection is a prefix or a suffix of the voter order, so their intersection is an interval. Thus, every voter between the first and last members of $S$ has the same $r^t$ most-preferred candidates in round $t$ as the members of $S$, for every $t\leq j$.
Let $\widehat S$ be that interval, and let $\widehat{\mathbf r}$ be its kept vector. Then, in every round $t\leq j$, its common approval set is contained in the common approval set of $(S,\mathbf r)$. Consequently,
\[
 \mathrm{rep}^{j-1}(\widehat S,\widehat{\mathbf r})
 \leq \mathrm{rep}^{j-1}(S,\mathbf r),
 \quad
 \pi^{j-1}(\widehat S,\widehat{\mathbf r})
 \geq\pi^{j-1}(S,\mathbf r).
\]
In particular, whenever the common approval set of $(S,\mathbf r)$ is a
proper subset of $X$, the claim $(\widehat S,\widehat{\mathbf r})$ is eligible and has at
least its priority. Apply the payment argument of Theorem~\ref{thm:online-scr},
with $p=n/|T|$ and its invariant for every solid group. In the first case, an
unrepresented group of priority at least $p$ intersecting the charged
group would give, through its interval $\widehat S$, an eligible next
refinement of priority at least $p$. This contradicts the choice of the
last such refinement. In the second case, any unrepresented group of
priority at least $p$ would give an eligible first refinement of at least
that priority. Thus both cases preserve the invariant, and the same final
budget argument proves rank-wPSC. Using $p=n/(|T|+1)$ and the final-budget argument of Corollary~\ref{cor:scr-droop} gives the Droop size condition.

There are $n(n+1)/2$ intervals and at most $m-1$ refinements per round.
Each smallest common approval set can be found by checking the $m$ possible
cutoffs. It suffices to store the common approval set in the current round and the
representation count for each interval; no earlier rank vectors need be
enumerated. All calculations take polynomial time, and none uses $|T|$.
\end{proof}

\paragraph{Online all-rank proportionality.}
The ordered budget rule uses the horizon to set the initial budgets. This
information is necessary even for all-rank-wJR on a static Euclidean
profile. The following example uses only three possible horizons, and the
preferences do not change between them.

\begin{theorem}
\label{thm:static-euclidean-online-jr}
No online rule always satisfies all-rank-wJR on static strict
one-dimensional Euclidean profiles. This holds with fifteen voters,
fifteen candidates, a common candidate axis, and a common single-crossing
voter order.
\end{theorem}
\begin{proof}
Place candidate $c_i$ at $i$ and voter $i$ at $i-1/10$, for $i\in[15]$,
and use the resulting strict rankings in every round. For a consecutive
group of voters $\set{a,\dots,b}$, the Pareto-efficient candidates are
exactly $c_a,\dots,c_b$. Every candidate in this set is a member's first
choice, and each candidate outside it is Pareto dominated by the nearer
endpoint candidate.

At horizon $t$, any block of at least $\lceil15/t\rceil$ consecutive
unselected candidates gives a group that violates all-rank-wJR, by
Lemma~\ref{lem:all-rank-minimum-representation}. Thus, at horizons $3$, $4$,
and $5$, every block of consecutive unselected candidates must have length
at most $4$, $3$, and $2$, respectively.

After three rounds, at least twelve candidates remain unselected, in at
most four such blocks. Add blocks of length zero if there are fewer than
four. Each length is at most four. Only one additional candidate can be
selected by horizon four, so at most one of these original blocks can have
length greater than three. Only two additional candidates can be selected
by horizon five, so at most two original blocks can have length greater
than two. Ordering the original lengths from largest to smallest, their
sum is therefore at most $4+3+2+2=11$, contrary to the at least twelve unselected candidates. Repeated winners
cannot improve this bound.

The preferences are identical at all three possible horizons. An online
rule must therefore make the same first three choices in each instance,
and the same first four choices in the latter two. It cannot meet all
three requirements.
\end{proof}

The same example rules out an exact guarantee with probability one for a
rule that uses random choices. Use the same random choices when comparing
the three horizons. For every resulting sequence, at least one of horizons
$3$, $4$, and $5$ violates all-rank-wJR. The sum of the three violation
probabilities is therefore at least one, so at one fixed horizon the
probability of a violation is at least $1/3$.

The difficulty here comes entirely from not knowing when the decisions
will end, rather than from changes in future preferences. In contrast,
Theorem~\ref{thm:ordered-budgets} achieves all-rank-wPJR on the same domain
when the horizon is known. We next quantify the loss without this information.
On a common single-crossing order, the optimal worst-case additive loss for
all-rank-wPJR has logarithmic order in the number of voters. The rule for
the upper bound selects first choices of voters so that every interval of
the voter order chooses approximately in proportion to its size.

\begin{theorem}
\label{thm:unknown-horizon}
Suppose that every round is single-crossing with respect to the same given
voter order. There is a polynomial-time online rule that, at every horizon, satisfies
all-rank-wPJR with additive loss $2\lceil\log_2 n\rceil$. No deterministic online rule can replace this additive loss by
$o(\log n)$ for all such profiles. The lower bound holds on static strict
one-dimensional Euclidean profiles with $m=n$. Thus, the optimal worst-case
additive loss for all-rank-wPJR on this domain is $\Theta(\log n)$.
\end{theorem}

\begin{proof}
For the upper bound, number voters according to the given order. Build a
binary tree whose leaves, from left to right, are $1,\dots,n$, splitting
each interval of leaves into two parts whose sizes differ by at most one.
In every round, start at the root and proceed to a leaf. At a node with
$s$ leaves, of which $s_L$ are in its left child, send the $t$-th visit to
the left child exactly when $\lfloor t s_L /s \rfloor
 > \lfloor (t-1)s_L/s\rfloor$, otherwise, send it to the right child. Let $v^j$ be the leaf reached in
round $j$, and select voter $v^j$'s first choice. The visit counts and choices
do not depend on the horizon.

After $t$ visits to a node, its left child has received exactly
$\lfloor t s_L/s\rfloor$ visits. We claim that, within a subtree of height
$H$, the number of visits to any prefix of its leaves differs from its
proportional share of the subtree's visits by at most $H$. This follows by
induction on $H$. At a leaf the error is zero. Consider a node with $t$
visits, and call $\lfloor ts_L/s\rfloor-ts_L/s$ its rounding error; its
absolute value is less than one. For a prefix of $a\leq s_L$ leaves, the
proportional share of the left child's visits, $\lfloor ts_L/s\rfloor a/s_L$,
differs from $ta/s$ by the rounding error multiplied by $a/s_L$. For a prefix
of $a>s_L$ leaves, the left child's visits plus the proportional share of the
right child's visits,
$\lfloor ts_L/s\rfloor+(t-\lfloor ts_L/s\rfloor)(a-s_L)/(s-s_L)$, differ from
$ta/s$ by the rounding error multiplied by $(s-a)/(s-s_L)$. Both factors are
in $[0,1]$. Only one child
contains part but not all of the prefix, so its inductive error plus the
current rounding error is at most $H$. Since the tree has height
$\lceil\log_2 n\rceil$, at every horizon
\[
 \left|
   |\{j\in T:v^j\leq a\}|-\frac{|T|a}{n}
 \right|\leq\lceil\log_2 n\rceil
 \quad \text{for all } 0\leq a\leq n.
\]
Subtracting the counts for two prefixes shows that, for every interval of
voters, the number of rounds $j$ in which $v^j$ lies in this interval is at
least $|T|/n$ times its length minus $2\lceil\log_2 n\rceil$.

Fix a group $S$ and take the interval from its first to its last
member. Whenever $v^j$ lies in this interval, her first choice is Pareto
efficient for $S$. Indeed, a candidate preferred to it by all of $S$ would
be preferred by both extreme members; single-crossing would then imply
that voter $v^j$ also prefers it, a contradiction. The interval contains
at least $|S|$ voters, so the bound on visits, together with
Lemma~\ref{lem:all-rank-minimum-representation}, proves the desired
representation guarantee. Each round needs only $O(1+\log n)$ visits in the
tree and the current chosen voter's first choice, so the rule is polynomial.

For the lower bound, place candidate $c_i$ at $i$ and voter $i$ at
$i-1/10$, and repeat these rankings in every round. The rankings are strict
and Euclidean. For every interval $S=\{a,\dots,b\}$, its Pareto-efficient
candidates are exactly $\{c_a,\dots,c_b\}$. Candidates outside this interval
are unanimously worse than the nearer extreme candidate; for a candidate
inside it, the two extreme voters rule out a unanimous improvement.

Run an arbitrary deterministic online rule for $n$ rounds on this profile,
and write $w^j=c_{x_j}$. Apply Schmidt's discrepancy theorem
\cite{schmidt1972irregularities} to the $n$ points
$((j-1/2)/n,(x_j-1/2)/n)$, for $j\in[n]$. It gives an axis-parallel
rectangle with lower-left corner $(0,0)$ whose number of points differs
from $n$ times its area by $\Omega(\log n)$. Move each upper endpoint
to a nearest multiple of $1/n$ that preserves the point count. Each move
has length at most $1/(2n)$, so $n$ times the area changes by at most one.
Thus, for an absolute
constant $c_0>0$, some $1\leq\tau\leq n$ and $0\leq a\leq n$ satisfy
\[
 \left|
   |\{j\leq\tau:x_j\leq a\}|-\frac{\tau a}{n}
 \right|\geq c_0\log n-1.
\]
For sufficiently large $n$, the right-hand side is positive, so
$1\leq a\leq n-1$. If the count is below $\tau a/n$, take
$S=\{1,\dots,a\}$. Otherwise, take $S=\{a+1,\dots,n\}$; its count is
below its proportional share by the same amount. In either case, at the
horizon $|T|=\tau$, the number of rounds in which the winner is Pareto
efficient for $S$ is at least $c_0\log n-2$ below $\lfloor |S||T|/n\rfloor$.
Lemma~\ref{lem:all-rank-minimum-representation} gives a rank matrix under
which $S$ agrees in every round and receives exactly this representation.
The rule is online, so its first $\tau$ choices are unchanged when the
instance ends at $\tau$ rather than $n$. This proves the lower bound.
\end{proof}

On a common single-crossing order, Theorems~\ref{thm:single-crossing-online-psc},
\ref{thm:ordered-budgets}, and~\ref{thm:unknown-horizon} distinguish three cases: rank-wPSC can be satisfied exactly online, all-rank-wPJR can be satisfied exactly semi-online, and online rules for all-rank-wPJR must incur a logarithmic additive loss. Theorem~\ref{thm:static-euclidean-online-jr} further shows that the need
for the horizon already appears at the level of weak JR. In contrast,
Theorem~\ref{thm:first-choice-intervals} needs no common voter order and
applies to every time interval, with a loss that depends on the number of
first choices rather than on the number of voters.

\section{Verifying Proportionality}
\label{sec:verification}
A rule may be chosen for reasons other than proportionality, or a sequence
may be proposed directly by participants. In either case, one can ask
whether the resulting sequence meets a proportionality requirement.
For an axiom, the \emph{verification problem} takes a profile $P$ and a
sequence $W$ and asks whether $W$ satisfies that axiom. We consider complete
rankings in this section and truncated ballots in the next.

Table~\ref{tab:verification-complete} gives the complexity of verifying the
weak axioms on unrestricted complete profiles. Fixed-rank-wPSC is always
polynomial-time verifiable. For fixed-rank-wJR, hardness begins with four
candidates; for every other defined entry, it begins with three. Preference
restrictions lead to further distinctions. On three-candidate single-peaked
profiles, verification of weak JR and PJR is easy for the fixed-rank and
all-rank versions, but hard for the rank versions.
For single-crossing profiles, a common voter order permits polynomial-time
verification of all four all-rank JR and PJR axioms, whereas changing orders
make all four problems coNP-complete.

\begin{table}[!ht]
\centering
\small
\renewcommand{\arraystretch}{1.25}
\begin{tabular}{@{}p{0.145\linewidth}p{0.25\linewidth}p{0.25\linewidth}p{0.25\linewidth}@{}}
\hline
\textbf{Cutoffs}&\textbf{wJR}&\textbf{wPJR}&\textbf{wPSC}\\
\hline
Fixed-rank&Polynomial: $m\leq3$\newline coNP-complete: $m\geq4$
 &Polynomial: $m\leq2$\newline coNP-complete: $m\geq3$
 &Polynomial for all $m$\\
Rank&Polynomial: $m\leq2$\newline coNP-complete: $m\geq3$
 &Polynomial: $m\leq2$\newline coNP-complete: $m\geq3$
 &Polynomial: $m\leq2$\newline coNP-complete: $m\geq3$\\
All-rank&Polynomial: $m\leq2$\newline coNP-complete: $m\geq3$
 &Polynomial: $m\leq2$\newline coNP-complete: $m\geq3$&---\\
\hline
\end{tabular}
\caption{Verification of the weak axioms on unrestricted complete rankings.
The fixed-rank row is Theorem~\ref{thm:fixed-rank-verification}; the rank row
follows from Theorems~\ref{thm:three-candidate-verification}
and~\ref{thm:two-candidate-verification}; and the all-rank row follows from
Theorems~\ref{thm:two-candidate-verification}
and~\ref{thm:varying-sc-all-rank-verification}. No all-rank PSC axiom is defined.}
\label{tab:verification-complete}
\end{table}

\subsection{Fixed Cutoffs}
\label{sec:fixed-verification}
We begin with the fixed-rank axioms, whose cutoff is common to all voters
and rounds. Solid
coalitions are particularly easy to check: for a fixed cutoff, voters who
are solid together throughout the horizon have the same sequence of approval
sets. Agreement allows different approval sets within a group, and the
complexity of verification then depends on the number of candidates.

\begin{theorem}
\label{thm:fixed-rank-verification}
For complete rankings, fixed-rank-wPSC can be verified in polynomial time
for every number of candidates. Fixed-rank-wJR can be verified in polynomial
time for $m\leq3$ and is coNP-complete to verify for every fixed $m\geq4$.
Fixed-rank-wPJR can be verified in polynomial time for $m\leq2$ and is
coNP-complete to verify for every fixed $m\geq3$.

The three-candidate hardness result for fixed-rank-wPJR holds even when
each round is single-crossing with respect to a given voter order, at most
two candidates are ranked first in each round, and the proposed winner is
constant. The order may differ across rounds.
\end{theorem}
\begin{proof}
We first prove the polynomial-time results. Fix a cutoff $r$, and partition
the voters according to their sequences
$(A_i^1(r),\dots,A_i^{|T|}(r))$. A group that is solid in every round is
contained in one class. Enlarging the group to this class preserves its
representation and can only increase its entitlement. Thus, checking at
most $n$ classes for each of the $m$ cutoffs verifies fixed-rank-wPSC.

At cutoff one, agreement also means identical approval sets. The same
classes therefore suffice for both wJR and wPJR at this cutoff. At cutoff
$m$, every voter approves every winner. These observations give polynomial
verification for both axioms with at most two candidates. With three
candidates, only cutoff two remains to be considered for wJR. A group that
is never represented at this cutoff contains only voters who rank the
winner last in every round. Conversely, all such voters together approve
exactly $C\setminus\set{w^j}$ in round $j$, so they agree in every round
and are never represented. It suffices to check whether there are at least
$\lceil n/|T|\rceil$ such voters.

Both verification problems are in coNP: a cutoff and a group violating the
axiom can be checked in polynomial time. We prove hardness by two reductions
from \textsc{Independent Set}~\cite{karp1972reducibility}.

\emph{Fixed-rank-wJR with four candidates.}
Let $(G_0,K_0)$ be an instance with $N_0\geq2$ vertices and
$1\leq K_0\leq N_0$. Add $N_0$ vertices with no incident edges to obtain a
graph $G$ with $N=2N_0$ vertices, and put $K=N_0+K_0\geq3$.
The graph $G$ has an independent set of size at least $K$ if and only if
$G_0$ has one of size at least $K_0$.

Set $h=\max\{2,|E(G)|\}$ and use $h$ rounds and $n=hK$ voters.
There is one vertex voter for each vertex; the remaining $n-N$ voters are
dummy voters, and $n-N\geq 2K-2N_0>0$. Let $C=\set{a,b,c,w}$, and propose $w$ in every round.
For each edge, choose an orientation $u\to v$ and create a round with
rankings
\[
\begin{array}{c|l}
\text{voters}&\text{ranking}\\\hline
u&a\succ b\succ w\succ c\\
v&c\succ b\succ w\succ a\\
\text{other vertex voters}&a\succ c\succ w\succ b.
\end{array}
\]
In each of the remaining $h-|E(G)|$ rounds, every vertex voter has the
third displayed ranking.
Dummy voters rank $w$ first in every round, with the remaining candidates
in any fixed order.

Consider a vertex group $S$ with $|S|\geq K\geq3$. At cutoff two, it
agrees in an edge round unless it contains both endpoints. Indeed, the
three possible approval sets are $\set{a,b}$, $\set{b,c}$, and
$\set{a,c}$. If both endpoints belong to $S$, a third member has the last
set, and the intersection is empty. If at most one endpoint belongs to
$S$, the intersection is nonempty. Thus $S$ agrees throughout at cutoff
two exactly when it is independent, and it is never represented.
The wJR size requirement is $n/h=K$.

No other cutoff introduces additional violations. A group agreeing
throughout at cutoff one also agrees throughout at cutoff two, where the
winner is still excluded by every vertex voter. At cutoffs three and four,
every vertex voter approves $w$. Any group containing a dummy voter is
represented in every round. Hence fixed-rank-wJR is violated exactly when
$G_0$ has the required independent set.

\emph{Fixed-rank-wPJR with three candidates.}
Here the reduction must distinguish a shortfall of one represented round
from a complete absence of representation. Let $(G_0,K_0)$ have $k\geq3$
vertices and $1\leq K_0\leq k$. Add $2k-K_0$ vertices with no incident
edges, producing a graph $G$ with $N=3k-K_0$ vertices and independent-set
target $2k$. Create one voter per vertex and $K_0$ dummy voters, so $n=3k$.
Let $D$ be one more than the largest vertex degree in $G$, put
$h=Dn+3$, and use $h$ rounds. There are three candidates $a,b,c$, and $c$ is the
proposed winner in every round. The following are the three ranking types:
\[
 \mathrm{I}:\ c\succ a\succ b,\quad
 \mathrm{II}:\ c\succ b\succ a,\quad
 \mathrm{III}:\ b\succ a\succ c.
\]
Dummy voters have type $\mathrm{I}$ in every round. For each oriented edge $u\to v$,
create a round in which $u$ has type $\mathrm{I}$, $v$ has type $\mathrm{II}$, and all other
vertex voters have type $\mathrm{III}$. For each vertex $u$, create
$D-\deg_G(u)\geq1$ further rounds in which $u$ has type $\mathrm{I}$ and every
other vertex voter has type $\mathrm{III}$. Add one round in which every voter has
type $\mathrm{I}$. In the remaining rounds, every vertex voter has type $\mathrm{III}$.
The number of rounds specified before this last step is
$DN-|E(G)|+1<h$, so the construction is well defined and polynomial in size.

Every round is single-crossing in the order of types $\mathrm{I},\mathrm{II},\mathrm{III}$, omitting
empty types. In an edge round, place the dummy voters and $u$ first,
then $v$, then the remaining vertices. The other rounds use only $\mathrm{I}$ and
$\mathrm{III}$, or only $\mathrm{I}$. There are at most two first choices, $b$ and $c$.

At cutoff two, a vertex group of size at least three agrees in every round
if and only if it is independent. To see this, the approval sets in an edge
round are $\set{c,a}$, $\set{c,b}$, and $\set{a,b}$. Their intersection
is empty precisely when the group contains both endpoints and at least one
other vertex. All non-edge rounds permit agreement for every vertex group.

For a vertex group $S$, let $E(G[S])$ be the set of edges whose two
endpoints belong to $S$. Writing $s=|S|$, its number of represented rounds
at cutoff two is $q(S)=Ds-|E(G[S])|+1$.
Each vertex is assigned type $\mathrm{I}$ or $\mathrm{II}$ in exactly $D$ edge or individual
rounds; an edge with both endpoints in $S$ is counted twice in this total
but represents $S$ only once. The round in which all voters have type $\mathrm{I}$
adds one. The group's entitlement is $\left\lfloor\frac{hs}{n}\right\rfloor
   =Ds+\left\lfloor\frac{s}{k}\right\rfloor$.
For an independent group, $q(S)=Ds+1$, so its entitlement is violated
exactly when $s\geq2k$. Groups of size at most two do not violate it:
$\lfloor s/k\rfloor=0$ and $|E(G[S])|\leq1$ imply $q(S)\geq Ds$.
Groups of size at least three that are not independent do not agree
throughout and impose no weak requirement at this cutoff.

It remains to check the other cutoffs. At cutoff one, distinct vertex voters
have distinct first-choice sequences, because each vertex has a round in
which it alone among the vertices has type $\mathrm{I}$. Each vertex voter has its
first choice selected in $D+1$ rounds, exceeding its singleton entitlement
$D$. Dummy voters have their first choice selected in every round; their
sequence differs from every vertex sequence in the added type-$\mathrm{III}$ rounds.
Thus all groups that agree throughout at cutoff one meet their entitlements.
Cutoff three is automatic. A group containing a dummy voter is represented
throughout at every cutoff. We have therefore obtained a violation of
fixed-rank-wPJR exactly when $G$ has an independent set of size at least
$2k$, or equivalently when $G_0$ has one of size at least $K_0$.

For either reduction, adding candidates at the bottom of every ranking
preserves the argument and proves hardness for each larger fixed number
of candidates.
\end{proof}

\subsection{Common Cutoffs That Vary across Rounds}
Allowing the common cutoff to change between rounds, as in the rank axioms,
strengthens the weak axioms. Their verification becomes hard already with
three candidates, even on Euclidean profiles. The same construction proves hardness for wJR,
wPJR, and wPSC.

\begin{theorem}
\label{thm:three-candidate-verification}
    Verifying rank-wJR, rank-wPJR, or rank-wPSC is coNP-complete, even with
three candidates at fixed locations on the real line, strict Euclidean
preferences in every round, and the same winner in every round.
\end{theorem}
\begin{proof}
    The complement of each problem is in NP. A group $S\subseteq V$ and a rank
    vector $\mathbf r\in[m]^{|T|}$ can be guessed, after which agreement, solidity, the number of represented rounds, and the proportional
    requirement can all be checked in polynomial time.
    
    For coNP-hardness, we reduce from \textsc{Independent Set}
    \cite{karp1972reducibility}. Let $(G_0,K_0)$ be an instance in which $G_0$ has
    $N_0$ vertices and $1\leq K_0\leq N_0$. Add $N_0$ new vertices with no
    incident edges to obtain $G$. Thus, $G$ has $2N_0$ vertices. Let $K=N_0+K_0$.
    Then $G$ has an independent set of size at least $K$ if and only if $G_0$ has
    an independent set of size at least $K_0$.
    
    Create one round for every edge of $G$, and add rounds if necessary so that     $|T|\geq 2$. There is one vertex voter for every vertex of $G$, and there are     $n-2N_0$ dummy voters, where $n=K|T|$.
    The number of dummy voters is nonnegative because $K\geq N_0+1$ and
    $|T|\geq 2$.
    
    Let $C=\set{a,b,w}$ and let $W=(w,\dots,w)$. Place $b,a,w$ at positions
    $0,1,3$, respectively. In the round corresponding to an edge $\set{u,v}$,
    orient the edge arbitrarily and place the voters as follows:
    \[
    \begin{array}{c|c|c}
    \text{voter type} & \text{position} & \text{ranking} \\ \hline
    u & 7/4 & a\succ w\succ b \\
    v & 0   & b\succ a\succ w \\
    \text{every other vertex voter} & 1 & a\succ b\succ w \\
    \text{every dummy voter} & 3 & w\succ a\succ b
    \end{array}
    \]
    In every added round, place every vertex voter at $1$ and every dummy voter at
    $3$. All rankings are strict, and the candidate locations are the same in
    every round.
    
    Fix a group $S$ containing only vertex voters and consider the round for
    $\set{u,v}$. If $u\in S$ and $v\notin S$, then $r^j=1$ gives the common
    approved set $\set{a}$ and excludes $w$. If $v\in S$ and $u\notin S$, then
    $r^j=2$ gives the common approved set $\set{a,b}$ and excludes $w$. If neither
    endpoint belongs to $S$, then $r^j=1$ again works. If both endpoints belong to
    $S$, then at cutoff $1$ the intersection is empty and the approval sets of $u$
    and $v$ differ; at cutoff $2$ voter $u$ approves $w$ and the approval sets of
    $u$ and $v$ again differ; and at cutoff $3$ every voter approves $w$. Thus, no
    cutoff gives agreement without representation, and no cutoff makes $S$ solid
    without representation. In every added round, cutoff $1$
    gives the common approved set $\set{a}$ and excludes $w$.
    
    It follows that a group of vertex voters admits a rank vector under which it
    agrees in every round and is represented in no round if and only if it is an
    independent set in $G$. The same equivalence holds when agreement is replaced
    by solidity.
    Any group containing a dummy voter is represented in every round because each
    dummy voter ranks $w$ first.
    
    For every group $S$ of vertex voters, $\lfloor\frac{|S||T|}{n}\rfloor = \lfloor\frac{|S|}{K}\rfloor$.
    Since $|S|\leq 2N_0<2K$, this number is positive exactly when $|S|\geq K$, and
    then it equals $1$. The size requirement of rank-wJR is also
    $n/|T|=K$. Therefore, a violation of rank-wJR, rank-wPJR, or rank-wPSC
    exists exactly when $G$ has an independent set of size at least $K$. This
    proves coNP-hardness for all three verification problems. Appending
    further candidates at the bottom of every ranking preserves the argument:
    the original cutoffs are unchanged, and every additional cutoff represents
    every group.
\end{proof}

Three candidates are necessary for this hardness result. With two candidates, a group can agree in a round without being represented exactly when none of its members ranks the winner first. This reduces verification to the following polynomial-time calculations.

\begin{theorem} \label{thm:two-candidate-verification}
    If $m\leq2$, then rank-wJR, rank-wPJR, rank-wPSC, all-rank-JR,
all-rank-PJR, all-rank-wJR, and all-rank-wPJR can be verified in polynomial time.
\end{theorem}
\begin{proof}
The case $m=1$ is immediate, so suppose that $m=2$. For every round $j\in T$,
let $H^j=\set{i\in V:\mathrm{rank}_i^j(w^j)=1}$.
Fix a group $S$. If $S\cap H^j=\varnothing$, all voters in $S$ rank the other
candidate first, so cutoff $1$ gives agreement, gives identical approval sets,
and excludes $w^j$. If $S\cap H^j\neq\varnothing$, then either every voter in
$S$ ranks $w^j$ first, or $S$ contains both possible first-choice types. In
the first case every cutoff represents $S$. In the second case, at cutoff $1$ the
intersection is empty and the approval sets differ, while cutoff $2$ represents
$S$. Consequently, the minimum number of represented rounds, over rank
vectors under which $S$ agrees in every round, is $|\set{j\in T:S\cap H^j\neq\varnothing}|$,
and the same statement holds for rank vectors under which $S$ is solid in every round.

Thus, rank-wPJR and rank-wPSC are both satisfied exactly when every group
$S$ satisfies $|\set{j\in T:S\cap H^j\neq\varnothing}|
\geq \lfloor |S||T|/n\rfloor$.
A violating group exists exactly when
\begin{equation}
|T||S|-n\left|\set{j\in T:S\cap H^j\neq\varnothing}\right|\geq n.
\label{eq:two-candidate-violation}
\end{equation}
To check this inequality, construct a directed network with a source, a sink,
one node for each voter, and one node for each round. Add an arc of capacity
$|T|$ from the source to every voter node, an arc of capacity $n$ from every
round node to the sink, and an arc of capacity $n|T|+1$ from voter $i$ to round
$j$ whenever $i\in H^j$.

No minimum cut uses an arc of capacity $n|T|+1$, because the cut that places all voter and round nodes on the sink side has capacity $n|T|$. If $S$ is the set of voter nodes on the source side, the round nodes on the source side are exactly the rounds $j$ with $S\cap H^j\neq\varnothing$: these rounds
are forced to that side, and moving any other round node to the sink side strictly decreases the capacity. Conversely, every $S$ defines a cut by placing precisely those round nodes on the source side. The cut capacity is therefore $|T|(n-|S|)+n\left|\set{j\in T:S\cap H^j\neq\varnothing}\right|$.
Hence, a group satisfying~\eqref{eq:two-candidate-violation} exists if and
only if the minimum cut has capacity at most $n|T|-n$.

For rank-wJR, a group that agrees in every round and is represented in no
round can contain only voters who never rank $w^j$ first. Conversely, all such
voters agree at cutoff $1$ in every round and are never represented. It is
therefore enough to count these voters and check whether their number is at
least $\lceil n/|T|\rceil$.

For the all-rank axioms, a two-candidate winner is Pareto efficient for $S$
exactly when it is a member's first choice. The weak characterization in
Lemma~\ref{lem:all-rank-minimum-representation} therefore gives the same
minimum representation as above. Moreover,
$D_W(S)+F_W(S)=|T|$, so Theorem~\ref{thm:all-rank-characterization} gives
the same tests for the non-weak all-rank axioms.
\end{proof}

\subsection{Individual Cutoffs on Single-peaked Profiles}
For individual cutoffs, Lemma~\ref{lem:all-rank-minimum-representation} gives a different description: only the rounds in which the winner is Pareto efficient for a group matter. On single-peaked preferences, these winners lie between the group's extreme peaks. This permits polynomial-time verification even with arbitrarily many candidates.

\begin{theorem}
\label{thm:single-peaked-verification}
Suppose that, for every round $j\in T$, $P^j$ is single-peaked with respect to
a given candidate axis. The axes may differ across rounds. Then
all-rank-wJR and all-rank-wPJR can be verified in polynomial time.
\end{theorem}

\begin{proof}
In a single-peaked profile, the Pareto-efficient candidates for a group $S$
are exactly those between its leftmost and rightmost peaks, including both
endpoints. A candidate outside this interval is unanimously worse than the
nearer extreme peak. For a candidate inside the interval, any alternative
on its left is ranked lower by a voter with the rightmost peak, and any
alternative on its right is ranked lower by a voter with the leftmost peak.

We first consider all-rank-wPJR. By
Lemma~\ref{lem:all-rank-minimum-representation}, an all-rank-wPJR violation
exists exactly when some group $S$ satisfies
\begin{equation}
|T||S|-n\left|\set{j\in T:w^j\text{ is Pareto efficient in round }j
\text{ for }S}\right|\geq n.
\label{eq:single-peaked-violation}
\end{equation}

Fix a voter $i\in V$. Construct a directed network with a source, a sink,
one node for each voter, and one node for each round. Add an arc of capacity
$n|T|+1$ from the source to voter $i$, an arc of capacity $|T|$ from the source
to every other voter, and an arc of capacity $n$ from every round node to the
sink. In addition, add an arc of capacity $n|T|+1$ from voter $i'$ to round $j$
whenever $w^j$ lies between the peaks of $i$ and $i'$ on the candidate axis in
round $j$, including the two endpoints.

The cut that places all voter and round nodes on the source side has capacity
$n|T|$, so a minimum cut does not use an arc of capacity $n|T|+1$. Its
source-side voter nodes form a group $S$ containing $i$. Since $S$ contains
$i$, the winner $w^j$ lies between the leftmost and rightmost peaks of $S$ if
and only if at least one source-side voter node has an arc of capacity
$n|T|+1$ to round $j$. Hence, the source-side round nodes are exactly the
rounds in which $w^j$ is Pareto efficient for $S$: including any additional
round node would strictly increase the capacity. Conversely, every group
containing $i$ defines a cut with exactly these round nodes. Its capacity is
\[
|T|(n-|S|)+n\left|\set{j\in T:w^j\text{ is Pareto efficient in round }j
\text{ for }S}\right|.
\]
Thus, a group containing $i$ satisfies
\eqref{eq:single-peaked-violation} if and only if this minimum cut has
capacity at most $n|T|-n$. Checking all $i\in V$ takes $n$ minimum-cut
computations and decides all-rank-wPJR.

For all-rank-wJR, Lemma~\ref{lem:all-rank-minimum-representation} shows that a
violation exists exactly when there is a group of at least
$\lceil n/|T|\rceil$ voters for which $w^j$ is not Pareto efficient in any
round. Partition the voters into classes according to whether, in each round
$j$, their peak lies to the left of $w^j$, at $w^j$, or to the right of $w^j$.
A group has no round in which the winner is Pareto efficient for it exactly
when, in every round, all its members have their peaks strictly on the same side of $w^j$,
that is, when the group is contained in a class whose members never have their
peak at the winner. Hence, a violation exists exactly when such a class has
size at least $\lceil n/|T|\rceil$.
Forming and counting these classes takes polynomial time.
\end{proof}

The three-candidate single-peaked case also admits a positive result for
fixed cutoffs. Unlike individual cutoffs, a fixed cutoff does not give the
Pareto characterization in general. With three candidates, however, the
middle candidate on the axis makes agreement automatic at cutoff two.

\begin{corollary}
\label{cor:three-candidate-sp-fixed-verification}
On complete three-candidate profiles that are single-peaked in every round,
fixed-rank-wJR, fixed-rank-wPJR, and fixed-rank-wPSC can be verified in
polynomial time. The axes may differ across rounds.
\end{corollary}
\begin{proof}
Theorem~\ref{thm:fixed-rank-verification} already gives the claims for wJR
and wPSC, and handles cutoffs one and three for wPJR. At cutoff two, every
voter approves the middle candidate on the round's axis, so every group
agrees in every round. Let $H^j=\set{i:w^j\in A_i^j(2)}$.
A group's representation is $|\set{j:S\cap H^j\neq\varnothing}|$.
The minimum-cut construction in the proof of
Theorem~\ref{thm:two-candidate-verification}, with these sets $H^j$,
checks whether any group has representation below
$\lfloor |T||S|/n\rfloor$.
\end{proof}

Thus, on the same three-candidate Euclidean domain, verification of weak
JR and PJR is polynomial for the fixed-rank versions, coNP-complete for the
rank versions, and polynomial again for the all-rank versions.
The implications among the axioms do not determine these computational
relationships. In the reduction of
Theorem~\ref{thm:three-candidate-verification}, a group containing both
endpoints of an edge cannot agree in its round without being represented
when its members must use the same cutoff in that round, as in the rank
axioms. Individual cutoffs remove that
restriction, and Lemma~\ref{lem:all-rank-minimum-representation} reduces
verification to Pareto efficiency on the given domain.

\subsection{Single-crossing Profiles}
\label{sec:single-crossing-verification}
For single-crossing preferences, whether the voter order is common to all
rounds makes a difference. A common order permits polynomial-time
verification of both weak and non-weak all-rank axioms. When the order can
change, all four verification problems are coNP-complete, even with three
candidates and only two distinct first choices per round.

\begin{theorem}
\label{thm:common-sc-all-rank-verification}
Suppose that every round is single-crossing with respect to the same given
voter order. Then all-rank-JR, all-rank-PJR, all-rank-wJR, and all-rank-wPJR
can be verified in polynomial time. One implementation uses
$O(n^4+n^2|T|m)$ arithmetic operations.
\end{theorem}
\begin{proof}
Number the voters according to the common order, and write
$[a,b]=\set{a,a+1,\dots,b}$. If $a$ and $b$ are the first and last members
of a group $S$, then
\begin{equation}
 D_W(S)=D_W(\set{a,b})=D_W([a,b]).
 \label{eq:sc-endpoint-domination}
\end{equation}
Indeed, a candidate preferred to the winner by both endpoints is preferred
to it by every voter between them: the voters making that comparison form
a prefix or a suffix of the order. The converse follows because both
endpoints belong to $S$.

For a weak axiom, enlarge any violating group $S$ to $[a,b]$. Its number
of Pareto-efficient winning rounds is unchanged, while its entitlement
can only increase. By Lemma~\ref{lem:all-rank-minimum-representation},
it therefore suffices to check all intervals. This proves the weak cases.

For the non-weak axioms, Theorem~\ref{thm:all-rank-characterization} also
requires the first-choice count $F_W(S)$. For each round, let
\[
 H^j=\set{i\in V:w^j\text{ is voter }i\text{'s first choice in round }j}.
\]
This is an interval, possibly empty: it is the intersection, over
$c\neq w^j$, of the prefixes or suffixes
$\set{i:w^j\succ_i^j c}$. Thus $F_W(S)$ counts the intervals $H^j$
that contain at least one member of $S$.

For fixed endpoints $a,b$, equation~\eqref{eq:sc-endpoint-domination}
fixes $D_W(S)$. For each possible size of $S$, we will minimize $F_W(S)$
by dynamic programming. Fix $a$. Let $\mathrm{dp}_a[s,x]$ be the smallest
number of intervals $H^j$ containing a member of an $s$-voter group whose
first member is $a$ and whose last member is $x$. Set $\mathrm{dp}_a[1,a]=|\set{j:a\in H^j}|$ and $\mathrm{dp}_a[1,x]=+\infty$ for all $x \neq a$.
For $y<x$, define $c(y,x)=|\set{j:x\in H^j,\ y\notin H^j}|$.
These are exactly the intervals newly included when voter $x$ is added
to a group ending at $y$. If an interval contains $x$ and an earlier member,
it also contains $y$. Conversely, an interval containing $x$ but not $y$
cannot contain any earlier member. Therefore
\begin{equation}
 \mathrm{dp}_a[s,x]
 =\min_{a\leq y<x}\{\mathrm{dp}_a[s-1,y]+c(y,x)\}.
 \label{eq:sc-first-choice-dp}
\end{equation}
Entries with no feasible group have value $+\infty$.

Write $D_{ab}=D_W(\set{a,b})$. For fixed $a,b,s$, minimizing $F_W(S)$
maximizes $sD_{ab}-(n-s)F_W(S)$, the left-hand side of the all-rank-PJR
condition in Theorem~\ref{thm:all-rank-characterization}, since
$n-s\geq0$. Thus all-rank-PJR is violated if and only if some feasible
triple satisfies
\[
 sD_{ab}-(n-s)\mathrm{dp}_a[s,b]\geq n.
\]
Similarly, by the all-rank-JR condition in
Theorem~\ref{thm:all-rank-characterization}, all-rank-JR is violated if and
only if some feasible triple satisfies
$\mathrm{dp}_a[s,b]=0$ and $sD_{ab}\geq n$.

Computing all costs $c(y,x)$ takes $O(n^2|T|)$ operations. All counts
$D_{ab}$ take $O(n^2|T|m)$ operations after the rankings have been converted
to rank positions. The dynamic programs over all left endpoints take
$O(n^4)$ operations. These bounds prove the stated running time. Keeping a minimizing
predecessor for each entry in~\eqref{eq:sc-first-choice-dp} also permits
recovery of a violating group. The rank matrix is then obtained from the
proof of Theorem~\ref{thm:all-rank-characterization}, or from
Lemma~\ref{lem:all-rank-minimum-representation} for a weak axiom.
\end{proof}

For all-rank-JR, the dynamic program is not needed. Keep only voters whose
first choice never wins, and check groups consisting of all remaining voters
between each pair of remaining endpoints. Every possible violating group is
contained in one of these groups, with the same domination count and zero
first-choice count. The dynamic program is needed for PJR because a group
may have some first-choice representation and still be owed more rounds.

The preceding theorem uses a common order throughout. The next construction
shows that allowing different orders changes the conclusion for all four
axioms. It also distinguishes construction from verification: the profiles
have only two first choices per round, so Theorem~\ref{thm:first-choice-intervals}
can construct an all-rank-PJR sequence even though checking a proposed
sequence is hard.

\begin{theorem}
\label{thm:varying-sc-all-rank-verification}
Verifying any of all-rank-JR, all-rank-PJR, all-rank-wJR, and all-rank-wPJR
is coNP-complete, even with three candidates, at most two first choices per
round, a constant proposed winner, and a given single-crossing voter order
for each round. The order may differ across rounds.
\end{theorem}
\begin{proof}
By Lemma~\ref{lem:all-rank-minimum-representation} and
Theorem~\ref{thm:all-rank-characterization}, a group violating any of the
four axioms can be checked in polynomial time. Each verification problem
is therefore in coNP.

For hardness, take an instance $(G_0,K_0)$ of \textsc{Independent Set}
with $N_0$ vertices and $1\leq K_0\leq N_0$. Let $e=\max\{1,|E(G_0)|\}$ and $h=2e$.
Add $(h-1)N_0$ vertices with no incident edges to obtain a graph $G$.
Its number of vertices and independent-set target are, respectively, $N=hN_0$ and $K=(h-1)N_0+K_0$.
The original instance has an independent set of size at least $K_0$ if
and only if $G$ has one of size at least $K$. We will use $N<2K$ and $(h-1)N<hK$.
There are $N$ vertex voters and two groups $L,R$ of dummy voters, each of
size
\[
 z=\frac{hK-N}{2}=e(K-N_0)>0.
\]
The total number of voters is $n=hK$, and the horizon is $h$.
All numbers are polynomial in the input size.

Use candidates $a,b,w$ and propose $w$ in every round. The first $e$
rounds contain one round for each edge; if the graph has no edges, use one
round of the form described below without special endpoint voters.
For an oriented edge $u\to v$, set the rankings to
\[
\begin{array}{c|l}
\text{voters}&\text{ranking}\\\hline
L\text{ and }u&a\succ w\succ b\\
\text{other vertex voters}&a\succ b\succ w\\
R\text{ and }v&b\succ w\succ a.
\end{array}
\]
Without an edge, all vertex voters are in the middle row. The displayed
row order is a single-crossing voter order. In the last $e$ rounds,
all vertex voters rank $a\succ b\succ w$ and all dummy voters rank
$w\succ a\succ b$. These rounds are single-crossing with the vertex voters
followed by the dummy voters. Every round has at most two first choices.

First consider a group $S$ containing only vertex voters. No winner is a
member's first choice, so $F_W(S)=0$. In an edge round, the candidates
ranked above $w$ are $\set{a}$ for $u$, $\set{b}$ for $v$, and
$\set{a,b}$ for every other vertex voter. The winner is therefore Pareto
efficient for $S$ exactly when both endpoints belong to $S$. It is Pareto
dominated in every other round. Thus, $h-D_W(S)=|E(G[S])|$.
The weak entitlement is $\lfloor |S|/K\rfloor$. By $|S|\leq N<2K$, it
is at most one, so either weak axiom is violated exactly when $S$ is an
independent set with at least $K$ vertices.

For either non-weak axiom, an independent group has $D_W(S)=h$ and violates
the axiom exactly when $|S|\geq K$, by
Theorem~\ref{thm:all-rank-characterization}. A group that is not independent
has $D_W(S)\leq h-1$, whence
\[
 |S|D_W(S)\leq N(h-1)<hK=n.
\]
Together with $F_W(S)=0$, this rules out violations of both non-weak axioms.
The number of added vertices is what makes this last inequality hold.

We must also check groups containing dummy voters. If $S$ meets both
$L$ and $R$, the winner is Pareto efficient in each of the first $e$ rounds
and is a member's first choice in each of the last $e$ rounds. Thus
$D_W(S)=0$, and no axiom is violated.

If $S$ meets exactly one of $L,R$, then
\[
 F_W(S)=e,\quad D_W(S)\leq e,\quad
 |S|\leq N+z=\frac{n+N}{2}.
\]
Its weak entitlement is at most $e$, since $\lfloor|S|/K \rfloor
 \leq \lfloor e+ N/2K\rfloor=e$.
The last $e$ rounds therefore meet both weak requirements. There is no
non-weak JR violation because $F_W(S)>0$. For PJR,
\[
 |S|D_W(S)-(n-|S|)F_W(S)
 \leq e(2|S|-n)\leq eN<2eK=n,
\]
so $S$ also satisfies the all-rank-PJR condition in
Theorem~\ref{thm:all-rank-characterization}.

Thus a violation of any of the four axioms exists exactly when the original
independent-set instance has a positive answer. This proves coNP-hardness.
Adding candidates below the three displayed candidates preserves the
construction for every larger fixed number of candidates.
\end{proof}

\subsection{Extended Justified Representation with Two Candidates}
The two-candidate example in Theorem~\ref{thm:ejr}(iii)
shows that the individual-cutoff EJR axioms cannot always be satisfied.
Nevertheless, both verification and the existence question for a given
profile are polynomial-time solvable on this domain. First-choice counts
again give the relevant test, this time for individual voters rather than
groups.

\begin{proposition}
\label{prop:binary-all-rank-ejr-verification}
On complete profiles with two candidates, all-rank-EJR and all-rank-wEJR
are equivalent. For a sequence $W$, let $u_i$ be the number of rounds in
which voter $i$'s first choice wins, and order these counts as
$u_{(1)}\leq\cdots\leq u_{(n)}$. Then either axiom holds if and only if
$u_{(s)}\geq\lfloor |T|s/n\rfloor$ for every $s\in[n]$.
Every satisfying sequence coincides with the sequence of first choices of
some voter.
Consequently, verification and deciding existence for a given profile can
both be done in polynomial time.
\end{proposition}
\begin{proof}
For each round, let $c$ be the candidate other than the winner. Give cutoff
two to voters who rank the winner first, and cutoff one to all other voters.
Every voter approves $c$, so every group agrees in every round under this
single rank matrix. Voter $i$ approves the winner exactly when it is her
first choice, and thus approves $u_i$ winners. This is the smallest possible
number for that voter under any rank matrix.

It follows that all-rank-wEJR holds exactly when every group of size $s$
has a member with $u_i\geq\lfloor |T|s/n\rfloor$. These conditions also
suffice for all-rank-EJR, because a group agreeing in $k$ rounds has
$k\leq|T|$. Among groups of size $s$, the smallest possible largest count
is $u_{(s)}$, obtained by the $s$ voters with the smallest counts. This proves
the stated condition and the equivalence.

Taking $s=n$ in this condition gives $u_{(n)}=|T|$. Hence some voter has her first choice
selected in every round, and any satisfying sequence must equal one of the
at most $n$ first-choice sequences in the profile. For verification,
compute and sort the counts and check the condition for every $s\in[n]$. To decide
existence, apply this test to each of those $n$ sequences. This takes
$O(n^2|T|+n^2\log n)$ time.
\end{proof}

This is a test for existence on a particular profile, not a universal
existence guarantee. It therefore complements, rather than contradicts,
the two-voter impossibility example in Theorem~\ref{thm:ejr}(iii).

\section{Truncated Ballots}
\label{sec:truncated-ballots}

Many elections allow voters to rank only the candidates they find acceptable.
We now extend the model to nonempty truncated ballots. This raises two
questions: which of our rules retain their guarantees, and how does truncation
change the complexity of verification? The guarantees of GCR, general SCR,
the first-choice and ordered budget rules, and a modified MES rule extend to
truncated ballots. The polynomial-time SCR rule for a common single-crossing
order also extends, after restricting each considered group to voters whose
lists are long enough (Section~\ref{sec:truncated-rules}).
In contrast, checking all-rank-wJR or all-rank-wPJR can become coNP-complete
even when each ballot omits only two candidates. We explain what changes in
the characterization of these axioms, then determine when verification
remains tractable.

In round $j$, voter $i$ reports a nonempty list
\[
 c_{i,1}^j\succ_i^j c_{i,2}^j\succ_i^j\cdots\succ_i^j
 c_{i,\lambda_i^j}^j,
 \quad 1\leq \lambda_i^j\leq m,
\]
of distinct candidates. Let
$L_i^j=\set{c_{i,1}^j,\dots,c_{i,\lambda_i^j}^j}$.
We interpret an unlisted candidate as unacceptable, so it never belongs to an
approval set. Accordingly, for every $r\in[m]$, define $A_i^j(r)=
 \set{c_{i,1}^j,\dots,c_{i,\min\{r,\lambda_i^j\}}^j}$.
Thus, the approval set stops growing once the end of the list is reached.
The JR, PJR, and EJR axioms use these approval sets with no other change.
For the all-rank axioms, cutoffs larger than $\lambda_i^j$ are redundant, so
we may always assume that $r_i^j\leq\lambda_i^j$.
For PSC, a group $S$ is solid at cutoff $r$ in round $j$ when
$|\bigcap_{i\in S}A_i^j(r)|=r$, as before. Equivalently, each member must
list at least $r$ candidates, and their first $r$ candidates must form the
same set. The definitions of fixed-rank-PSC, rank-PSC, and their weak versions
then apply unchanged.

A truncated ballot is single-peaked on a given candidate axis if the unlisted
candidates can be ordered below all listed candidates so that the resulting
strict ranking is single-peaked on that axis. We use the analogous convention
for single-crossing and one-dimensional Euclidean profiles. Whenever a result
below assumes such a domain, the relevant axis, voter order, or coordinates
are given.

\subsection{Algorithms with Truncated Ballots}
\label{sec:truncated-rules}

The proofs for GCR and SCR use only properties that remain true for truncated
lists: approval sets are nested, and every ballot has a first choice.
For SCR, however, only voters who list at least $r$ candidates can belong to
a group that is solid at cutoff $r$.

\begin{proposition}
\label{prop:truncated-gcr-scr}
With the approval sets defined above, Algorithm~\ref{alg:fixed-rank-gcr}
still satisfies fixed-rank-PJR. Algorithm~\ref{alg:online-scr} still satisfies
fixed-rank-wPSC for constant rank vectors and rank-wPSC for arbitrary rank
vectors after the following change: for a rank vector $\mathbf r$, form
claims only from voters satisfying $\lambda_i^t\geq r^t$ in every
relevant round $t$. Both SCR versions also retain the Droop guarantee of
Corollary~\ref{cor:scr-droop}.
\end{proposition}

\begin{proof}
For GCR, the definition of $\tsr{S}{r}$ and the demand $\dem{S}{r}$ does not
change. The proof of Lemma~\ref{lem:fixed-rank-packing} uses only the demand
bound and the fact that groups selected in the iteration for the same cutoff are
disjoint. Corollary~\ref{cor:fixed-rank-stage-two} then assigns each selected
group the required number of rounds in which its members have a common
approved candidate. Finally, the coverage argument in the proof of
Theorem~\ref{thm:fixed-rank-gcr} uses
$A_i^j(t)\subseteq A_i^j(r)$ for $t\leq r$. This inclusion also holds for the
truncated approval sets. Hence the proof applies without any other change.

For SCR, fix a rank vector $\mathbf r=(r^1,\dots,r^j)$. Among the voters who
list at least $r^t$ candidates in every round $t\leq j$, group together the
voters with the same sets $A_i^t(r^t)$ in all these rounds, and keep the
maximal groups. The refinement loop is well defined because the vector
$(1,\dots,1)$ is available and every ballot is nonempty. Moreover, if two
claims share a voter, then their common approval sets in round $j$ are prefixes of that
voter's reported list and are therefore nested. These are the two facts used
in the refinement argument.

It remains to check the retrospective payment argument from
Theorem~\ref{thm:online-scr}. Any group $S$ that is solid under $\mathbf r$ in every round $t\leq j$ has
$|\bigcap_{i\in S}A_i^t(r^t)|=r^t$. Since every $A_i^t(r^t)$ has size at most
$r^t$, each member of $S$ lists at least $r^t$ candidates and all these
approval sets are equal. Hence $S$ is contained in one of the maximal groups
considered by the modified rule, with the same representation and at least the
same priority. In the first case of the original payment argument, suppose
this maximal group intersects the group that pays in the current round but is
not represented in that round. The two common approval sets in round $j$ are prefixes of
a shared voter's list and are nested. The common approval set of the maximal group is
a proper subset of that of the paying group, so its claim remains eligible after
the paying claim is chosen, with at least the required priority. This
contradicts the choice of the last selected claim with that priority. In the
second case, an unrepresented group has a common approval set in round $j$ that is strictly
contained in $C$, so the same maximal claim is eligible at the first
refinement. This contradicts the assumption that all selected claims have lower
priority. The invariant and the final budget
contradiction are therefore unchanged. Taking $p=n/(|T|+1)$ instead
of $n/|T|$ gives the Droop guarantee by Corollary~\ref{cor:scr-droop}.
\end{proof}

\paragraph{A common single-crossing order.}
Extending the polynomial-time rule of
Theorem~\ref{thm:single-crossing-online-psc} requires more care. With complete
rankings, every voter between the first and last members of a solid group
has the same relevant prefix. An intermediate voter may not report enough
candidates for this to remain true of her truncated ballot. For example,
on voter order $1,2,3$, the complete rankings
\[
 a\succ b\succ w,\quad a\succ b\succ w,\quad b\succ a\succ w
\]
are single-crossing. If they are truncated to $(a,b)$, $(a)$, and $(b,a)$,
then voters $1$ and $3$ are solid together at cutoff two, but voter $2$
cannot join them. We therefore keep only intermediate voters whose lists
are long enough in every relevant round.

\begin{theorem}
\label{thm:truncated-common-sc-scr}
Suppose that the nonempty truncated ballots admit single-crossing completions
with respect to the same given voter order. Then rank-wPSC, also with the
Droop size condition, can be satisfied by a polynomial-time online rule.
Both the Hare and Droop versions of rank-wPSC can be verified in polynomial
time on this domain.
\end{theorem}
\begin{proof}
Number voters according to the common order. For each pair $a\leq b$ and
round $t$, let
\[
 r_{ab}^t=
 \min\{r\leq\min(\lambda_a^t,\lambda_b^t):A_a^t(r)=A_b^t(r)\},
 \quad B_{ab}^t=A_a^t(r_{ab}^t),
\]
provided such a cutoff exists. If it does not exist in some round, omit
that pair from this and all subsequent rounds. For a remaining pair,
in round $j$ consider the group
\begin{equation}
 S_{ab}^j=\{i\in\{a,\dots,b\}:\lambda_i^t\geq r_{ab}^t
                      \text{ for every }t\leq j\}.
 \label{eq:truncated-sc-groups}
\end{equation}
Both endpoints belong to this group. In any single-crossing completion,
every voter between the endpoints has top-$r_{ab}^t$ set $B_{ab}^t$.
A member of $S_{ab}^j$ reports at least that many candidates in each round
$t\leq j$, so her reported prefix also equals $B_{ab}^t$. Thus this group,
with rank vector $(r_{ab}^1,\dots,r_{ab}^j)$, is a valid solid claim.
Its representation before round $j$ and its priority are
\[
 q_{ab}^{j-1}=|\{t<j:w^t\in B_{ab}^t\}|,
 \quad \frac{|S_{ab}^j|}{q_{ab}^{j-1}+1}.
\]
The representation count depends on the endpoint prefixes, not on which
intermediate voters remain in the group.

Use these claims in the refinement loop of Algorithm~\ref{alg:online-scr}.
There are at most $n(n+1)/2$ of them. A singleton pair is never omitted and
has cutoff one in every round. It ensures an eligible refinement whenever
$|X|>1$: choose any voter initially, or any member of the preceding selected
claim thereafter, and use the singleton claim for her first choice in $X$.

To show that the smaller collection suffices, fix any group $S$ solid
throughout rounds $1,\dots,j$ under a rank vector $\mathbf r$. Let $a$
and $b$ be its first and last members. Their common reported prefix shows
that $r_{ab}^t$ exists and satisfies $r_{ab}^t\leq r^t$ for every $t\leq j$.
Every member of $S$ therefore meets the list-length requirements in
\eqref{eq:truncated-sc-groups}. Moreover,
\[
 S\subseteq S_{ab}^j,\quad
 B_{ab}^t\subseteq C^t(S,\mathbf r),\quad
 q_{ab}^{j-1}\leq\mathrm{rep}^{j-1}(S,\mathbf r).
\]
The considered claim has at least the priority of $(S,\mathbf r)$.
Whenever $C^j(S,\mathbf r)$ is a proper subset of the current candidate
set $X$, so is $B_{ab}^j$. Thus an arbitrary eligible solid claim can always
be replaced by a considered eligible claim of at least its priority.

We apply the payment argument of Theorem~\ref{thm:online-scr} with
$p=n/(|T|+1)$. As there, the invariant concerns every solid group whose
priority after the current round is at least $p$, not only the groups
considered by the rule. Such a group is charged only in represented rounds.
Although $S_{ab}^j$ may shrink, the induction hypothesis applies to its
current members: this very group was solid in all previous rounds under
the corresponding shorter rank vector.

If the refinement loop selects a claim of priority at least $p$, take the
last such claim. Its members can pay $p$ by the induction hypothesis. An
unrepresented group of priority at least $p$ cannot intersect this paying
group. Otherwise their common prefixes are nested, and the unrepresented
group's prefix is a proper subset of the paying group's prefix. The
endpoint pair of the unrepresented group would then give an eligible next
refinement of priority at least $p$, contrary to the choice of the last
such claim. If no selected claim has priority at least $p$, any
unrepresented group of that priority would give an eligible first refinement
of priority at least $p$, again a contradiction. These are exactly the two
cases needed to preserve the invariant. There is enough total remaining
budget to pay $p$ in every round.

At the end, the electorate retains $p$ units. If a solid group $S$ with
$|S|>\ell p$ receives only $q\leq\ell-1$ represented rounds, its final
priority exceeds $p$. The invariant implies that it retains
$|S|-qp>p$ units, a contradiction. This proves the Droop guarantee and
hence the Hare guarantee.

For an implementation, keep the current group, endpoint prefix, and
representation count for each pair. In a new round, find the smallest
common reported endpoint prefix, remove intermediate voters whose current
lists are too short, and update the count after selecting the winner.
Finding prefixes by direct comparison and scanning the claims during
refinement takes $O(n^3+n^2m^2)$ elementary operations per round.
No future rankings or horizon are used.

Finally, given a proposed sequence, form the same groups and counts for
all rounds. Each final group is solid throughout. Conversely, the endpoint
group of any violating solid group has at least as many members and at most
as much representation. Since both entitlements are increasing in group
size, a violation exists exactly when one of these at most $n(n+1)/2$
groups violates the relevant entitlement. This gives polynomial-time
verification for both quotas.
\end{proof}

The rule needs the common voter order, but not the unreported parts of any
completion. Its guarantee assumes that compatible single-crossing completions
exist; it does not require finding them. In the preceding three-voter
example, the pair $(1,3)$ keeps exactly voters $1$ and $3$, since voter $2$'s
list is too short.

The fixed-rank verification result also extends to arbitrary truncation,
without a restriction on preferences.

\begin{corollary}
\label{cor:truncated-fixed-psc-verification}
Fixed-rank-wPSC can be verified in polynomial time for nonempty truncated
ballots, under either the Hare or Droop size condition.
\end{corollary}
\begin{proof}
For each cutoff $r$, retain voters whose lists have length at least $r$
in every round. Partition them according to their sequences of top-$r$
sets, as in Theorem~\ref{thm:fixed-rank-verification}. Every solid group is
contained in one class with the same representation and a weakly larger
entitlement. Checking the classes therefore suffices for either quota.
\end{proof}

\paragraph{MES and the budget rules.}

MES requires one additional step because, unlike in the complete-ranking
model, cutoff $m$ need not make any candidate affordable. In
Algorithm~\ref{alg:MES1}, the sets $N_c^j(r)=\set{i\in V:c\in A_i^j(r)}$
are now defined by the truncated approval sets. If no candidate is affordable at any
cutoff in a round, select the first choice of a voter according to a fixed
tie-breaking rule and make no deductions in that round. In every other round,
use the original MES step and deduct one unit in total.

\begin{proposition}
\label{prop:truncated-mes}
The modified MES rule is a polynomial-time semi-online rule satisfying
rank-wPJR for nonempty truncated ballots. Starting each voter with
$(|T|+1)/n$ units instead gives the Droop size condition.
\end{proposition}

\begin{proof}
Suppose, for contradiction, that a rank vector $\mathbf r$, a group
$S$, and an integer $\ell\geq1$ violate rank-wPJR. Thus, $S$ agrees in every
round, $|S|\geq \ell n/|T|$, and $S$ is represented in at most
$\ell-1$ rounds. The initial total budget of $S$ is at least $\ell$.

We prove by induction over the rounds that members of $S$ are charged only in
rounds that represent $S$. At the start of any round, the induction hypothesis
implies that all earlier payments by $S$ occurred in represented rounds. Since
there are at most $\ell-1$ represented rounds in the entire sequence and each
round deducts at most one unit, $S$ still has total budget at least one. Let
$c$ be a candidate in $\bigcap_{i\in S}A_i^j(r^j)$. The voters in $S$ alone
have enough remaining budget to buy $c$ at cutoff $r^j$: the function
$\sum_{i\in N_c^j(r^j)}\min\{\rho,b_i\}$ increases continuously from zero
to a value of at least one. The rule therefore finds an affordable candidate
at some cutoff $r\leq r^j$ and does not enter the no-deduction case in this
round. If a voter $i\in S$ is charged, then $w^j\in A_i^j(r)\subseteq A_i^j(r^j)$ and the round represents $S$. This proves the induction claim.

It follows that $S$ retains at least one unit after the final round. The same
argument also shows that an affordable candidate exists in every round, so the
no-deduction case is never used. The rule therefore deducts exactly $|T|$ units,
which is the total initial budget of all voters. Since budgets remain
nonnegative, every voter must end with budget zero, a contradiction.

For the Droop variant, suppose that a group with initial total budget
$x=(|T|+1)|S|/n$ is represented in $q$ rounds and violates the claimed
entitlement. Then $x>q+1$. The same induction, now using $x-q>1$, shows
that the group pays only in represented rounds and that an affordable
candidate exists in every round. The electorate ends with exactly one unit,
whereas the group retains more than one, a contradiction.

Affordability can be computed as in the proof of
Theorem~\ref{thm:mes-rank-wpjr}; the additional check and tie-breaking step are
polynomial. The rule uses the horizon but no future rankings, and is therefore
semi-online.
\end{proof}

The JR and PJR guarantees of the first-choice rules in
Section~\ref{sec:first-choices} extend without a change to their algorithms
or proofs: every approval set $A_i^j(r)$ still contains voter $i$'s first
choice. The EJR conclusion of Corollary~\ref{cor:ejr-candidate-horizon}
uses complete two-candidate ballots, and is not included in this extension. The same is true for
Algorithm~\ref{alg:ordered-budgets}. To see this, complete each truncated
ballot within the given single-peaked or single-crossing domain. The algorithm
uses only the first choices and the given order, so it returns the same
sequence on the completion. Every approval set of the truncated
ballot is also an approval set of that completion. Theorem~\ref{thm:ordered-budgets}
therefore implies all-rank-wPJR for the truncated profile.

\subsection{A Characterization for All-rank Verification}
\label{sec:truncated-characterization}

With complete rankings, every group can agree in every round by placing all
cutoffs at $m$. Truncation removes this possibility: the members must first
have a common listed candidate in every round. Among groups that meet this
condition, we can again determine the minimum number of represented rounds
over all cutoffs. This gives the counterpart of
Lemma~\ref{lem:all-rank-minimum-representation} needed for verification.
Fix a proposed sequence $W$. For every voter $i$ and round $j$, let
\[
 B_i^j=
 \begin{cases}
  \set{c_{i,t}^j:t<s},
   &\text{if }w^j=c_{i,s}^j\text{ for some }s,\\
  L_i^j, &\text{if }w^j\notin L_i^j.
 \end{cases}
\]
Thus, $B_i^j$ contains exactly the listed candidates that voter $i$ can
approve without also approving $w^j$.

\begin{lemma}
\label{lem:truncated-minimum-representation}
A group $S$ can agree in every round under some
rank matrix if and only if $\bigcap_{i\in S}L_i^j\neq\varnothing$ for every $j\in T$.
In this case, the minimum number of represented
rounds over all such rank matrices is
$q(S)=|\set{j\in T:\bigcap_{i\in S}B_i^j=\varnothing}|$.
\end{lemma}

\begin{proof}
If $S$ agrees in round $j$, every common approved candidate belongs to
$\bigcap_{i\in S}L_i^j$. Conversely, if a candidate $c$ belongs to this
intersection, each voter can place her cutoff at $c$, so the group agrees.
This proves the first statement.

Now fix a round $j$. If $\bigcap_{i\in S}B_i^j$ contains a candidate $c$,
placing every cutoff at $c$ makes the group agree and excludes $w^j$ from
every member's approval set. If this intersection is empty, suppose that some
cutoffs make the group agree without representing it. Any common approved
candidate must then precede $w^j$ on every list on which $w^j$ appears, and
$w^j$ must be unlisted on every remaining list. The candidate would therefore
belong to every $B_i^j$, a contradiction. Hence representation is unavoidable
exactly in the rounds counted by $q(S)$. Cutoffs can be
chosen independently across rounds, which proves the formula.
\end{proof}

By Lemma~\ref{lem:truncated-minimum-representation}, $W$ violates
all-rank-wPJR if and only if some group that can agree in every round also satisfies
\begin{equation}
 |T||S|-nq(S)\geq n.
 \label{eq:truncated-violation}
\end{equation}
Indeed, this inequality is equivalent to
$q(S)<\lfloor |S||T|/n\rfloor$. For all-rank-wJR, the corresponding conditions
are $q(S)=0$ and $|T||S|\geq n$.

For a given group $S$, the intersections $\bigcap_{i\in S}L_i^j$ and the value $q(S)$ can be computed in polynomial time.
Thus one can test directly whether $S$ satisfies
\eqref{eq:truncated-violation}, or whether $q(S)=0$ and $|T||S|\geq n$.
Hence, deciding whether $W$ violates either weak axiom is in NP, and all
verification problems considered below are in coNP.

\subsection{When Does Truncation Make Verification Hard?}

We consider three restrictions on the reported lists: two candidates, few
omitted candidates, and short lists. Our results show that verification can become
hard after only a small change to ballots for which it is easy.

\paragraph{Two candidates.}
With two candidates, all-rank-wJR and all-rank-wPJR behave differently under
truncation. All-rank-wJR remains easy to verify for arbitrary truncation. For
all-rank-wPJR, verification remains tractable when at most one ballot per round
has length one, whereas two such ballots per round suffice for coNP-completeness.

\begin{theorem}
\label{thm:truncated-binary-verification}
Suppose that $m=2$ and every ballot is nonempty.
\begin{enumerate}[(i)]
 \item All-rank-wJR can be verified in $O(n|T|)$ time.
 \item All-rank-wPJR can be verified in polynomial time if at most one ballot
 in each round has length one.
 \item All-rank-wPJR is coNP-complete even if at most two ballots in each
 round have length one, all other ballots are complete, the proposed winner
 is the same in every round, and the ballots have strict Euclidean completions
 with fixed candidate positions.
\end{enumerate}
\end{theorem}

\begin{proof}
Let $U_W$ be the voters whose first listed candidate differs from $w^j$ in
every round $j$. If a group is unrepresented in every round, each of its members
belongs to $U_W$. Conversely, since there are two candidates, every voter in $U_W$
lists the other candidate first in each round. Hence $U_W$ agrees at cutoff $1$
in every round and is never represented. All-rank-wJR is violated exactly
when $|T||U_W|\geq n$, which can be checked in $O(n|T|)$ time.

Now suppose that at most one ballot per round has length one. In every round,
the candidate on that ballot is listed by every voter; hence every group can
agree. Append the missing candidate to each length-one ballot. This does not
change any set $B_i^j$: if the winner is listed, the appended candidate comes
after it, and if the winner is unlisted, it is appended after the listed
candidate. Lemma~\ref{lem:truncated-minimum-representation} therefore gives
the same minimum representation for every group before and after completion.
Every complete profile with two candidates is single-peaked, so
Theorem~\ref{thm:single-peaked-verification} gives a polynomial-time
verification algorithm.

For the third statement, reduce from \textsc{Independent Set}
\cite{karp1972reducibility}. Let $G=(X,E)$ have $N$ vertices, and let
$K\in[N]$. Orient every edge arbitrarily, and let $d_u^+$ be the outdegree of
$u$. Set $D=\max\set{1,\max_{u\in X}d_u^+}$, $n=K\lceil N/K \rceil$, and $|T|=Dn+ \lceil  N/K \rceil$.
Create one vertex voter for each $u\in X$ and $n-N$ dummy voters. The latter
number is nonnegative because $K\lceil N/K\rceil\geq N$. There are two
candidates, $a$ and $b$, and $w^j=a$ in every round. Every dummy voter reports
$(a,b)$ in every round.

For every oriented edge $u\to v$, add one round in which $u$ reports $(a)$,
$v$ reports $(b)$, and every other vertex voter reports $(b,a)$. For every
vertex $u$, add $D-d_u^+$ rounds in which $u$ reports $(a,b)$ and every other
vertex voter reports $(b,a)$. The number of rounds added so far is $|E|+\sum_{u\in X}(D-d_u^+)=DN$.
Add $|T|-DN=D(n-N)+ \lceil\frac{N}{K}\rceil$ more rounds in which every vertex voter reports $(b,a)$. Only the edge rounds
contain length-one ballots, and each such round contains exactly two. Place
$a$ at $0$ and $b$ at $1$; in each round, place every voter at the position of
her first choice. This gives strict Euclidean completions with the same
candidate positions in all rounds.

Consider a group $S$ of vertex voters. In the round for $u\to v$, its members
have a common listed candidate exactly when $S$ does not contain both $u$ and
$v$. No other round restricts agreement. Hence $S$ can agree in every round if
and only if it is an independent set.

In an edge round $u\to v$, the set $B_u^j$ is empty and every other vertex
voter has $B_i^j=\set{b}$. The same statement holds in each of the
$D-d_u^+$ additional rounds for $u$. Each vertex is therefore the unique
vertex voter with an empty $B_i^j$ in exactly $D$ rounds, and these rounds are
different for different vertices. For every independent set $S$, $q(S)=D|S|$.
On the other hand, $\lfloor |T||S|/n \rfloor
 =D|S|+\lfloor |S|/K \rfloor$.
Thus $S$ violates all-rank-wPJR exactly when $|S|\geq K$. Any group containing
a dummy voter is represented in every round and cannot violate the axiom.
Consequently, $W$ satisfies all-rank-wPJR if and only if $G$ has no independent
set of size at least $K$. The construction is polynomial, so verification is
coNP-complete.
\end{proof}

\paragraph{Few omitted candidates.}
With at least three candidates, omitting at most one candidate from each ballot
preserves polynomial-time verification on single-peaked profiles.
The reason is that every group still has a common listed candidate in every
round, and appending the missing candidate does not change its minimum
representation. Omitting two candidates is already enough for
coNP-completeness, even when the list lengths are equal.

\begin{theorem}
\label{thm:truncated-omission-boundary}
Suppose that every round has a given candidate axis and every truncated ballot
has a single-peaked completion on that axis.
\begin{enumerate}[(i)]
 \item If $m\geq3$ and every ballot omits at most one candidate, then
 all-rank-wJR and all-rank-wPJR can be verified in polynomial time.
 \item Both verification problems are coNP-complete when $m=3$ and each
 ballot omits at most two candidates, even with fixed Euclidean candidate
 positions, a constant proposed winner, and a candidate listed by every voter
 in every round.
 \item For every fixed $m\geq4$, both problems remain coNP-complete when
 every ballot lists exactly $m-2$ candidates. The candidate positions are
 fixed and the proposed winner is constant. For $m\geq5$, one candidate is
 listed by every voter in every round.
\end{enumerate}
\end{theorem}

\begin{proof}
Suppose first that every ballot omits at most one candidate. Append the missing
candidate to the end of each incomplete list. This is the unique completion of
that list, and it is single-peaked by assumption. The completion does not
change $B_i^j$: if $w^j$ is listed, the new candidate is appended after it; if
$w^j$ is unlisted, then $w^j$ is the appended candidate and every listed
candidate precedes it.

Each set $L_i^j$ is an interval of the candidate axis and has size at least
$m-1$. When $m\geq3$, any two such intervals intersect. Intervals on a line
that intersect pairwise have a common point, so
$\bigcap_{i\in V}L_i^j\neq\varnothing$ in every round. Hence every group can
agree in every round, and Lemma~\ref{lem:truncated-minimum-representation}
gives the same minimum representation in the truncated profile and its
completion. Theorem~\ref{thm:single-peaked-verification} now applies.

For the two hardness statements, start with an instance $(G_0,K_0)$ of
\textsc{Independent Set}, where $G_0$ has $N_0$ vertices and
$1\leq K_0\leq N_0$~\cite{karp1972reducibility}. Add $N_0$ new vertices with
no incident edges, and call the resulting graph $G$. Let $N=2N_0$ and $K=N_0+K_0$.
Then $N<2K$, and $G$ has an independent set of size at least $K$ if and only
if $G_0$ has one of size at least $K_0$. Set $|T|=\max\set{2,|E(G)|}$ and $n=K|T|$.
Create one vertex voter for each vertex of $G$ and $n-N$ dummy voters. Since
$n\geq2K>N$, this number is nonnegative. Use one round for each edge of $G$,
and add extra rounds, specified below, until the horizon is $|T|$. The proposed winner
is $w$ in every round.

For the second statement, place $a,b,w$ at $0,10,12$, respectively. For each
edge, choose either endpoint as $u$ and denote the other by $v$. In the
corresponding round, use the following positions and lists:
\[
\begin{array}{c|c|l}
\text{voter}&\text{position}&\text{reported list}\\\hline
u&0&a\\
v&10&b\succ w\succ a\\
\text{other vertex voter}&0&a\succ b\\
\text{dummy voter}&12&w\succ b\succ a
\end{array}
\]
Every list is a prefix of the strict Euclidean ranking at the displayed
position. Every voter lists $a$, and no ballot omits more than two candidates.
Relative to the proposed winner $w$, the four sets $B_i^j$ are $\set{a}, \set{b}, \set{a,b}, \varnothing$.
In each added round, every vertex voter uses the third row and every dummy
voter uses the fourth row. Thus every group can agree in every round. A group
containing a dummy voter is represented in every round. For a group $S$ of
vertex voters,
\begin{equation}
 q(S)=\left|\set{e\in E(G):e\subseteq S}\right|.
 \label{eq:truncated-edge-count}
\end{equation}
Indeed, in an edge round the intersection of the relevant $B_i^j$ sets is
empty exactly when both endpoint voters belong to $S$.

Since $n=K|T|$ and $|S|\leq N<2K$, $\lfloor |T||S|/n \rfloor  =\lfloor |S|/K \rfloor$
is positive exactly when $|S|\geq K$, and then it equals one. By
\eqref{eq:truncated-edge-count}, a vertex group violates either weak axiom
exactly when it is an independent set of size at least $K$. This proves the
second statement.

For the third statement, first suppose that $m\geq5$. Place $a_0,a_1,\dots,a_{m-3},b,w$ at positions $0,1,\dots,m-3,10m,10m+1$ respectively. For each edge, again choose one endpoint as $u$ and the other as
$v$. In its round, use
\[
\begin{array}{c|c|l}
\text{voter}&\text{position}&\text{reported list}\\\hline
u&0&a_0\succ a_1\succ\cdots\succ a_{m-3}\\
v&10m&b\succ w\succ a_{m-3}\succ\cdots\succ a_2\\
\text{other vertex voter}&5m+\tfrac14
 &a_{m-3}\succ\cdots\succ a_1\succ b\\
\text{dummy voter}&10m+1
 &w\succ b\succ a_{m-3}\succ\cdots\succ a_2
\end{array}
\]
Every list has length $m-2$, and every list contains $a_2$. The first,
second, and fourth rows are immediate from the displayed positions. For the
third row, all candidates $a_t$ with $t\geq1$ are closer than $b$; moreover,
$b$, $a_0$, and $w$ have respective distances
$5m-\tfrac14$, $5m+\tfrac14$, and $5m+\tfrac34$. Hence the displayed list is
exactly the first $m-2$ candidates in that voter's Euclidean ranking.

The four sets $B_i^j$ are $\set{a_0,\dots,a_{m-3}}$, $\set{b}$, $\set{a_1,\dots,a_{m-3},b}$, $\varnothing$.
Every group can agree because every voter lists $a_2$. For a vertex group,
the intersection of the $B_i^j$ sets is empty in an edge round exactly when
both endpoint voters belong to the group. Use the third row for every vertex
voter and the fourth row for every dummy voter in the added rounds. Equation
\eqref{eq:truncated-edge-count} and the same size calculation prove
coNP-hardness.

When $m=4$, place $a_0,a_1,b,w$ at $0,1,40,41$, respectively. In an edge
round, use
\[
\begin{array}{c|c|l}
\text{voter}&\text{position}&\text{reported list}\\\hline
u&0&a_0\succ a_1\\
v&40&b\succ w\\
\text{other vertex voter}&20+\tfrac14&a_1\succ b\\
\text{dummy voter}&41&w\succ b
\end{array}
\]
These are the first two candidates in the corresponding strict Euclidean
rankings. In each added round, every vertex voter uses the third row and every
dummy voter uses the fourth row. A vertex group has a common listed candidate
in an edge round exactly when it does not contain both endpoints. In every
such case, the intersection of its $B_i^j$ sets is also nonempty, so the group
can choose cutoffs that exclude $w$. Hence a vertex group can agree in all
rounds and remain unrepresented exactly when it is an independent set. The
same size calculation again proves coNP-hardness for both weak axioms.
\end{proof}

Thus, for single-peaked truncated ballots with $m\geq3$, one omitted candidate
per ballot preserves polynomial verification, whereas two omitted candidates
can make verification coNP-complete. The hardness does not depend on varying
list lengths. For every fixed ballot length of at least two, it holds when the
number of candidates is two larger than the ballot length.

\paragraph{Lists of length one and two.}
The preceding theorem concerns ballots that are almost complete. Short lists
do not in general make verification easier. Lists of length one can be checked
directly by grouping voters with identical lists in every round, but allowing
length two is enough for coNP-completeness. This holds even on a common
single-crossing voter order.

\begin{theorem}
\label{thm:truncated-short-list-boundary}
If every ballot has length one, then all-rank-wJR and all-rank-wPJR can be
verified in polynomial time. If every ballot has length at most two, both
problems are coNP-complete, even with three candidates at fixed Euclidean
positions, a constant proposed winner, and strict completions that are
single-crossing on the same given voter order in every round.
\end{theorem}

\begin{proof}
Suppose first that every ballot has length one. Partition the voters according
to their sequence of listed candidates over all rounds. A group can agree in
every round if and only if it is contained in one class of this partition.
All subgroups of a class are represented in the same rounds, while the
proportional demand is nondecreasing with the size of the group. It therefore
suffices to check each whole class, which takes polynomial time.

For hardness, use the graph $G$, target $K$, horizon $|T|$, and electorate size
$n=K|T|$ from the proof of
Theorem~\ref{thm:truncated-omission-boundary}. Let
$C=\set{a,b,w}$, place these candidates at $0,1,3$, and set $w^j=w$ in every
round. Order the vertex voters as $1,\dots,N$, followed by all dummy voters.
In the round for edge $\set{u,v}$, rename its endpoints so that $u<v$. Place
vertex voters $1,\dots,u$ at $0$, the remaining vertex voters at $1$, and all
dummy voters at $3$. Their complete rankings are, respectively,
\[
 a\succ b\succ w,\quad b\succ a\succ w,\quad w\succ b\succ a.
\]
These rankings are Euclidean and single-crossing on the fixed voter order.
Truncate voter $u$ to $(a)$ and voter $v$ to $(b)$. Every other vertex voter
reports her first two candidates, and every dummy voter reports $(w)$. In each
added round, place all vertex voters at $0$, let them report $(a,b)$, and let
every dummy voter report $(w)$.

No vertex voter ever lists $w$. In an edge round, a group of vertex voters has
a common listed candidate exactly when it does not contain both endpoints.
Thus a vertex group can agree in every round exactly when it is an independent
set, and every such group is unrepresented in every round. A group containing
a dummy voter is represented in every round. Since $n/|T|=K$ and $N<2K$, a
violation of either weak axiom exists exactly when $G$ has an independent set
of size at least $K$. This proves coNP-hardness.
\end{proof}

\subsection{Verification with Few Deeply Truncated Ballot Types}
\label{sec:truncated-fpt}

The hardness results leave open whether verification is tractable when few
voters omit several candidates. We prove a stronger statement: it suffices
that these voters have few distinct types, defined by the sets relevant to
agreement and representation.
For $m\geq3$, call a voter \emph{deeply truncated} if she lists at most $m-2$
candidates in some round. For $m=2$, call her deeply truncated if she submits
a length-one ballot in some round. The case $m=1$ is immediate. Relative to
the proposed sequence $W$, two deeply truncated voters have the same type if
their sets $L_i^j$ and $B_i^j$ are identical in every round. Let $\theta$ be the
number of these types. We can decide which types belong to a group by
considering $2^\theta$ possibilities; for each choice, the remaining voters
can be handled by the minimum-cut method for complete rankings.

\begin{theorem}
\label{thm:truncated-fpt-verification}
For $m\geq2$, suppose that the truncated profile is single-peaked with
respect to a given axis in each round; the axes may differ across rounds. Then
all-rank-wJR and all-rank-wPJR can each be verified in
$2^\theta\operatorname{poly}(n,m,|T|)$ time, using at most $2^\theta n$ minimum-cut
computations. The algorithm also returns a violating group and a
corresponding rank matrix whenever the proposed sequence violates the axiom.
\end{theorem}

\begin{proof}
Let $F$ be the voters who are not deeply truncated. We first reduce the choice
among deeply truncated voters to $2^\theta$ cases. If a group contains one voter of
a given deeply truncated type, adding every voter of that type changes neither
$\bigcap_{i\in S}L_i^j$ nor $\bigcap_{i\in S}B_i^j$ in any round, and it can
only increase the group's demand. Hence, whenever a violating group exists,
there is one that contains either all voters of each deeply truncated type or
none of them.

Enumerate all $2^\theta$ unions $Q$ of deeply truncated types. For a fixed $Q$, let $I_Q^j=\bigcap_{i\in Q}L_i^j$ and $D_Q^j=\bigcap_{i\in Q}B_i^j$, where both intersections are defined as $C$ when $Q=\varnothing$. If some $I_Q^j$ is empty, no group whose deeply truncated part is $Q$ can agree in every round, so discard this choice. Otherwise, let
\[
 F_Q=\set{i\in F:I_Q^j\cap L_i^j\neq\varnothing
                   \text{ for every }j\in T}.
\]
We claim that $Q\cup Y$ can agree in every round for every $Y\subseteq F_Q$.
Every listed set is a top prefix of a single-peaked completion and is therefore
an interval of the relevant candidate axis. If $m\geq3$,
each $L_i^j$ with $i\in F$ has size at least $m-1$, so these intervals
intersect pairwise. By the definition of $F_Q$, each also intersects $I_Q^j$,
which is itself an interval. Hence the whole family is pairwise intersecting,
and its members have a common point. For $m=2$, every voter in $F$ lists both
candidates, and the statement is immediate. Conversely, if a group with
deeply truncated part $Q$ can agree in every round, then each of its remaining
voters belongs to $F_Q$.

When $Q\neq\varnothing$, first compute $q(Q)$ and test $Q$ directly using
\eqref{eq:truncated-violation} for all-rank-wPJR, or the conditions
$q(Q)=0$ and $|T||Q|\geq n$ for all-rank-wJR. It remains to optimize over
nonempty sets $Y\subseteq F_Q$. Fix a voter $v\in F_Q$, which we require to belong to $Y$, and, for each round, let $J^j=D_Q^j\cap B_v^j$.
Every set $B_i^j$ is an interval because it is a prefix of a single-peaked
completion. For $i\in F$, append the sole missing candidate when necessary;
this does not change $B_i^j$, and the resulting complete ranking is
single-peaked by the domain assumption. If the peak
is to the left of $w^j$, then $B_i^j$ is an interval ending immediately to the
left of $w^j$. If the peak is to the right, it is an interval beginning
immediately to the right of $w^j$; if the peak is $w^j$, it is empty. Thus the
nonempty sets $B_i^j$ on either side of $w^j$ are nested.

For every $Y\subseteq F_Q$ containing $v$,
\begin{equation}
 \bigcap_{i\in Q\cup Y}B_i^j=\varnothing
 \quad\Longleftrightarrow\quad
 \text{there is an }i\in Y\text{ with }J^j\cap B_i^j=\varnothing.
 \label{eq:truncated-fixed-voter}
\end{equation}
If $J^j\cap B_i^j$ is empty for some $i\in Y$, then the intersection on the
left is empty. Conversely, suppose that $J^j$ is nonempty and meets every
$B_i^j$ with $i\in Y$. All these nonempty sets lie on the same side of $w^j$
as $B_v^j$ and form a nested family. Their smallest member meets $J^j$, so
their intersection with $J^j$ is nonempty. This proves
\eqref{eq:truncated-fixed-voter}. If $J^j$ is empty, both sides of the
equivalence hold.

For fixed $Q$ and $v$, adding a voter to the group contributes $|T|$ to the
left-hand side of~\eqref{eq:truncated-violation}, while each additional round
in which representation is unavoidable costs $n$. We encode this calculation
by one minimum-cut computation. For all-rank-wPJR, construct a directed
network with a source, a sink, one node for
each voter in $F_Q$, and one node for each round. Add an arc of capacity
$|T|$ from the source to every voter node and an arc of capacity $n$ from every
round node to the sink. Add an arc of capacity $|T||F_Q|+n|T|+1$ from voter $i$ to round $j$ whenever $J^j\cap B_i^j=\varnothing$, and add an arc of the same capacity from the source to $v$.

There is a cut avoiding all large-capacity arcs whose capacity is at most
$|T||F_Q|+n|T|$: put $v$ and every round reached from $v$ on the source
side, and put all other voter nodes and round nodes on the sink side. Thus no
minimum cut uses a large-capacity arc. If $Y$ is the set of voter nodes on the
source side, then $v\in Y$. The large-capacity arcs force every round required
by a voter in $Y$ to the source side. No other round is placed there in a
minimum cut, because moving it to the sink side strictly lowers the capacity.
By~\eqref{eq:truncated-fixed-voter}, the number of source-side round nodes is
$q(Q\cup Y)$. The cut
capacity is therefore
\[
 |T|(|F_Q|-|Y|)+nq(Q\cup Y).
\]
If $\kappa(Q,v)$ denotes the minimum cut value, then
\begin{equation}
 |T|(|Q|+|F_Q|)-\kappa(Q,v)
 =\max_{Y\subseteq F_Q, v\in Y}
   (|T||Q\cup Y|-nq(Q\cup Y)).
 \label{eq:truncated-cut-pjr}
\end{equation}
By~\eqref{eq:truncated-violation}, a value of at least $n$ in
\eqref{eq:truncated-cut-pjr} is exactly an all-rank-wPJR violation.

For all-rank-wJR, replace the capacity $n$ on every round-to-sink arc by
$n|T|$, and use $|T||F_Q|+n|T|^2+1$ for every large-capacity arc. The same
cut argument gives
\[
 \max_{Y\subseteq F_Q, v\in Y}
 (|T||Q\cup Y|-n|T|q(Q\cup Y)).
\]
If $q(Q\cup Y)\geq1$, this expression is at most zero because
$|Q\cup Y|\leq n$. It is at least $n$ precisely when $q(Q\cup Y)=0$ and
$|T||Q\cup Y|\geq n$, which is exactly an all-rank-wJR violation.

There are at most $n$ choices of $v$ for each of the $2^\theta$ sets $Q$. Every
network has $O(n+|T|)$ nodes, $O(n|T|)$ arcs, and capacities with polynomially
many bits. This proves the running-time bound.

Finally, let $S$ consist of $Q$ and the voters whose nodes are on the source
side; this is the violating group. In each round, choose a candidate from $\bigcap_{i\in S}B_i^j$ when
this intersection is nonempty, and otherwise choose a candidate from
$\bigcap_{i\in S}L_i^j$. Place each member's cutoff at the chosen candidate,
and set the cutoffs of voters outside $S$ to one. By
Lemma~\ref{lem:truncated-minimum-representation}, the resulting rank matrix
represents the group in exactly $q(S)$ rounds.
\end{proof}

\begin{corollary}
\label{cor:truncated-fpt-eth}
Assuming the Exponential Time Hypothesis, neither verification problem admits
a $2^{o(\theta)}\operatorname{poly}(n,m,|T|)$ algorithm, even with three
candidates.
\end{corollary}

\begin{proof}
Under the Exponential Time Hypothesis, \textsc{Independent Set} on an
$N_0$-vertex graph has no $2^{o(N_0)}\operatorname{poly}(N_0)$ algorithm
\cite{impagliazzo2001strongly}. In the three-candidate reduction used in the
proof of Theorem~\ref{thm:truncated-omission-boundary}, only original graph
vertices can have a list of length one. The $N_0$ added vertices with no
incident edges and all dummy voters list at least two candidates in every
round. Hence $\theta\leq N_0$, and the size of the constructed instance is
polynomial in $N_0$. A $2^{o(\theta)}\operatorname{poly}(n,m,|T|)$ verification
algorithm for either weak axiom would therefore give a
$2^{o(N_0)}\operatorname{poly}(N_0)$ algorithm for \textsc{Independent Set}.
\end{proof}

Truncation thus affects construction and verification differently. GCR, MES,
and general SCR retain their respective guarantees, as do the first-choice
and ordered budget rules on the corresponding domains. A common single-crossing
order also permits polynomial-time online rank-wPSC and polynomial-time
verification of that axiom, despite truncation. Yet checking
all-rank-wJR or all-rank-wPJR can be coNP-complete even with three candidates
and Euclidean completions, for either nearly complete or very short ballots.
On single-peaked profiles with few deeply truncated types, the minimum-cut
approach for complete rankings still gives an efficient verification algorithm.

\section{Discussion}
Ranked temporal proportionality depends on how cutoffs may vary, which rounds
of agreement give a group an entitlement, and what information is available
when a winner is selected. Our axioms separate these choices. On unrestricted
rankings, fixed-rank-PJR is attainable offline, rank-wPJR is attainable
semi-online, and the weak PSC axioms are attainable online. The corresponding
impossibilities explain why restricting either cutoff flexibility or the
groups protected can be necessary. For all-rank-PJR, the optimal additive
loss describes what remains achievable when exact proportionality is
impossible. The EJR results give a further distinction. When there are only
two candidates and every voter ranks both, fixed-rank-wEJR and rank-wEJR can
be satisfied even on every time interval, whereas individual cutoffs can make
all-rank-wEJR unattainable.

Restrictions on preferences improve different parts of the theory. Few first
choices give guarantees even to groups that agree in only some rounds;
single-peaked and single-crossing preferences give exact weak all-rank
proportionality without bounding the number of first choices. A common
single-crossing order permits efficient online rank-wPSC, including with
truncated ballots. For the approval-based axioms, knowing the horizon can
still be necessary even when the entire profile is static and Euclidean.

Verification behaves differently from construction. On three-candidate
Euclidean profiles, weak JR and PJR are easy to check in their fixed-rank
and all-rank versions but hard in their rank versions. A common
single-crossing order makes both weak and non-weak all-rank verification
polynomial, whereas changing orders makes it hard even when an exact
all-rank-PJR sequence is easy to construct. Truncation can also make
verification hard without removing the constructive guarantees.

Several questions remain open. Can fixed-rank-PJR be achieved by a
polynomial-time rule on unrestricted rankings? Can rank-wPSC be achieved
by a polynomial-time online rule without a common single-crossing order?
With three first choices, can an online rule satisfy all-rank-PJR exactly
at every horizon, counting representation from the first round? The
impossibility on every time interval does not answer this question, because
it uses intervals that start later. The static Euclidean example does not
answer it either, since it uses fifteen first choices. Finally, for more
than three first choices, the optimal additive loss on time intervals lies
between $\Delta_d=d/e+O(1)$ and $d-2$; determining it exactly remains open.

\bibliographystyle{alpha}
\bibliography{sample}

@inproceedings{lackner2023proportional,
  title={Proportional decisions in perpetual voting},
  author={Lackner, Martin and Maly, Jan},
  booktitle={Proceedings of the 37th AAAI Conference on Artificial Intelligence (AAAI)},
  pages={5722--5729},
  year={2023}
}

@incollection{karp1972reducibility,
  author = {Karp, Richard M.},
  title = {Reducibility among Combinatorial Problems},
  booktitle = {Complexity of Computer Computations},
  pages = {85--103},
  publisher = {Springer},
  year = {1972}
}

@article{chandak2024proportional,
  author = {Nikhil Chandak and Shashwat Goel and Dominik Peters},
  title = {Proportional Aggregation of Preferences for Sequential Decision Making},
  journal = {Journal of Artificial Intelligence Research},
  volume = {85},
  year = {2026},
}

@book{dummett1984voting,
	author = {Michael Dummett},
	editor = {},
	publisher = {Oxford University Press UK},
	title = {Voting Procedures},
	year = {1984}
}

@inproceedings{brill2023robust,
  title={Robust and Verifiable Proportionality Axioms for Multiwinner Voting},
  author={Brill, Markus and Peters, Jannik},
  booktitle={Proceedings of the 24th ACM Conference on Economics and Computation (EC)},
  pages={301},
  year={2023}
}

@article{bulteau2021justified,
  title={Justified representation for perpetual voting},
  author={Bulteau, Laurent and Hazon, Noam and Page, Rutvik and Rosenfeld, Ariel and Talmon, Nimrod},
  journal={IEEE Access},
  volume={9},
  pages={96598--96612},
  year={2021},
  publisher={IEEE}
}

@inproceedings{aziz2017condorcet,
  title={The Condorcet Principle for Multiwinner Elections: From Shortlisting to Proportionality},
  author={Haris Aziz and Edith Elkind and Piotr Faliszewski and Martin Lackner and Piotr Skowron},
  booktitle={Proceedings of the 26th International Joint Conference on Artificial Intelligence (IJCAI)},
  pages={84--90},
  year={2017}
}

@inproceedings{elkind2025verifying,
  title={Verifying proportionality in temporal voting},
  author={Elkind, Edith and Obraztsova, Svetlana and Peters, Jannik and Teh, Nicholas},
  booktitle={Proceedings of the 39th AAAI Conference on Artificial Intelligence (AAAI)},
  pages={13805--13813},
  year={2025}
}

@inproceedings{phillips2026strengthening,
author = {Phillips, Bradley and Elkind, Edith and Teh, Nicholas and W{\k{a}}s, Tomasz},
title = {Strengthening Proportionality in Temporal Voting},
booktitle = {Proceedings of the 25th International Conference on Autonomous Agents and Multiagent Systems (AAMAS)},
pages = {3823--3832},
year = {2026}
}

@inproceedings{aziz2024committee,
  title={Committee monotonicity and proportional representation for ranked preferences},
  author={Aziz, Haris and Lederer, Patrick and Peters, Dominik and Peters, Jannik and Ritossa, Angus},
  booktitle={Proceedings of the 26th ACM Conference on Economics and Computation (EC)},
  year={2025},
  pages = {896}
}

@article{schmidt1972irregularities,
  author = {Schmidt, Wolfgang M.},
  title = {Irregularities of Distribution, {VII}},
  journal = {Acta Arithmetica},
  volume = {21},
  number = {1},
  pages = {45--50},
  year = {1972}
}

@article{borodin2026online,
  author = {Borodin, Allan and Lueger, Tristan},
  title = {Online Temporal Voting: Strategyproofness, Proportionality and Asymptotic Analysis},
  journal = {arXiv preprint arXiv:2603.26504},
  year = {2026}
}

@article{holte1992pinwheel,
  author = {Holte, Robert and Rosier, Louis and Tulchinsky, Igor and Varvel, Donald},
  title = {Pinwheel Scheduling with Two Distinct Numbers},
  journal = {Theoretical Computer Science},
  volume = {100},
  number = {1},
  pages = {105--135},
  year = {1992},
  doi = {10.1016/0304-3975(92)90365-M}
}

@article{impagliazzo2001strongly,
  author = {Impagliazzo, Russell and Paturi, Ramamohan and Zane, Francis},
  title = {Which Problems Have Strongly Exponential Complexity?},
  journal = {Journal of Computer and System Sciences},
  volume = {63},
  number = {4},
  pages = {512--530},
  year = {2001}
}

@inproceedings{ouyang2022training,
  author = {Ouyang, Long and Wu, Jeffrey and Jiang, Xu and Almeida, Diogo and Wainwright, Carroll and Mishkin, Pamela and Zhang, Chong and Agarwal, Sandhini and Slama, Katarina and Ray, Alex and Schulman, John and Hilton, Jacob and Kelton, Fraser and Miller, Luke and Simens, Maddie and Askell, Amanda and Welinder, Peter and Christiano, Paul F. and Leike, Jan and Lowe, Ryan},
  title = {Training Language Models to Follow Instructions with Human Feedback},
  booktitle = {Proceedings of the 36th International Conference on Neural Information Processing Systems (NeurIPS)},
  pages = {27730--27744},
  year = {2022}
}

@inproceedings{conitzer2024social,
  author = {Conitzer, Vincent and Freedman, Rachel and Heitzig, Jobst and Holliday, Wesley H. and Jacobs, Bob M. and Lambert, Nathan and Moss{\'e}, Milan and Pacuit, Eric and Russell, Stuart and Schoelkopf, Hailey and Tewolde, Emanuel and Zwicker, William S.},
  title = {Position: Social Choice Should Guide {AI} Alignment in Dealing with Diverse Human Feedback},
  booktitle = {Proceedings of the 41st International Conference on Machine Learning (ICML)},
  pages = {9346--9360},
  year = {2024}
}

@inproceedings{sorensen2024roadmap,
  author = {Sorensen, Taylor and Moore, Jared and Fisher, Jillian and Gordon, Mitchell L. and Mireshghallah, Niloofar and Rytting, Christopher Michael and Ye, Andre and Jiang, Liwei and Lu, Ximing and Dziri, Nouha and Althoff, Tim and Choi, Yejin},
  title = {Position: A Roadmap to Pluralistic Alignment},
  booktitle = {Proceedings of the 41st International Conference on Machine Learning (ICML)},
  pages = {46280--46302},
  year = {2024}
}

@inproceedings{elkind2024temporal,
  author = {Elkind, Edith and Obraztsova, Svetlana and Teh, Nicholas},
  title = {Temporal Fairness in Multiwinner Voting},
  booktitle = {Proceedings of the 38th AAAI Conference on Artificial Intelligence (AAAI)},
  pages = {22633--22640},
  year = {2024},
  doi = {10.1609/aaai.v38i20.30273}
}

@inproceedings{elkind2024temporalelections,
  author = {Elkind, Edith and Neoh, Tzeh Yuan and Teh, Nicholas},
  title = {Temporal Elections: Welfare, Strategyproofness, and Proportionality},
  booktitle = {Proceedings of the 27th European Conference on Artificial Intelligence (ECAI)},
  pages = {3292--3299},
  year = {2024},
  doi = {10.3233/FAIA240877}
}

@inproceedings{zech2024multiwinner,
  author = {Zech, Valentin and Boehmer, Niclas and Elkind, Edith and Teh, Nicholas},
  title = {Multiwinner Temporal Voting with Aversion to Change},
  booktitle = {Proceedings of the 27th European Conference on Artificial Intelligence (ECAI)},
  pages = {3236--3243},
  year = {2024},
  doi = {10.3233/FAIA240870}
}

@inproceedings{elkind2025nimby,
  author = {Elkind, Edith and Neoh, Tzeh Yuan and Teh, Nicholas},
  title = {Not in My Backyard! {T}emporal Voting over Public Chores},
  booktitle = {Proceedings of the 34th International Joint Conference on Artificial Intelligence (IJCAI)},
  pages = {3814--3820},
  year = {2025},
  doi = {10.24963/ijcai.2025/424}
}

@inproceedings{teh2026price,
  author = {Teh, Nicholas},
  title = {The Price of Proportional Representation in Temporal Voting},
  booktitle = {Proceedings of the 35th International Joint Conference on Artificial Intelligence (IJCAI)},
  year = {2026},
  note = {Extended version available at arXiv:2605.11157}
}

@inproceedings{lackner2020perpetual,
  author = {Lackner, Martin},
  title = {Perpetual Voting: Fairness in Long-Term Decision Making},
  booktitle = {Proceedings of the 34th AAAI Conference on Artificial Intelligence (AAAI)},
  pages = {2103--2110},
  year = {2020},
  doi = {10.1609/aaai.v34i02.5584}
}

@inproceedings{kozachinskiy2025optimal,
  author = {Kozachinskiy, Alexander and Shen, Alexander and Steifer, Tomasz},
  title = {Optimal Bounds for Dissatisfaction in Perpetual Voting},
  booktitle = {Proceedings of the 39th AAAI Conference on Artificial Intelligence (AAAI)},
  pages = {13977--13984},
  year = {2025},
  doi = {10.1609/aaai.v39i13.33529}
}

@inproceedings{elkind2022fairness,
  author = {Elkind, Edith and Kraiczy, Sonja and Teh, Nicholas},
  title = {Fairness in Temporal Slot Assignment},
  booktitle = {Proceedings of the 15th International Symposium on Algorithmic Game Theory (SAGT)},
  pages = {490--507},
  year = {2022},
  doi = {10.1007/978-3-031-15714-1_28}
}

@inproceedings{elkind2025temporal,
  author = {Elkind, Edith and Lam, Alexander and Latifian, Mohamad and Neoh, Tzeh Yuan and Teh, Nicholas},
  title = {Temporal Fair Division of Indivisible Items},
  booktitle = {Proceedings of the 24th International Conference on Autonomous Agents and Multiagent Systems (AAMAS)},
  pages = {676--685},
  year = {2025}
}

@article{choi2026temporal,
  author = {Choi, Kui-Wang and Li, Minming and Teh, Nicholas},
  title = {Temporal Fair Division of Indivisible Mixed Manna: Tractable Settings},
  journal = {arXiv preprint arXiv:2608.20033},
  year = {2026}
}

@article{teh2026FJRtemporal,
  author = {Nicholas Teh},
  title = {Full Justified Representation in Temporal Voting: Efficient Computation and Verification},
  journal = {SSRN preprint ssrn.7508758},
  year = {2026}
}

@article{aziz2017justified,
  author = {Aziz, Haris and Brill, Markus and Conitzer, Vincent and Elkind, Edith and Freeman, Rupert and Walsh, Toby},
  title = {Justified Representation in Approval-Based Committee Voting},
  journal = {Social Choice and Welfare},
  volume = {48},
  pages = {461--485},
  year = {2017}
}

@article{sanchez2017proportional,
  author = {S{\'a}nchez-Fern{\'a}ndez, Luis and Elkind, Edith and Lackner, Martin and Fern{\'a}ndez, Norberto and Fisteus, Jes{\'u}s A. and Basanta Val, Pablo and Skowron, Piotr},
  title = {Proportional Justified Representation},
  journal = {Artificial Intelligence},
  volume = {353},
  pages = {104503},
  year = {2026}
}

@inproceedings{peters2021proportional,
  author = {Peters, Dominik and Pierczy{\'n}ski, Grzegorz and Skowron, Piotr},
  title = {Proportional Participatory Budgeting with Additive Utilities},
  booktitle = {Proceedings of the 35th International Conference on Neural Information Processing Systems (NeurIPS)},
  pages = {12726--12737},
  year = {2021}
}

@book{lackner2023multi,
  author = {Lackner, Martin and Skowron, Piotr},
  title = {Multi-Winner Voting with Approval Preferences},
  series = {SpringerBriefs in Intelligent Systems},
  publisher = {Springer},
  year = {2023}
}

@article{teh2026strengthening,
  author = {Teh, Nicholas},
  title = {Strengthening Full Justified Representation: Efficient Verification and Computation},
  journal = {arXiv preprint arXiv:2608.11500},
  year = {2026}
}

\clearpage
\appendix

\section{Omitted Proofs from Section~\ref{sec:ranked-axioms}}
\label{app:ranked-axioms-proofs}

\subsection{Proof of Lemma~\ref{lem:all-rank-minimum-representation}}

Suppose that $S$ agrees in round $j$ under $R$ and that $w^j$ is Pareto efficient for $S$. If no voter in $S$ approves $w^j$, choose $c\in\bigcap_{i\in S}A_i^j(r_i^j)$. Then $c\succ_i^j w^j$ for every $i\in S$, contradicting Pareto efficiency. Thus, $S$ is represented in every such round, under every such $R$.

Conversely, we construct a rank matrix under which $S$ agrees in every round and is represented only in rounds in which $w^j$ is Pareto efficient for $S$. If $w^j$ is not Pareto efficient for $S$, choose $c^j\in C$ with $c^j\succ_i^j w^j$ for every $i\in S$, and set $r_i^j=\mathrm{rank}_i^j(c^j)$ for every $i\in S$. Then all voters in $S$ approve $c^j$, while none approves $w^j$. In every other round, set $r_i^j=m$ for every $i\in S$. The remaining entries of $R$ are arbitrary. Under this rank matrix, $S$ agrees in every round and, by the first part of the proof, is represented exactly in the rounds in which $w^j$ is Pareto efficient for $S$. This proves the lemma.

\subsection{Proof of Theorem~\ref{thm:all-rank-characterization}}

Put $\alpha=|S|/n$. Since $\mathrm{rep}_R(S)$ is an integer,
$\lfloor\alpha k_R(S)\rfloor-\mathrm{rep}_R(S)
=\lfloor\alpha k_R(S)-\mathrm{rep}_R(S)\rfloor$. We may therefore
maximize $\alpha k_R(S)-\mathrm{rep}_R(S)$ first and take the floor afterward.
The expression is a sum over rounds, and cutoffs in different rounds can
be chosen independently. The cutoffs of voters outside $S$ are arbitrary.
We consider three cases in a round $j$.

If $w^j$ is Pareto dominated for $S$, place each member's cutoff at a
candidate that every member prefers to $w^j$. The group agrees and is not
represented, giving a contribution of $\alpha$, the largest possible
contribution from one round.

If $w^j$ is a member's first choice, representation is unavoidable.
Placing every cutoff at $m$ also gives agreement, for a contribution of
$\alpha-1$. No other choice can give more.

Otherwise, $w^j$ is Pareto efficient for $S$ but no member ranks it first.
Cutoff one for every member excludes the winner and gives no agreement:
a common first choice would Pareto dominate $w^j$. This choice contributes
zero. Agreement without representation is impossible by Pareto efficiency,
and either outcome with representation contributes at most zero, since
$\alpha\leq1$.

The maximum sum is therefore
$\alpha(D_W(S)+F_W(S))-F_W(S)$, which proves the stated equality.
By definition, $S$ violates all-rank-PJR exactly when some rank matrix $R$
gives $\mathrm{rep}_R(S)<\lfloor k_R(S)|S|/n\rfloor$, that is, when the
maximum is at least one. This holds if and only if
$|S|(D_W(S)+F_W(S))\geq n(F_W(S)+1)$, which is equivalent to
$|S|D_W(S)-(n-|S|)F_W(S)\geq n$.

For JR, $S$ violates all-rank-JR exactly when some rank matrix $R$ gives
$\mathrm{rep}_R(S)=0$ and $k_R(S)|S|\geq n$. Zero representation is
possible only when $F_W(S)=0$. Under this
condition, exactly the $D_W(S)$ Pareto-dominated rounds can be made agreeing
and unrepresented; all other rounds can be made non-agreeing and
unrepresented by cutoff one. The largest agreement count consistent with
zero representation is thus $D_W(S)$, which proves the stated condition
for all-rank-JR.

\section{Omitted Proofs from Section~\ref{sec:attainability}}

\subsection{Proof of Theorem~\ref{thm:fixed-rank-gcr}}
\label{app:gcr-proof}
We prove Theorem~\ref{thm:fixed-rank-gcr}, which states that Algorithm~\ref{alg:fixed-rank-gcr} satisfies fixed-rank-PJR. The proof has two parts. Lemma~\ref{lem:fixed-rank-packing} and Corollary~\ref{cor:fixed-rank-stage-two} show that all rounds reserved in Stage~1 can be assigned in Stage~2. The proof of the theorem at the end of this section then shows that, for every cutoff $r$ and every group $S$, the assigned rounds represent $S$ under $A(r)$ in at least $\dem{S}{r}$ rounds. We state the following bound for every subset $V'$ of voters, not only for $V$, because the inductive step applies the bound to a proper subset of $V'$.

\begin{lemma}
\label{lem:fixed-rank-packing}
At every point during Stage~1, for every $V'\subseteq V$ and every
$k\in \set{1,\dots,|T|}$, the sum of $g(S',t)$ over all pairs $(S',t)$ with $t\in[m]$, $S'\subseteq V'$, and $|\tsr{S'}{t}|\leq k$ is at most $k|V'|/n$.
\end{lemma}

\begin{proof}
We use induction on the number of groups selected during Stage~1.
Before any group is selected, all the values of $g$ are zero, and
therefore the inequality holds.
Suppose that the inequality holds for every $V'\subseteq V$ and every $k\in\{1,\dots,|T|\}$ immediately before the algorithm
selects a group $S$ during the iteration for cutoff $r$, and fix $V'$ and $k$.
If $S\nsubseteq V'$ or $|\tsr{S}{r}| > k$, then the new value $g(S,r)$ does not appear in the sum. Since no other value of $g$ changes, the inequality continues to hold.
Thus, suppose that $S\subseteq V'$ and $|\tsr{S}{r}| \leq k$.
Recall that all values $g(S',t)$ are initialized to zero and that
$g(S',t)$ can become positive only when $S'$ is selected during the iteration for cutoff $t$. Hence, at the current point
of the execution, every pair $(S',t)$ that has not yet been selected has $g(S',t)=0$.
In particular, only pairs with $t \leq r$ contribute to the sum.
Now, partition the non-zero terms that appear in the sum into two classes: terms associated with groups that are disjoint from $S$ and terms associated with groups that intersect $S$. Let $D$ and $I$ be the sums of the terms from the first and second classes, respectively. Every group counted by $D$ is contained in $V' \setminus S$. Therefore, by the induction hypothesis applied to $V' \setminus S$, $D \leq \frac{k|V' \setminus S|}{n}$.

Let $z = \sum_{t=1}^{r-1}\sum_{S'\subseteq V, S'\cap S\neq\varnothing} g(S',t)$. Since $S$ has positive remaining demand when it is selected, the algorithm sets $g(S,r) = \drem{S}{r} = \dem{S}{r}-z$.
Note that every group selected before the selection of $S$ during the current iteration for cutoff $r$ is disjoint from $S$. Indeed, when such a group was selected, the algorithm removed from $\mathcal{V}$ every group intersecting it. Since $S$ is still contained in $\mathcal{V}$, it cannot intersect any group selected earlier during the current iteration. 
Consequently, immediately before $S$ is selected, every pair $(S',t)$ with $g(S',t)>0$ and $S'\cap S\neq\varnothing$ satisfies $t < r$. That is, every non-zero term
counted by $I$ is associated with a cutoff $t < r$. Hence, every term
counted by $I$ is also counted by $z$, and thus $I\leq z$.
Therefore, 
$I+g(S,r) \leq z+\dem{S}{r}-z = \dem{S}{r}$. 
Moreover,
$\dem{S}{r} = \lfloor |\tsr{S}{r}| \cdot \frac{|S|}{n} \rfloor \leq |\tsr{S}{r}| \cdot \frac{|S|}{n} \leq \frac{k|S|}{n}$.
Overall, after selecting $S$ and setting $g(S,r)$, the sum is $D + I + g(S,r) \leq \frac{k|V' \setminus S|}{n} + \frac{k|S|}{n} = \frac{k|V'|}{n}$.
This completes the induction.
\end{proof}

\begin{corollary}
\label{cor:fixed-rank-stage-two}
Whenever a pair $(S,r)$ with $g(S,r)>0$ is processed in Stage~2,
there are at least $g(S,r)$ unassigned rounds in which the members
of $S$ approve a common candidate under $A(r)$.
\end{corollary}

\begin{proof}
Let $(S_1,r_1),\dots,(S_q,r_q)$ be the pairs with positive $g$-value, ordered as in Stage~2, so that $|\tsr{S_1}{r_1}| \leq \cdots \leq |\tsr{S_q}{r_q}|$. 
Suppose, towards a contradiction, that the corollary is false, and
let $(S_h,r_h)$ be the first pair that cannot be assigned its required number of rounds. Write $k=|\tsr{S_h}{r_h}|$.
That is, the members of $S_h$ approve a common candidate under
$A(r_h)$ in exactly $k$ rounds.

If fewer than $g(S_h,r_h)$ of these rounds are still unassigned, then at least $k-g(S_h,r_h)+1$ of them have already been assigned to earlier pairs. Since distinct pairs are assigned disjoint sets of rounds, it follows that $\sum_{a=1}^{h-1}g(S_a,r_a) \geq k-g(S_h,r_h)+1$.
Consequently,
$\sum_{a=1}^{h}g(S_a,r_a)>k$.
On the other hand, by the ordering of the pairs,
$|\tsr{S_a}{r_a}|\leq k$ for every $a\leq h$.
Applying Lemma~\ref{lem:fixed-rank-packing} after Stage~1 with
$V'=V$ gives $\sum_{t\in[m],\ S\subseteq V, |\tsr{S}{t}|\leq k}
g(S,t) \leq k$.
In particular,
$\sum_{a=1}^{h} g(S_a,r_a)\leq k$, which is a contradiction.
Hence, every pair can be assigned all of its reserved rounds.
\end{proof}

\begin{proof}[Proof of Theorem~\ref{thm:fixed-rank-gcr}]
Fix a cutoff $r\in[m]$ and a group
$S\subseteq V$. We first show that
\begin{equation}
\label{eq:coverage-in-theorem}
\sum_{t=1}^{r}
\sum_{S'\subseteq V, S'\cap S\neq\varnothing}
g(S',t) \geq \dem{S}{r}.
\end{equation}
If $\dem{S}{r}=0$, the inequality follows from $g\geq 0$. Otherwise,
$\tsr{S}{r}\neq\varnothing$, so $S$ belongs to $\mathcal V$ at the beginning
of the iteration for cutoff $r$.
Suppose first that $S$ is removed from $\mathcal{V}$ when the algorithm selects a group $\widehat{S}$. Immediately before $\widehat S$ is
selected, let $z = \sum_{t=1}^{r-1} \sum_{S'\subseteq V, S'\cap S\neq\varnothing} g(S',t)$.
If $z\geq \dem{S}{r}$, then the required inequality already holds.
Otherwise, $\drem{S}{r} = \dem{S}{r}-z > 0$.
Since $\widehat S$ is selected to maximize the remaining demand among all groups currently in $\mathcal V$, 
$g(\widehat S,r) = \drem{\widehat S}{r} \geq \drem{S}{r} = \dem{S}{r}-z$.
Moreover, $\widehat S\cap S\neq\varnothing$,
since this is why $S$ is removed from $\mathcal V$. Therefore,
$g(\widehat S,r)$ is included in the sum on the left-hand side of
\eqref{eq:coverage-in-theorem}. Immediately after
$\widehat S$ is selected, this sum is exactly
$z+g(\widehat S,r) \geq z+\dem{S}{r}-z = \dem{S}{r}$.
Since all $g$ values are nonnegative, the
left-hand side of \eqref{eq:coverage-in-theorem} cannot decrease during the remainder of the iteration for cutoff $r$, and thus \eqref{eq:coverage-in-theorem} holds at the end of this iteration.

Suppose now that $S$ is never removed from $\mathcal V$. It then
remains in $\mathcal V$ until the while-loop terminates. The loop
terminates only when every group remaining in $\mathcal V$ has zero
remaining demand. In particular, $\drem{S}{r}=0$.
By the definition of remaining demand, this implies
$\sum_{t=1}^{r} \sum_{S'\subseteq V, S'\cap S\neq\varnothing}
g(S',t) \geq \dem{S}{r}$, as required.

Now, consider a pair $(S',t)$ contributing a positive value to the sum in \eqref{eq:coverage-in-theorem}. Thus, $t \leq r$ and $S'\cap S\neq\varnothing$.
Stage~2 assigns $g(S',t)$ distinct rounds to this pair. In each such round $j$, the algorithm chooses $w^j\in\bigcap_{i\in S'}A_i^j(t)$.
Choose a voter $i\in S'\cap S$. Then $w^j\in A_i^j(t)$.
Since $t\leq r$, the fixed-rank approval sets are nested: $A_i^j(t)\subseteq A_i^j(r)$. It follows that $w^j\in A_i^j(r)
\subseteq \bigcup_{i'\in S}A_{i'}^j(r)$.
Thus, every round assigned to $(S',t)$ is a round in which the winner is approved under $A(r)$ by at least one member of $S$.
The sets of rounds assigned to distinct pairs are disjoint. Therefore,
\eqref{eq:coverage-in-theorem} implies that there are at least
$\dem{S}{r}$ distinct rounds $j\in T$ such that $w^j\in\bigcup_{i\in S}A_i^j(r)$. That is, $S$ is represented in at least $\dem{S}{r}$ rounds. 
Since $r$ and $S$ were arbitrary, $W$ satisfies fixed-rank-PJR.

\end{proof}

\section{A Stronger Size Condition with Three First Choices}
\label{app:three-choice-quota}
The three-first-choice guarantee permits a smaller size requirement than all-rank-PJR. The improvement comes from giving voters slightly more initial budget while keeping the total unspent budget at most one.

\begin{theorem}
\label{thm:three-choice-quota}
Suppose that at most three candidates are ranked first in each round. If Algorithm~\ref{alg:first-choice-budgets} starts each voter with $(|T|+1/4)/n$ units, then for every group $S$,
\[
 \left|\set{j\in T:w^j\in\bigcup_{i\in S}A_i^j(1)}\right|
 \geq\left\lceil\frac{(|T|+1/4)|S|}{n}\right\rceil-1.
\]
Thus, a group with more than $\ell n/(|T|+1/4)$ members has a member's first choice selected in at least $\ell$ rounds. No constant greater than $1/4$ can replace $1/4$ for all numbers of voters and all horizons, even for an offline rule and only for groups that agree in every round under some rank matrix.
\end{theorem}

\begin{proof}
We apply the bound from the proof of Theorem~\ref{thm:first-choice-loss} on the total remaining budget, with $d=3$ and initial total $|T|+1/4$. For $|T|=1$, the final total is at most $(2/3)(5/4)=5/6$. For $|T|\geq2$, the upper-bound recurrence subtracts one unit in each of its first $|T|-2$ steps, leaving $9/4$, and multiplies the total by $2/3$ in each of its last two steps, leaving one. Every group pays at least $(|T|+1/4)|S|/n-1$, only in rounds where a member's first choice wins and at most one unit per round. Taking the ceiling proves the bound. Since $S$ is nonempty, this bound is at least $\lfloor |T||S|/n\rfloor$.

For optimality, use the construction in Theorem~\ref{thm:first-choice-loss} with $d=3$, $|T|=2$, and $n=9$. For every sequence, this construction gives a group of four voters that agrees in every round under some rank matrix and is never represented. Replacing $1/4$ by any constant $\beta>1/4$ would require this group to receive at least one represented round, because $(2+\beta)4/9>1$. This is impossible.
\end{proof}

\end{document}